\documentclass[11pt, a4paper, oneside]{Thesis} 

\graphicspath{{./img/}} 

\usepackage[square, numbers, comma, sort&compress]{natbib} 
\hypersetup{urlcolor=blue, colorlinks=true} 
\title{\ttitle} 

\usepackage{graphicx} 
\usepackage{mathrsfs}
\usepackage{amsmath} 
\usepackage{amssymb} 
\usepackage{amsthm} 
\usepackage{tikz}
\usetikzlibrary{matrix,arrows}

\theoremstyle{definition}
\newtheorem{notation}[theorem]{Notation}

\def\({\left(}
\def\){\right)}

\newcommand{\image}[3]{\begin{figure*}[ht]
\begin{center}
\includegraphics[width=#2\textwidth]{img/#1}
\caption{\small{\label{#1}#3}}
\end{center}
\end{figure*}}

\newcommand{\imagerot}[4]{\begin{figure*}[ht]
\begin{center}
\includegraphics[width=#2\textwidth,angle=#4]{img/#1}
\caption{\small{\label{#1}#3}}
\end{center}
\end{figure*}}

\newcommand{\R}{\mathbb{R}}
\newcommand{\N}{\mathbb{N}}
\newcommand{\CC}{\mathbb{C}}

\newcommand{\de}{\textnormal{d}}

\newcommand{\grad}{\textnormal{grad }}

\newcommand{\vol}{\de_{vol}}

\newcommand{\tn}{\textnormal}
\newcommand{\ds}{\displaystyle}

\newcommand{\ie}{\textit{i.e.}}
\newcommand{\cf}{\textit{cf.}}
\newcommand{\vs}{\textit{vs.}}
\newcommand{\eg}{\textit{e.g.}}
\newcommand{\etc}{\textit{etc.}}

\newcommand{\citepcf}[2]{({\cf} \citep{#1}{#2})}
\newcommand{\cfeg}[2]{({\cf} {\eg} \citep{#1}{#2})}

\newcommand{\rank}{\textnormal{rank }}
\newcommand{\diag}{\textnormal{diag}}

\newcommand{\mf}[1]{\mathfrak{#1}}
\newcommand{\mc}[1]{\mathcal{#1}}
\newcommand{\ms}[1]{\mathscr{#1}}

\newcommand{\IM}{\tn{im }}
\newcommand{\tensors}[3]{\mc T{}^{#1}_{#2}#3}

\newcommand{\sseuclid}[3]{\R^{#1,#2,#3}}

\newcommand{\supp}[1]{\tn{supp}(#1)}

\newcommand{\idxannih}[2]{#1{}^{#2}{}}
\newcommand{\idxcoannih}[2]{#1{}_{#2}{}}

\newcommand{\radix}[1]{\idxcoannih{#1}{\circ}}
\newcommand{\annih}[1]{\idxannih{#1}{\bullet}}
\newcommand{\coannih}[1]{\idxcoannih{#1}{\bullet}}
\newcommand{\coradix}[1]{\idxannih{#1}{\circ}}

\newcommand{\annihg}{\coannih{g}}
\newcommand{\coannihg}{\annih{g}}
\newcommand{\metric}[1]{\langle#1\rangle}
\newcommand{\annihprod}[1]{\coannih{\langle\!\langle#1\rangle\!\rangle}}

\newcommand{\cocontr}{{{}_\bullet}}

\newcommand{\vectmodule}{\mf X}
\newcommand{\fivect}[1]{\vectmodule(#1)}
\newcommand{\fivectnull}[1]{\vectmodule_\circ(#1)}
\newcommand{\fiscal}[1]{\ms F(#1)}

\newcommand{\fiformk}[2]{\mc A^{#1}(#2)}

\newcommand{\fivectlift}[1]{\mf L(#1)}
\newcommand{\annihforms}[1]{\annih{\mc A}(#1)}
\newcommand{\annihformsk}[2]{\annih{\mc A}{}^{#1}(#2)}

\newcommand{\metricformsk}[2]{\annih{\ms A}{}^{#1}(#2)}
\newcommand{\srformsk}[2]{\annih{\ms A}{}^{#1}(#2)}

\newcommand{\ric}{\tn{Ric}}

\newcommand{\rhord}[1]{\mc O_\rho\(#1\)}

\newcommand{\lie}{\mc L}
\newcommand{\kosz}{\mc K}
\newcommand{\der}{\nabla}
\newcommand{\dera}[1]{\der_{#1}}
\newcommand{\derb}[2]{\dera{#1}{#2}}
\newcommand{\derc}[3]{({\derb{#1}{#2}})(#3)}
\newcommand{\lder}{\der^{\flat}}
\newcommand{\ldera}[1]{\lder_{#1}}
\newcommand{\lderb}[2]{\ldera{#1}{#2}}
\newcommand{\lderc}[3]{(\lderb{#1}{#2})(#3)}
\newcommand{\curv}[2]{\mc R^\flat_{#1#2}}

\newcommand{\Ric}{\textnormal{Ric}}

\newcommand{\cyclic}{\circlearrowleft}

\newcommand{\abs}[1]{\left|#1\right|}
\newcommand{\dsfrac}[2]{\ds{\frac{#1}{#2}}}

\newcommand{\metricname}{metric}
\newcommand{\koszulname}{Koszul object}

\newcommand{\CS}{\mathbb{S}}
\newcommand{\CT}{\mathbb{T}}
\newcommand{\CM}{\mathbb{M}}

\def\hyph{-\penalty0\hskip0pt\relax}
\newcommand{\semiriem}{semi{\hyph}Riemannian}
\newcommand{\ssemiriem}{Semi{\hyph}Riemannian}
\newcommand{\semieucl}{semi{\hyph}Euclidean}

\newcommand{\semireg}{semi{\hyph}regular}
\newcommand{\ssemireg}{Semi{\hyph}regular}
\newcommand{\quasireg}{quasi{\hyph}regular}
\newcommand{\qquasireg}{Quasi{\hyph}regular}
\newcommand{\nondeg}{non{\hyph}degenerate}

\newcommand{\nonsing}{non{\hyph}singular}

\newcommand{\rstationary}{radical{\hyph}stationary}
\newcommand{\rrstationary}{Radical{\hyph}stationary}
\newcommand{\rannih}{radical{\hyph}annihilator}

\newcommand{\nonrenormalizable}{non{\hyph}renormalizable}
\newcommand{\nonrenormalizability}{non{\hyph}renormalizability}

\newcommand{\flrw}{Friedmann-Lema\^itre-Robertson-Walker}
\newcommand{\FLRW}{FLRW}
\newcommand{\schw}{Schwarzschild}
\newcommand{\rn}{Reissner-Nordstr\"om}
\newcommand{\kn}{Kerr-Newman}
\newcommand{\Wch}{Weyl curvature hypothesis}
\newcommand{\WCH}{WCH}
\newcommand{\hor}{Ho{\v{r}}ava}
\newcommand{\HL}{\hor-Lifschitz}

\newcommand{\chap}[2]{\chapter{#2}\lhead{\emph{#2}}\label{#1}}

\newcommand{\movedfrom}[1]{}
\newcommand{\movedto}[1]{}

\begin{document}
\makeatletter
\let\citep\@undefined
\newcommand{\citep}[2]{\cite{#1}, p. #2}
\let\sref\@undefined
\newcommand{\sref}[1]{\S\ref{#1}}
\makeatother

\frontmatter 

\setstretch{1.15} 

\fancyhead{} 
\rhead{\thepage} 
\lhead{} 

\pagestyle{fancy} 

\newcommand{\HRule}{\rule{\linewidth}{0.5mm}} 

\hypersetup{pdftitle={\ttitle}}
\hypersetup{pdfsubject=\subjectname}
\hypersetup{pdfauthor=\authornames}
\hypersetup{pdfkeywords=\keywordnames}


\addtotoc{Title page}
\begin{titlepage}
\begin{center}

\textsc{\LARGE \univname}
\\[0.1cm] 
\textsc{\LARGE \deptname}
\\[7.0cm] 


\textsc{\Huge \ttitle}\\[0.5cm] 

{\Large Ph.D. Thesis} 
\\[0.5cm]

\textsc{\LARGE \authornames}
\\[6.5cm]

{\large Supervisor} \\[0.1cm]
{\Large Prof. dr. CONSTANTIN UDRI{\c S}TE} 
\\{\large Corresponding Member of the Accademia Peloritana dei Pericolanti, Messina}
\\{\large Full Member of the Academy of Romanian Scientists}
\\[3.0cm]

{\large Bucharest 2013}
 
\vfill
\end{center}

\end{titlepage}


\Declaration{

\addtocontents{toc}{\vspace{1em}} 

I, \authornames, declare that this Thesis titled, `\ttitle' and the work presented in it are my own. I confirm that:

\begin{itemize} 
\item[\tiny{$\blacksquare$}] Where I have consulted the published work of others, this is always clearly attributed.
\item[\tiny{$\blacksquare$}] Where I have quoted from the work of others, the source is always given. With the exception of such quotations, this Thesis is entirely my own work.
\item[\tiny{$\blacksquare$}] I have acknowledged all main sources of help.
\end{itemize}
 
Signed:\\
\rule[1em]{25em}{0.5pt} 
 
Date:\\
\rule[1em]{25em}{0.5pt} 
}

\clearpage 








\renewenvironment{abstract}
{
  \btypeout{Abstract Page}
  \thispagestyle{empty}
  \null\vfil
  \begin{center}
    \setlength{\parskip}{0pt}
    {\huge{\textit{Abstract}} \par}
    \bigskip
    \bigskip
  \end{center}
}


\addtotoc{Abstract} 

\abstract{\addtocontents{toc}{\vspace{1em}} 

This work presents the foundations of Singular Semi-Riemannian Geometry  and Singular General Relativity, based on the author's research. An extension of differential geometry and of Einstein's equation to singularities is reported. Singularities of the form studied here allow a smooth extension of the Einstein field equations, including matter. This applies to the Big-Bang singularity of the FLRW solution. It applies to stationary black holes, in appropriate coordinates (since the standard coordinates are singular at singularity, hiding the smoothness of the metric). In these coordinates, charged black holes have the electromagnetic potential regular everywhere. Implications on Penrose's Weyl curvature hypothesis are presented. In addition, these singularities exhibit a (geo)metric dimensional reduction, which might act as a regulator for the quantum fields, including for quantum gravity, in the UV regime. This opens the perspective of perturbative renormalizability of quantum gravity without modifying General Relativity.
}

\clearpage 


\setstretch{1.15} 

\acknowledgements{\addtocontents{toc}{\vspace{1em}} 

I thank my PhD advisor, Prof. dr. Constantin Udri\c{s}te from the \textit{\univname}, who accepted and supervised me, with a theme with enough obstacles and subtleties.
I thank my former Professors Gabriel Pripoae and Olivian Simionescu-Panait from the \textit{University of Bucharest}, for comments and important suggestions. I offer my profound gratitude to the late Professor Kostake Teleman \cite{pripoae2009teleman}, for mentoring me during my Master at the \textit{University of Bucharest}.
I wish to thank Professor Vasile Br{\^i}nz{\u a}nescu, my former PhD advisor during the preparation years of my PhD degree program at the {\em Institute of Mathematics of the Romanian Academy}, for preparing high standard exams containing topics suited for my purposes, and for giving me the needed freedom, when my research evolved in a different direction.
I thank Professor M. Vi\c sinescu from the \textit{National Institute for Physics and Nuclear Engineering Magurele, Bucharest, Romania}, for his unconditioned help and suggestions.

I thank Professor David Ritz Finkelstein from \textit{Georgia Tech Institute of Technology}, whose work was such an inspiration in finding the coordinates which remove the infinities in the black hole solutions, and for his warm encouragements and moral support.

I thank Professors P. Fiziev and D. V. Shirkov for helpful discussions and advice received during my stay at the {\em Bogoliubov Laboratory of Theoretical Physics, JINR, Dubna}.

I thank Professor A. Ashtekar from \textit{Penn State University}, for an inspiring conversation about some of the ideas from this Thesis, and possible relations with {\em Loop Quantum Gravity}, during private conversations in Turin.

I thank the anonymous referees for the valuable comments and suggestions to improve the clarity and the quality of the published papers on which this Thesis is based.

This work was partially supported by the {\em Romanian Government grant PN II Idei 1187}.
}
\clearpage 


\pagestyle{fancy} 

\addtotoc{Contents}
\setcounter{tocdepth}{3}
\lhead{\emph{Contents}} 
\tableofcontents 

\mainmatter 

\pagestyle{fancy} 













\chap{ch_intro}{Introduction}

\section{Historical background}

We are interested in the properties of a class of smooth differentiable manifolds which have on the tangent bundle a smooth symmetric bilinear form (also named \textit{metric}), which is allowed to change its signature.

The first such manifolds which were studied have {\nondeg} metric -- starting with the Euclidean plane and space, and continuing with the non-Euclidean geometries introduced by Lobachevsky, Gauss, and Bolyai. After Gauss extended the study of the Euclidean plane to curved surfaces, Bernhard Riemann generalized it to curved spaces with arbitrary number of dimensions \cite{riemann1873hypotheses}. Riemann hoped to give a geometric description of the physical space, in the idea that matter is in fact the effect of the curvature. The previously discovered geometries -- the Euclidean and non-Euclidean ones, and Gauss's geometry of surfaces -- are all particular cases of \textit{Riemannian geometry}. A Riemannian manifold is a differentiable manifold endowed with a symmetric, non-degenerate and positive definite bilinear form on its tangent bundle.

The necessity of studying spaces having a symmetric, non-degenerate bilinear form which is not positive definite appeared with the Theory of Relativity \cite{einstein1920relativity}. A differentiable manifold having on its tangent bundle a symmetric, non-degenerate bilinear form which is not necessarily positive or negative definite is named \textit{{\semiriem} manifold} (sometimes \textit{pseudo-Riemannian manifold}, and in older textbooks even is called Riemannian manifold). {\ssemiriem} geometry constitutes the mathematical foundation of General Relativity. It was thoroughly studied, and the constructions made starting from the non-degenerate metric, such as the Levi-Civita connection, the covariant derivative, the Riemann, Ricci and scalar curvatures are very similar to the Riemannian case, when the metric is positive definite. On the other hand, other properties, especially the global ones, are very different in the indefinite case. Very good references for {\semiriem} geometry are the textbooks \cite{ONe83,BE96,HE95}.

If we allow the metric to be degenerate, many difficulties occur. For this reason, advances were made slower than for the non-degenerate case, and only in particular situations. The study of manifolds endowed with degenerate metric is pioneered by Moisil \cite{Moi40}, Strubecker \cite{Str41,Str42a,Str42b,Str45}, Vr\u{a}nceanu \cite{Vra42}
\footnote{Gheorghe Vr\u{a}nceanu also introduced in 1926 the {\em non-holonomic geometry}, nowadays known as {\em sub-Riemannian geometry} \cite{Vra26,Vra36,teleman2004vranceanu,pripoae2004vranceanu}. Although at a point, there is a duality between the sub-Riemannian and the singular {\semiriem} metrics, locally and especially globally there are too many differences to allow an import of the rich amount of results already obtained in sub-Riemannian geometry.}.

One situation when the metric can be degenerate occurs in the study of submanifolds of {\semiriem} manifolds. In the Riemannian case, the submanifolds are Riemannian too. But in general the image of a smooth mapping from a differentiable manifold to a Riemannian or {\semiriem} manifold may be singular, in particular may have degenerate metric. In the case of \textit{varieties} the problem of finding a \textit{resolution of its singularities} was proved to have positive answer by Hironaka \cite{hironaka1964resolution1,hironaka1964resolution2}. In the {\semiriem} case, the metric induced on a submanifold can be degenerate, even though the larger manifold has non-degenerate metric. The properties of such submanifolds, studied in many articles, {\eg} in \cite{Ros72}, \cite{Kup87a,Kup87c}, \cite{Bej95,Bej96}, were extended by Kupeli to manifolds endowed with degenerate metric of constant signature \cite{Kup87b,Kup96}. The situation is much more difficult when the signature changes.

There are some situations in General Relativity when the metric becomes degenerate or changes its signature. There are cosmological models of the Universe in which the initial singularity of the Big Bang is replaced, by making the metric of the early Universe Riemannian. Such models, constructed in connection to the Hartle-Hawking no-boundary approach to Quantum Cosmology, assume that the metric was Riemannian, and it changed, becoming Lorentzian, so that time emerged from a space dimension. Such a change of signature is considered to take place when traversing a hypersurface, on which the metric becomes degenerate \cite{Sak84},\cite{Ellis92a,Ellis92b},\cite{Hay92,Hay93,Hay95}, \cite{Der93}, \cite{Dray91,Dray93,Dray94,Dray95,Dray96,Dray01}, \cite{Koss85,Koss87,Koss93a,Koss93b,Koss94a,Koss94b}.

Another situation where the metric can become degenerate was proposed by Einstein and Rosen, as a model of charged particles. They were the first to model charged particles as \textit{wormholes}, also named \textit{Einstein-Rosen bridges} \cite{ER35}, and inspired Wheeler's \textit{charge without charge} program \cite{MW57,whe:1962}.

The Einstein's equation, as well as its Hamiltonian formulation due to Arnowitt, Deser and Misner \cite{ADM62}, may lead to cases when the metric is degenerate. As the Penrose and Hawking \textit{singularity theorems} show, the conditions leading to singularities are very general, applying to the matter distribution in our Universe \cite{Pen65,Haw66i,Haw66ii,Haw67iii,HP70,HE95}. Therefore, it is important to know how we can deal with such singularities. Many attempts were done to solve this issue.

For example it was suggested that Ashtekar's method of ``new variables'' \cite{ASH87,ASH91,Rom93a} can be used to pass beyond the singularities, because the variable $\widetilde E^a_i$ -- a densitized frame of vector fields -- defines the metric, which can be degenerate. Unfortunately, it turned out that in this case the connection variable $A_a^i$ may become singular \cf \eg \cite{Yon97}. In fact, this is the general case, because the connection variable $A_a^i$ contains as a term the Levi-Civita connection, which is singular in most cases when $g$ becomes degenerate.

Quantum effects are suspected to play an important role in \textit{avoiding the singularities}. Loop Quantum Cosmology, by quantizing spacetime, provided a mean to avoid the Big-Bang singularity and replace it with a Big-Bounce, due to the fact that the curvature is bounded, because there is a minimum distance \cite{bojowald2003absenceLQC,ashtekar2011LQC,visinescu2009bianchi,visinescu2012bianchi}. A very interesting model which is nonsingular at the end of the evolution and does not allow the anisotropic universe to turn into an isotropic one is studied in \cite{visinescu2006bel,visinescu2010bianchi}. Another possibility to avoid singularities is given in \cite{poplawski2012nonsingular}, within the Einstein-Cartan-Sciama-Kibble theory \cite{cartan1922torsion,cartan1925varietes,hehl1976cartan}, and also in $f(R)$ modifications of General Relativity. Singularity removal may arise arise in more classical contexts. For example, certain symmetries eliminates the so-called NUT singularity, in certain conditions \cite{misner1963flatter,cotuaescu2001dirac,cotaescu2001dynamical,cotaescu2004gravitational}.

A more classical proposal to avoid the consequences of singularities was initiated by R. Penrose, with the \textit{cosmic censorship hypothesis} \cite{Pen69,Pen74,Pen79,Pen98}. According to the \textit{weak cosmic censorship hypothesis}, all singularities (except the Big-Bang singularity) are hidden behind an \textit{event horizon}, hence are not \textit{naked}. The \textit{strong cosmic censorship hypothesis} conjectures that the maximal extension of spacetime as a regular Lorentzian manifold is \textit{globally hyperbolic}.

The literature in the approaches to singularities is too vast, and it would be unjust to claim to review it properly in a research work which is not a dedicated review. A great review on the problem of singularities in General Relativity is given in \cite{tipler1980singularities_review}, and a more up-to-date one in \cite{garcia2005singularities_review}.

\section{Motivation for this research}

In this Thesis is developed the mathematical formalism for a large class of manifolds having on the tangent structure symmetric bilinear forms, which are allowed to become degenerate. Then these results are applied to the singularities in General Realtivity.

This special type of singular {\semiriem} manifold has regular properties in what concerns
\begin{enumerate}
	\item 
the Riemann curvature $R_{abcd}$ (and not $R^a{}_{bcd}$, which in general diverges at singularities),
	\item 
the covariant derivative of an important class of differential forms,
	\item 
and other geometric objects and differential operators which in general cannot be defined properly because they require the inverse of the metric.
\end{enumerate}
These properties of regularity are not valid for any type of degenerate metric. This justifies the name of \textit{semi-regular {\semiriem} manifolds} given here to these special singular {\semiriem} manifolds. Semi-regular metrics can also be used to give a densitized version to the Einstein equation, and to approach the problem of singularities in General Relativity.

The signature of the metric of a semi-regular {\semiriem} manifold can change, but when it doesn't change, we obtain the ``stationary singular {\semiriem} manifolds'' with constant signature, researched by Kupeli \cite{Kup87b,Kup96}. These in turn contain as particular case the {\semiriem} manifolds, which contain the Riemannian manifolds.

In General Relativity, Einstein's equation encodes the relation between the stress-energy tensor of matter, and the Ricci curvature. 
In 1965 Roger Penrose \cite{Pen65}, and later he and S. Hawking \cite{Haw66i,Haw66ii,Haw67iii,HP70,HE95}, proved a set of \textit{singularity theorems}. These theorems state that, under reasonable conditions, the spacetime turns out to be \textit{geodesic incomplete} -- {\ie} it has \textit{singularities}. 
They show that the conditions of occurrence of singularities are quite common. 
Christodoulou \cite{Chr09} showed that these conditions are in fact more common, and then Klainerman and Rodnianski \cite{KR09} proved that they are even more common.

Consequently, some researchers proclaimed that General Relativity predicts its own breakdown, by predicting the singularities \cite{HP70,Haw76,ASH91,HP96,Ash08,Ash09}. Hawking's discovery of the black hole evaporation, leading to his \textit{information loss paradox} \cite{Haw75,Haw76}, made the things even worse. The singularities seem to destroy information, in particular violating the unitary evolution of quantum systems. 
The reason is that the field equations cannot be continued through singularities.

Therefore, it would be important to better understand the singularities.

There are two main situations in which singularities appear in General Relativity and Cosmology: at the Big-Bang, and in the black holes. We will see that the mathematical apparatus developed in this Thesis finds applications in both these situations.

As the previous section shows, much work is done on singular metrics. But to the author's knowledge, the systematic approach presented here and the results are novel, as well as the applications to General Relativity, and have no overlap with the research that is previously done by other researchers.

\section{Presentation per chapter}

Chapter \ref{ch_ssr} introduces  the geometry of metrics which can be degenerate, and are allowed to change signature. It studies the main properties of such manifolds, and is based almost entirely on the author's research. It contains an invariant definition of metric contraction between covariant indices, which works also when the metric is degenerate. With the help of the metric, it is shown that in some cases one can construct covariant derivatives, and even define a Riemann curvature tensor. If the metric is {\nondeg}, all these geometric objects become the ones known from {\semiriem} geometry. Section \ref{s_cartan} contains the derivation of a generalization of Cartan's structural equations for degenerate metric. Section \ref{s_warped} generalizes the notion of warped product to the case when the metric can become degenerate. It provides a simple way to construct examples of singular {\semiriem} manifolds. The degenerate warped product turns out to be relevant in some of the singularities which will be analyzed in the Thesis.

The following chapters, also based on author's own research, applies the mathematics developed in the first part to the singularities encountered in General Relativity. Chapter \ref{ch_einstein_equation} introduces two equations equivalent to Einstein's equation at the points where the metric is {\nondeg}, but which remain smooth at singularities. The first equation remains smooth at the so-called {\semireg} singularities. The second of the equations applies to the more restricted case of {\quasireg} singularities, which will turn out to be important in the following chapters.

Chapter \ref{ch_big_bang} shows, with the apparatus developed so far, that the {\flrw} spacetime is {\semireg}, but also {\quasireg}. It also studies some important properties of the {\FLRW} singularity. The {\nonsing} version of Einstein's equation is written explicitly, and shown to be smooth at singularity. Then, in section \sref{s_wch}, a more general solution, which is not homogeneous or isotropic, is presented. It is shown that it contains as particular cases the isotropic singularities and the {\flrw} singularities, and, being {\quasireg}, satisfies the Weyl curvature hypothesis of Penrose.

The black hole singularities are studied in chapter \ref{ch_black_hole}. It is shown that the Schwarzschild singularity is semi-regularizable, by using a method inspired by that of Eddington and Finkelstein \cite{eddington1924comparison,finkelstein1958past}, used to prove that the metric is regular on the event horizon. This method is also applied to make the Reissner-Nordstrom and the Kerr-Newman singularities analytic. Then, it is shown that this approach allows the construction of globally hyperbolic spacetimes with singularities (Chapter \ref{ch_global_hyperbolic}). This shows that it doesn't follow with necessity that black holes destroy information, leading by this to Hawking's information loss paradox and violations of unitary evolution. This suggests that a possible resolution of Hawking's paradox can be obtained without modifying General Relativity.

Chapter \ref{ch_quantum_gravity} explores the possibility that the dimensional reduction at singularities improves the renormalization in Quantum Field Theory, and makes Quantum Gravity preturbatively renormalizable.


\chap{ch_ssr}{Singular semi-Riemannian manifolds}

The text in this chapter is based on author's original results, communicated in the papers \cite{Sto11a}, \cite{Sto11d}, and \cite{Sto11b}.

\section{Introduction}

\subsection{Motivation and related advances}

On a {\semiriem} manifold (including the Riemannian case), one can use the metric to define geometric objects like the Levi-Civita connection and the Riemann curvature. One essential ingredient is the metric covariant contraction \cite{Balan1999diffgeom,udriste2005linear}. In singular {\semiriem} geometry, the metric is not necessarily {\nondeg}, and these standard constructions no longer work, being based on the inverse of the metric, and on operations defined with its help, like the metric contraction between covariant indices.

In this chapter, we define in an invariant way the canonical metric contraction between covariant indices, which works even for degenerate metrics (although in this case is well defined only on a special type of tensor fields, but this suffices for our purposes). Then, we use this contraction to define in a canonical manner the covariant derivative for {\rannih} indices of covariant tensor fields, on a class of singular {\semiriem} manifolds. This newly defined covariant derivative helps us construct the Riemann curvature, which is smooth on a class of singular {\semiriem} manifolds, named {\semireg}. We also define the Ricci and scalar curvatures, which are smooth for degenerate metric, so long as the signature is constant.

\subsection{Presentation of this chapter}

Section \sref{s_singular_semi_riemannian} recalls known generalities about singular {\semiriem} manifolds, such as the space of degenerate tangent vectors at a point.

Section \sref{s_dual_inner_prod} continues with the properties of the {\rannih} space, consisting in the covectors annihilating the degenerate vectors at a point. These spaces of covectors are used to define tensor fields. On the space of {\rannih} covectors we can define a symmetric bilinear form which generalizes the inverse of the metric. This metric will be used to perform metric contractions between covariant indices, for the tensor fields constructed with the help of the {\rannih} space, in section \sref{s_tensors_contraction_sign_const}.

The standard way to define the Levi-Civita connection is by raising an index in the Koszul formula. But this method works only for {\nondeg} metrics, since it relies on raising indices. We will be able to obtain differential operators like the covariant derivative for a class of tensor fields, by using the right-hand side of the Koszul formula (named here the {\koszulname}). The properties of the {\koszulname} are studied in section \sref{s_koszul_form}. They are similar to those of the Levi-Civita connection, and are used to construct a kind of covariant derivative for vector fields in section \sref{s_cov_der}, and a covariant derivative for differential forms in \sref{s_cov_der_forms}.

An important class of singular {\semiriem} manifolds is that on which the lower covariant derivative of any vector field, which is a $1$-form, admits smooth covariant derivatives. We define them in section \sref{s_riemann_curvature}, and call them {\em {\semireg} manifolds}.

The {\koszulname} and of the covariant derivatives for differential forms introduced in section \sref{s_cov_der} are used to construct, in \sref{s_riemann_curvature}, the Riemann curvature tensor, which has the same symmetry properties as in the {\nondeg} case. We show that on a {\semireg} {\semiriem} manifold, the Riemann curvature tensor is smooth. Since it is {\rannih} in all of its indices, we can construct from it the Ricci and scalar curvatures. Section \sref{s_riemann_curvature_ii}, proves a useful formula of the Riemann curvature tensor, directly in terms of the {\koszulname}. We compare the Riemann curvature we obtained, with the one obtained by Kupeli by different methods in \cite{Kup87b}.

We give in section \sref{s_semi_reg_semi_riem_man_example} two examples of {\semireg} {\semiriem} metrics: diagonal metrics, and degenerate conformal transformations of {\nondeg} metrics.

In section \sref{s_cartan}, we derive structural equations similar to those of Cartan, but for degenerate metrics.

Section \sref{s_warped} contains the generalization of the warped product to singular {\semiriem} manifolds, where the warping function can be allowed to vanish at some points, and the manifolds whose product is taken may be singular {\semiriem}. It is shown that the degenerate warped product of {\semiriem} manifolds is {\semiriem}, under very general conditions imposed to the warping function. Degenerate warped products will have important applications, both in the case of big-bang, and the black hole singularities.

\section{Singular {\semiriem} manifolds}
\label{s_singular_semi_riemannian}

\subsection{Definition of singular {\semiriem} manifolds}
\label{s_singular_semi_riemannian_def}

\begin{definition}(see {\eg} \cite{Kup87b}, \citep{Pam03}{265} for comparison)
\label{def_sing_semiRiemm_man}
Let $M$ be a differentiable manifold, and $g\in \tensors 0 2 M$ a symmetric bilinear form on $M$. Then the pair $(M,g)$ is called {\em singular {\semiriem} manifold}. The bilinear form $g$ is called the {\em metric tensor}, or the {\em metric}, or the {\em fundamental tensor}. 

For any point $p\in M$, the tangent space $T_pM$ has a frame in which the metric is
$$g_p=\diag(0,\ldots,0,-1,\ldots,-1,1,\ldots,1),$$
where $0$ appears $r$ times, $-1$ appears $s$ times, and $1$ appears $t$ times. The triple $(r,s,t)$ is called the \textit{signature} of $g$ at $p$. The relations $\dim M=r+s+t$, and $\rank g=s+t=n-r$, hold. 
If the signature of $g$ is fixed, then 
$(M,g)$ is said to have {\em constant signature}, otherwise is said to be with {\em variable signature}. If $g$ is {\nondeg}, then the manifold $(M,g)$ is {\em {\semiriem}}, and if $g$ is positive definite, it is {\em Riemannian}.
\end{definition}

\begin{remark}
The name ``singular {\semiriem} manifold'' may suggest that such a manifold is {\semiriem}. In fact it is more general, containing the {\nondeg} case as a subcase. This may be a misnomer, but we adhere to it, because it is generally used in the literature \cite{Moi40,Vra42,Kup96}. Moreover, for the geometric objects we introduce, which are similar and generalize objects from the {\nondeg} case, we will prefer to import the standard terminology from {\semiriem} geometry (see {\eg} \cite{ONe83}).
\end{remark}

The results obtained in the following don't necessarily require the {\metricname} to be degenerate. Hence, they also apply to {\semiriem} manifolds, including Riemannian. Moreover, while in standard materials like \cite{Kup87b,Kup96} was preferred to maintain the signature constant (because the most acute singular behavior takes place at the points where the signature changes) we will not assume constant signature.

\begin{example}
\label{ex_sing_semi_euclidean}
Let $r,s,t\in\N$, $n=r+s+t$. The space
\begin{equation}
	\sseuclid r s t:=(\R^n,\metric{,}),
\end{equation}
with the metric $g$ given, for any two vector fields $X$, $Y$ on $\R^n$ at a point $p$ on the manifold, in the natural chart, by
\begin{equation}
	\metric{X_p,Y_p} = -\sum_{i=r+1}^s X^i Y^i + \sum_{j=r+s+1}^n X^j Y^j,
\end{equation}
is called the singular {\semieucl} space $\sseuclid r s t$ (see {\eg} \citep{Pam03}{262}).
If $r=0$ we obtain the {\semieucl} space $\R^n_s:=\sseuclid 0 s t$ (see {\eg} \citep{ONe83}{58}).
If $r=s=0$, then $t=n$ and we obtain as particular case the Euclidean space $\R^n$, endowed with the natural scalar product.
\end{example}

\subsection{The radical of a singular {\semiriem} manifold}
\label{s_radix}

\begin{definition}(\cf \eg \citep{Bej95}{1}, \citep{Kup96}{3}, and \citep{ONe83}{53})
Let $V$ be a finite dimensional vector space, endowed with a symmetric bilinear form $g$ which may be degenerate. The space $\radix{V}:=V^\perp=\{X\in V|\forall Y\in V, g(X,Y)=0\}$ is named the {\em radical} of $V$. The symmetric bilinear form $g$ is {\nondeg} if and only if $\radix{V}=\{0\}$.
\end{definition}

\begin{definition}(see {\eg} \citep{Kup87b}{261}, \citep{Pam03}{263})
\label{def_radix}
Let $(M,g)$ be a singular {\semiriem} manifold. The subset of the tangent bundle defined by $\radix{T}M=\cup_{p\in M}\radix{(T_pM)}$ is called {\em the radical of $TM$}.
It is a vector bundle if and only if the signature of $g$ is constant on the entire $M$, and in this case, $\radix{T}M$ is a distribution.
We denote by $\fivectnull{M}\subseteq\fivect M$ the set of vector fields $W\in\fivect M$ for which $W_p\in\radix{(T_pM)}$.  
\end{definition}

\begin{example}
\label{ex_sing_semi_euclidean_radix}
In the case of the the singular {\semieucl} manifold $\sseuclid r s t$  from the Example \ref{ex_sing_semi_euclidean}, the radical $\radix{T}\sseuclid r s t$ is:
\begin{equation}
	\radix{T}\sseuclid r s t = \bigcup_{p\in\sseuclid r s t}\tn{span}({\{(p,\partial_{ap})|\partial_{ap}\in T_p\sseuclid r s t,a\leq r\}}),
\end{equation}
and
\begin{equation}
	\fivectnull{\sseuclid r s t} = \{X\in\fivect{\sseuclid r s t}|X=\sum_{a=1}^r X^a\partial_a\}.
\end{equation}
\end{example}

\section{The {\rannih}}
\label{s_dual_inner_prod}

Let $(V,g)$ be a vector space endowed with a bilinear form (also named inner product). If $g$ is {\nondeg}, it defines an isomorphism $\flat:V\to V^*$ (see {\eg} \citep{Gibb06}{15}; \citep{GHLF04}{72}). If $g$ is degenerate, $\flat$ is a linear morphism, but not an isomorphism. This prevents the definition of a dual for $g$ on $V^*$ in the usual sense. But we can still define canonically a symmetric bilinear form $\annihg\in\flat(V)^*\odot\flat(V)^*$, which extends immediately to singular {\semiriem} manifolds, and can be used to contract covariant indices and construct the needed geometric objects.

\subsection{The {\rannih} vector space}
\label{s_rad_annih_space}

\begin{definition}
\label{def_inner_morphism}
Let $\flat:V\to V^*$, which associates to any $u\in V$ a linear form $\flat(u):V\to \R$, defined by $\flat(u)v:=\metric{u,v}$.
Then, $\flat$ is a vector space morphism, called the {\em index lowering morphism}. We will also use the notation $u^\flat$ for $\flat(u)$, and sometimes $\annih{u}$.
\end{definition}

\begin{remark}
\label{thm_radix_ker}
It is easy to see that $\radix{V}=\ker\flat$, so $\flat$ is an isomorphism if and only if $g$ is {\nondeg}.
\end{remark}

\begin{definition}
\label{def_radical_annihilator}
The vector space $\annih{V}:=\IM\flat\subseteq V^*$ of $1$-forms $\omega$ which can be expressed, for some $u\in V$, as $\omega=u^\flat$, is called the {\em {\rannih} space} (it is the annihilator of the radical space $\radix{V}$). The forms $\omega\in\annih{V}$ are called {\em {\rannih} forms}, and act on $V$ by $\omega(v)=\metric{u,v}$.
\end{definition}

It is easy to see that $\dim\annih{V}+\dim\radix{V}=n$. If $g$ is {\nondeg}, $\annih{V}=V^*$.
If $u'\in V$ is another vector so that $u'^\flat=\omega$, then $u'-u\in\radix{V}$. Such $1$-forms $\omega\in\annih{V}$ satisfy $\omega|_{\radix{V}} = 0$.

\begin{definition}
\label{def_co_inner_product}
The symmetric bilinear form $g$ on $V$ defines on $\annih{V}$ a canonical {\nondeg} symmetric bilinear form $\annihg$, by $\annihg(\omega,\tau):=\metric{u,v}$, where $u^\flat=\omega$ and $v=\tau^\flat$. We alternatively use the notation $\annihprod{\omega,\tau}=\annihg(\omega,\tau)$. This notation is consistent with the notation $u^\flat=\annih{u}\in\annih{V}$.
\end{definition}

\begin{proposition}
The inner product $\annihg$ is well-defined, in the sense that it doesn't depend on the vectors $u,v$ which represent the $1$-forms $\omega$, $\tau$.
It is {\nondeg}, and has the signature $(0,s,t)$, where $(r,s,t)$ is the signature of $g$.
\end{proposition}
\begin{proof}
Let $u',v'\in V$ be other vectors so that $u'^\flat=\omega$ and $v'^\flat=\tau$. Then, since $u'-u\in\radix{V}$ and $v'-v\in\radix{V}$, $\metric{u',v'}=\metric{u,v}+\metric{u'-u,v}+\metric{u,v'-v}+\metric{u'-u,v'-v}=\metric{u,v}$.

Let $(e_a)_{a=1}^n$ be a basis in which $g$ is diagonal, the first $r$ diagonal elements being $0$. For $a\in\{1,\ldots,r\}$, $e_a^\flat=0$. For $a\in\{1,\ldots,s+t\}$, $\omega_a:=e_{r+a}^\flat$ are the generators of $\annih{V}$, and $\annihprod{\omega_a,\omega_b}=\metric{e_{r+a},e_{r+b}}$. Hence, $(\omega_a)_{a=1}^{s+t}$ are linear independent, and $\annihg$ has the signature $(0,s,t)$.
\end{proof}

In Figure \ref{degenerate-metric} we can see the various spaces associated with $(V,g)$, and the inner products induced by $g$ on them.

\image{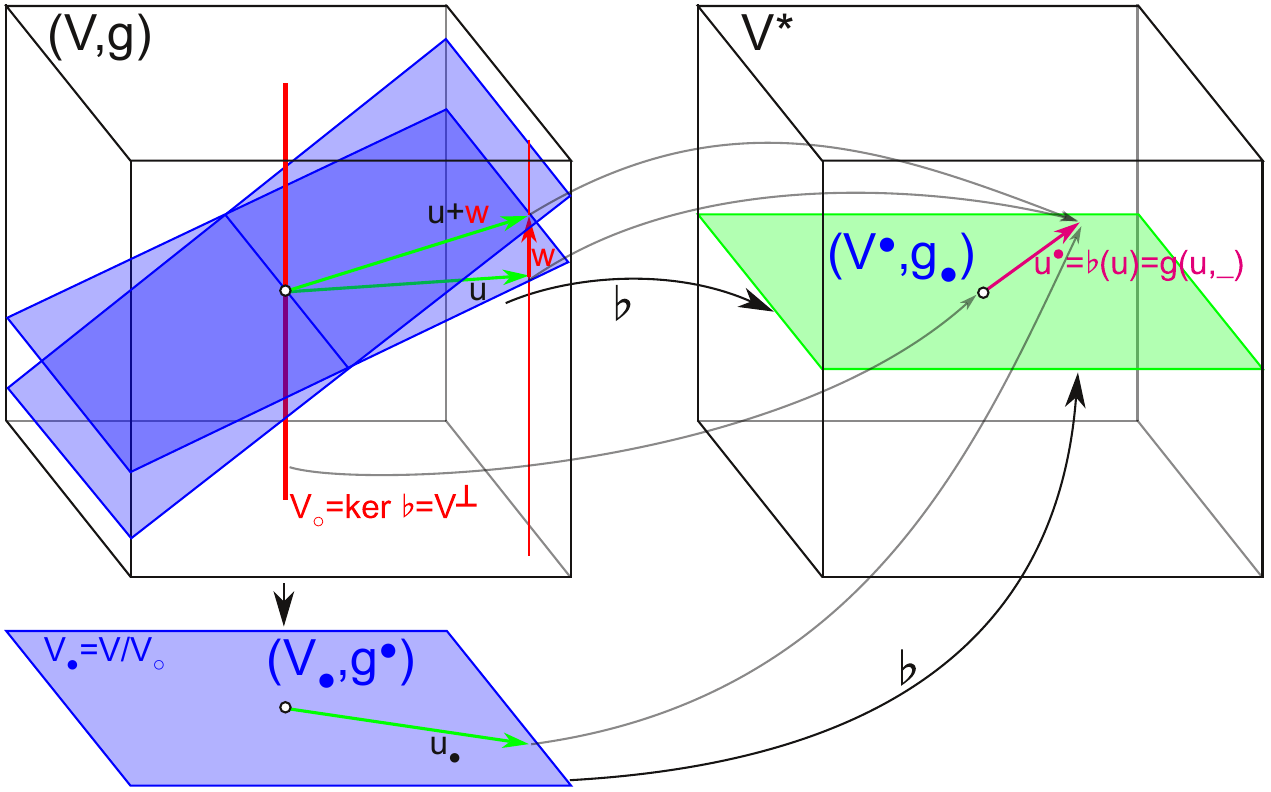}{0.7}{
Let $(V,g)$ be an inner product space. The morphism $\flat:V\to V^*$ is defined by $u\mapsto \annih{u}:=\flat(u)=u^\flat=g(u,\_)$. The set of isotropic vectors in $V$ forms the radical $\radix{V}:=\ker\flat=V^\perp$. The image of $\flat$ is $\annih{V}:=\IM{\flat}\leq V^*$. The inner product $g$ defines an inner product on $\annih{V}$, by $\annihg(u_1^\flat,u_1^\flat):=g(u_1,u_2)$. The inner product $\annihg$ is the inverse of $g$ iff $\det g\neq 0$. The quotient space $\coannih{V}:=V/\radix{V}$ is made of the equivalence classes of the form $u+\radix{V}$. On $\coannih{V}$, $g$ defines an inner product $\coannihg(u_1+\radix{V},u_2+\radix{V}):=g(u_1,u_2)$.}

The relations between the radical, the radical annihilator and the factor spaces can be collected in the diagram:
\begin{center}
\begin{tikzpicture}
\matrix (m) [matrix of math nodes, row sep=4em,
column sep=4em, text height=1.5ex, text depth=0.25ex]
{ 0  & \radix{V} & (V,g) & (\coannih{V},\coannihg) & 0 \\
0	& \coradix{V} & V^*	& (\annih{V},\annihg) & 0	\\ };
\path[right hook->]
(m-1-1) edge (m-1-2)
(m-1-2) edge node[auto] {$\radix{i}$}(m-1-3);
\path[->>]
(m-1-3) edge node[auto] {$\coannih{\pi}$}(m-1-4)
(m-1-4) edge(m-1-5);
\path[->>]
(m-2-2) edge (m-2-1)
(m-2-3) edge node[auto,swap] {$\coradix{\pi}$}(m-2-2)
(m-1-3) edge node[auto,sloped] {$\flat_V$}(m-2-4);
\path[left hook->]
(m-2-4) edge node[auto,swap] {$\annih{i}$}(m-2-3)
(m-2-5) edge(m-2-4);
\path[->]
(m-1-4) edge[bend right=15] node[auto,left] {$\flat$} (m-2-4)
(m-2-4) edge[bend right=15] node[auto,right]{$\sharp$} (m-1-4);
\end{tikzpicture}
\end{center}
where $\coannih{V}=\annih{V}^*=\dsfrac V{\radix V}$ and $\coradix{V}=\radix{V}^*=\frac {V^*}{\annih V}$.

Let $(e_a)_{a=1}^n$ be a basis of $V$ in which $g=\diag(\alpha_1,\alpha_2,\ldots,\alpha_n)$, $\alpha_a\in\R$ for all $1\leq a\leq n$. Then
\begin{equation}
g_{ab}=\metric{e_a,e_b}=\alpha_a\delta_{ab},
\end{equation}
\begin{equation*}
\annih{e_a}(e_b):=\metric{e_a,e_b}=\alpha_a\delta_{ab},
\end{equation*}
and, if $(e^{*a})_{a=1}^n$ is the dual basis of $(e_a)_{a=1}^n$,
\begin{equation}
\annih{e_a}=\alpha_a e^{*a}.
\end{equation}
It is easy to see that
\begin{equation}
\label{thm_cometric_in_basis}
\annihg^{ab}=\frac 1{\alpha_a}\delta^{ab},
\end{equation}for all $a$ so that $\alpha_a\neq 0$.

\subsection{The {\rannih} vector bundle}
\label{s_annih}

It is straightforward to extend the notions from previous section to the cotangent bundle. We define
\begin{equation}
	\annih{T}M=\bigcup_{p\in M}\annih{(T_pM)},
\end{equation}
and we call
\begin{equation}
	\annihforms{M}:=\{\omega\in\fiformk 1{M}|\omega_p\in\annih{(T_pM)}\tn{ for any }p\in M\}
\end{equation}
the space of {\em {\rannih} $1$-forms}.

\begin{example}
\label{ex_sing_semi_euclidean_annih}
The {\rannih} $\annih{T}\sseuclid r s t$ of the space $\sseuclid r s t$ from Example \ref{ex_sing_semi_euclidean} is
\begin{equation}
	\annih{T}\sseuclid r s t = \bigcup_{p\in\sseuclid r s t}\tn{span}({\{\de x^a\in T^*_p\sseuclid r s t|a> r\}}).
\end{equation}
The space of {\rannih} $1$-forms is
\begin{equation}
	\annihforms{\sseuclid r s t}=\{\omega\in\fiformk 1{\sseuclid r s t}|\omega^i=0,i\leq r\},
\end{equation}
and they have the general form
\begin{equation}
	\omega=\sum_{a=r+1}^n\omega_a\de x^a.
\end{equation}
\end{example}

\subsection{Radical and {\rannih} tensors}
\label{s_radix_annih_tensors}

The radical and {\rannih} spaces determine special subspaces of tensors on $T_pM$, on which we can define the metric contraction in covariant indices.

\begin{definition}
\label{def_radix_annih_tensor_field}
Let $T$ be a tensor of type $(r,s)$. If $T\in \tensors{k-1}0 {M}\otimes_M\radix{T}M\otimes_M \tensors{r-k}s{M}$, we call it {\em radical} in the $k$-th contravariant slot. If  $T\in \tensors r{l-1}{M}\otimes_M\annih{T}M\otimes_M \tensors 0{s-l}{M}$, we call it {\em {\rannih}} in the $l$-th covariant slot.
\end{definition}

\begin{proposition}
\label{thm_radical_contravariant_index}
Let $T\in\tensors r s {M}$ be a tensor. Then, $T$ is radical in the $k$-th contravariant slot if and only if its contraction with any {\rannih} linear $1$-form $\omega\in \fiformk 1{M}$, $C^k_{s+1}(T\otimes\omega)=0$.
\end{proposition}
\begin{proof}
Suppose $k=r$ (the case $k<r$ reduces to $k=r$, by using the permutation automorphisms of the tensor space $\tensors r s T_pM$). 
The tensor $T$ can be written as $T=\sum_{\alpha}S_{\alpha}\otimes v_{\alpha}$, with $S_{\alpha}\in\tensors{r-1}s T_pM$ and $v_{\alpha}\in T_pM$. The contraction $C^k_{s+1}(T\otimes\omega)$ becomes $\sum_{\alpha}S_{\alpha}\omega(v_{\alpha})$.
Since $T$ is radical in the $r$-th contravariant slot, $\omega(v_{\alpha})=0$, for all $\alpha$ and any $\omega\in\annih{T_pM}$. Therefore $\sum_{\alpha}S_{\alpha}\omega(v_{\alpha})=0$.

Reciprocally, if $\sum_{\alpha}S_{\alpha}\omega(v_{\alpha})=0$, it follows that for any $\alpha$, $S_{\alpha}\omega(v_{\alpha})=0$, for any $\omega\in\annih{T_pM}$. Then, $v_{\alpha}\in\radix{V}$.
\end{proof}

\begin{proposition}
\label{thm_radical_annihilator_covariant_index}
Let $T\in\tensors r s {M}$ be a tensor. Then, $T$ is {\rannih} in the $l$-th covariant slot if and only if its $l$-th contraction with any radical vector field $X$, $C^{r+1}_l(T\otimes X)=0$.
\end{proposition}
\begin{proof}
The proof is similar to that of Proposition \ref{thm_radical_contravariant_index}.
\end{proof}

\begin{example}
\label{thm_metric_radical_annihilator}
The inner product $g\in\tensors 0 2 {M}$ is {\rannih} in both of its slots.
\end{example}

\begin{proposition}
\label{thm_radical_annihilator_vs_radical_contraction}
Let $T\in\tensors r s {M}$ be a tensor, radical in it's $k$-th contravariant slot, and {\rannih} in it's $l$-th covariant slot. Then, the contraction $C^k_l(T)=0$.
\end{proposition}
\begin{proof}
The result follows from Proposition \ref{thm_radical_contravariant_index}.
\end{proof}

\section{Covariant contraction of tensor fields}
\label{s_tensors_contraction_sign_const}

Contractions between one covariant and one contravariant indices don't require a metric, and the inner product $g$ can contract between two contravariant indices, obtaining the {\em contravariant contraction operator} $C^{kl}$ \cfeg{ONe83}{83}. But how we define the metric contraction between two covariant indices, in the absence of an inverse of the metric? We will see that $\annihg$ can do this, although only for vectors or tensors which are {\rannih} in covariant slots. Luckily, these tensors are the relevant ones for our purpose.

\subsection{Covariant contraction on inner product spaces}
\label{s_tensors_covariant_contraction_inner_prod}

\begin{definition}
\label{def_contraction_covariant}
Let $T\in\annih{V}\otimes\annih{V}$, be a $(2,0)-$tensor which is {\rannih} in both its slots. Then, $C_{12}T:=\annihg^{ab}T_{ab}$ is the metric {\em covariant contraction}. This definition is independent on the basis, because $\annihg\in\annih{V}^*\otimes\annih{V}^*$. 
Let $r\geq 0$ and $s\geq 2$, and let $T\in\tensors r s V$ be a tensor which satisfies
\begin{equation}
	T\in V^{\otimes r}\otimes {V^*}^{\otimes {k-1}}\otimes\annih{V}\otimes {V^*}^{\otimes l-k-1}\otimes\annih{V}\otimes {V^*}^{\otimes s-l},
\end{equation}
$1\leq k<l\leq s$.
We define the operator
\begin{equation*}
C_{kl}:V^{\otimes r}\otimes {V^*}^{\otimes {k-1}}\otimes\annih{V}\otimes {V^*}^{\otimes l-k-1}\otimes\annih{V}\otimes {V^*}^{\otimes s-l}
\to V^{\otimes r}\otimes {V^*}^{\otimes {s-2}},
\end{equation*}
by $C_{kl}:=C_{s-1\,s}\circ P_{k,s-1;l,s}$, where $P_{k,s-1;l,s}:T\in\tensors r s V\to T\in\tensors r s V$ is the permutation isomorphisms which moves the $k$-th and $l$-th slots in the last two positions, and $C_{s-1\,s}$ is the operator
\begin{equation*}
C_{s-1\,s}:=1_{\tensors r {s-2} V}\otimes C_{1,2}:\tensors r {s-2} V\otimes \annih{V}\otimes\annih{V}\to\tensors r {s-2} V,
\end{equation*}
where $1_{\tensors r {s-2} V}:\tensors r {s-2} V\to\tensors r {s-2}V$ is the identity. Then, the operator $C_{kl}$ is named the metric {\em covariant contraction} between the covariant slots $k$ and $l$.
In a radical basis, the contraction has the form
\begin{equation}
\label{eq_contraction_covariant_inner_prod_space}
(C_{kl} T)^{a_1\ldots a_r}{}_{b_1\ldots\widehat{b}_k\ldots\widehat{b}_l\ldots b_s} :=	\annihg^{b_k b_l}T^{a_1\ldots a_r}{}_{b_1\ldots b_k\ldots b_l\ldots b_s}.
\end{equation}
We denote the contraction $C_{kl} T$ of $T$ also by 
\begin{equation*}
T(\omega_1,\ldots,\omega_r,v_1,\ldots,\cocontr,\ldots,\cocontr,\ldots,v_s).
\end{equation*}
\end{definition}

\subsection{Covariant contraction on singular {\semiriem} manifolds}
\label{ss_tensors_contraction_manifolds}

We can now extend the metric covariant contraction in {\rannih} slots to singular {\semiriem} manifolds.

\begin{definition}
\label{def_contraction_covariant_ct_sign}
Let $T\in\tensors r s {M}$, $s\geq 2$, be a tensor field on $M$, which is {\rannih} in the $k$-th and $l$-th covariant slots, $1\leq k<l\leq s$. The operator
\begin{equation*}
C_{kl}:\tensors r{k-1} {M}\otimes_M\annihforms M\otimes_M\tensors 0{l-k-1} {M}\otimes_M\annihforms M\otimes_M \tensors 0 {s-l}{M} \to \tensors{r}{s-2}{M}
\end{equation*}
\begin{equation*}
(C_{kl}T)(p)=C_{kl}(T(p))
\end{equation*}
is called the metric {\em covariant contraction} operator.
We denote it also by 
\begin{equation*}
T(\omega_1,\ldots,\omega_r,X_1,\ldots,\cocontr,\ldots,\cocontr,\ldots,X_s).
\end{equation*}
\end{definition}

\begin{lemma}
\label{thm_contraction_with_metric}
Let $T$ be a tensor field $T\in\tensors r s {M}$, {\rannih} in the $k$-th covariant slot, $1\leq k\leq s$, where $r\geq 0$ and $s\geq 1$. Then
\begin{equation}
\label{eq_contraction_with_metric}
\begin{array}{l}
T(\omega_1,\ldots,\omega_r,X_1,\ldots,\cocontr,\ldots,X_s)\metric{X_k,\cocontr}\\
\,\,\,\,\,=T(\omega_1,\ldots,\omega_r,X_1,\ldots,X_k,\ldots,X_s).
\end{array}
\end{equation}
\end{lemma}
\begin{proof}
Let's work at a point $p\in M$, and consider first the case $T_p\in\tensors 0 1 T_pM$, in fact, $T_p=\omega_p\in\annih{T}_p{M}$. Then, equation \eqref{eq_contraction_with_metric} becomes
\begin{equation}
	\omega_p(\cocontr)\metric{X_p,\cocontr}=\omega_p(X_p).
\end{equation}
But $\omega_p=Y_p^\flat$ for some vector $Y_p\in T_pM$, and $\omega_p(\cocontr)\metric{X_p,\cocontr}=\annihprod{\omega_p,X_p^\flat}=\omega(X_p)$.

The general case results from the linearity of the tensor product.
\end{proof}

\begin{corollary}
\label{thm_contracted_metric_w_metric}
$\metric{X,\cocontr}\metric{Y,\cocontr}=\metric{X,Y}.$
\end{corollary}
\begin{proof}
It is obtained by applying Lemma \ref{thm_contraction_with_metric} to $g\in\annihforms{M}\odot_M\annihforms{M}$.
\end{proof}

\begin{theorem}
\label{thm_contraction_orthogonal}
Let $(M,g)$ be a singular {\semiriem} manifold with constant signature. Let $T\in\tensors r s M$, be a tensor field as in Lemma \ref{thm_contraction_with_metric}.
Then
\begin{equation}
\label{eq_cov_contraction_orthogonal}
\begin{array}{l}
T(\omega_1,\ldots,\omega_r,X_1,\ldots,\cocontr,\ldots,\cocontr,\ldots,X_s) \\
\,\,\,\,\,
=\sum_{a=n-\rank g+1}^n \dsfrac{1}{\metric{E_a,E_a}}T(\omega_1,\ldots,\omega_r,X_1,\ldots,E_a,\ldots,E_a,\ldots,X_s),
\end{array}
\end{equation}
for any $X_1,\ldots,X_s\in\fivect M,\omega_1,\ldots,\omega_r\in\fiformk 1 M$,
where $(E_a)_{a=1}^n$ is a local orthogonal basis on $M$, so that $E_1,\ldots,E_{n-\rank g}\in\fivectnull{M}$. 
\end{theorem}
\begin{proof}
Since $\annihg$ is diagonal and $\annihg^{aa}=\dsfrac 1{g_{aa}}=\dsfrac{1}{\metric{E_a,E_a}}$ (Proposition \ref{thm_cometric_in_basis}),
\begin{equation*}
\begin{array}{l}
\annihg^{ab}T(\omega_1,\ldots,\omega_r,X_1,\ldots,E_a,\ldots,E_b,\ldots,X_s) \\
\,\,\,\,\,
=\sum_{a=n-\rank g+1}^n \dsfrac{1}{\metric{E_a,E_a}}T(\omega_1,\ldots,\omega_r,X_1,\ldots,E_a,\ldots,E_a,\ldots,X_s).
\end{array}
\end{equation*}
\end{proof}

The metric covariant contraction of a smooth tensor is smooth, except for the points where the signature changes, where the inverse of the metric becomes divergent, as seen in equation \eqref{thm_cometric_in_basis}.

\begin{example}
\label{thm_contracted_metric_w_itself}
Let $p\in M$. Then, 
$\metric{\cocontr,\cocontr}_p=\rank g_p$. This is also an example of metric covariant contraction which is discontinuous when the signature changes.
\end{example}

\begin{example}
\label{ex_contraction_sign_change_smooth}
Let $X\in \fivect{M}$ and $\omega\in\annihforms{M}$. Then, $C_{12}(\omega\otimes_M X^\flat)=\annihprod{\omega,X^\flat}=\omega(X)$ is smooth, even if the signature is not constant.
\end{example}

\section{The {\koszulname}}
\label{s_koszul_form}

\begin{definition}[The {\koszulname}, see {\eg} \citep{Kup87b}{263}]
\label{def_Koszul_form}
The object defined as
\begin{equation*}
	\kosz:\fivect M^3\to\R,
\end{equation*}
\begin{equation}
\label{eq_Koszul_form}
\begin{array}{llll}
	\kosz(X,Y,Z) &:=&\ds{\frac 1 2} \{ X \metric{Y,Z} + Y \metric{Z,X} - Z \metric{X,Y} \\
	&&\ - \metric{X,[Y,Z]} + \metric{Y, [Z,X]} + \metric{Z, [X,Y]}\}.
\end{array}
\end{equation}
is called {\em the {\koszulname}}.
\end{definition}

For {\nondeg} metric, the {\koszulname} is nothing but the right hand side of the Koszul formula
\begin{equation}
\label{eq_koszul_formula}
	\metric{\derb X Y,Z} = \kosz(X,Y,Z),
\end{equation}
which is used to construct the Levi-Civita connection, by raising the $1$-form $\kosz(X,Y,\_)$ (see {\eg} \citep{ONe83}{61})
\begin{equation}
\label{eq_koszul_formula_inv}
	\derb X Y = \kosz(X,Y,\_)^\sharp.
\end{equation}
This is not possible for degenerate metric. An alternative was proposed by Kupeli (\citep{Kup87b}{261--262}). He defined the so-called  {\em Koszul derivatives}, by raising the index with the help of a distribution complementary to $\radix{T}M$, provided that $(M,g)$ is a singular {\semiriem} manifold with metric of constant signature, which satisfies the condition of {\em radical-stationarity} (Definition \ref{def_radical_stationary_manifold}). His Koszul derivative is therefore not unique, depending on the choice of the complementary distribution, and is not a connection.
By contrast, our method doesn't rely on arbitrary constructions, and works even if the metric changes its signature. We will only rely on the {\koszulname}.

\subsection{Basic properties of the {\koszulname}}
\label{s_koszul_form_props}

We prove here explicitly some properties of the {\koszulname} which will be useful in the following, and which correspond to known properties of the Levi-Civita connection of a {\nondeg} metric \cfeg{ONe83}{61}, but are valid even on singular manifolds, where the Levi-Civita connection can't be defined.

\begin{theorem}
\label{thm_Koszul_form_props}
Let $(M,g)$ be a singular {\semiriem} manifold. Then, its {\koszulname} has, for any $X,Y,Z\in\fivect M$ and $f\in\fiscal M$, the properties:
\begin{enumerate}
	\item \label{thm_Koszul_form_props_linear}
	It is additive and $\R$-linear in all of its arguments.
	\item \label{thm_Koszul_form_props_flinearX}
	It is $\fiscal M$-linear in the first argument:

	$\kosz(fX,Y,Z) = f\kosz(X,Y,Z).$
	\item \label{thm_Koszul_form_props_flinearY}
	Satisfies the {\em Leibniz rule}:

	$\kosz(X,fY,Z) = f\kosz(X,Y,Z) + X(f) \metric{Y,Z}.$
	\item \label{thm_Koszul_form_props_flinearZ}
	It is $\fiscal M$-linear in the third argument:

	$\kosz(X,Y,fZ) = f\kosz(X,Y,Z).$
	\item \label{thm_Koszul_form_props_commutYZ}
	It is {\em metric}:

	$\kosz(X,Y,Z) + \kosz(X,Z,Y) = X \metric{Y,Z}$.
	\item \label{thm_Koszul_form_props_commutXY}
	It is {\em symmetric} or {\em torsionless}:

	$\kosz(X,Y,Z) - \kosz(Y,X,Z) = \metric{[X,Y],Z}$.
	\item \label{thm_Koszul_form_props_commutZX}
	Relation with the Lie derivative of $g$:

	$\kosz(X,Y,Z) + \kosz(Z,Y,X) = (\lie_Y g)(Z,X)$,

	where $(\lie_Y g)(Z,X):=Y\metric{Z,X} - \metric{[Y,Z],X} - \metric{Z,[Y,X]}$ is the Lie derivative of $g$ with respect to a vector field $Y\in\fivect M$.
	\item \label{thm_Koszul_form_props_commutX2Y}

	$\kosz(X,Y,Z) + \kosz(Y,Z,X) = Y\metric{Z,X} + \metric{[X,Y],Z}$.
	\end{enumerate}
\end{theorem}
\begin{proof}
\eqref{thm_Koszul_form_props_linear}\ 
It is a direct consequence of Definition \ref{def_Koszul_form}, the fact that $g$ is tensor, and the linearity of the action of vector fields on scalars, and of the Lie brackets.

The other properties can also be proved by using the properties of $g$, of the action of vector fields on scalars, of the Lie brackets, and of the Lie derivative, so we will give just one example.

\begin{equation*}
\begin{array}{llll}
\eqref{thm_Koszul_form_props_flinearX}\ 
&2\kosz(fX,Y,Z) &=& fX \metric{Y,Z} + Y \metric{Z,fX} - Z \metric{fX,Y} \\
	&&&- \metric{fX,[Y,Z]} + \metric{Y, [Z,fX]} + \metric{Z, [fX,Y]} \\
	&&=& fX \metric{Y,Z} + Y (f\metric{Z,X}) - Z (f\metric{X,Y}) \\
	&&&- f\metric{X,[Y,Z]}+ \metric{Y, f[Z,X] + Z(f)X} \\
	&&&+ \metric{Z, f[X,Y] - Y(f)X} \\
	&&=& fX \metric{Y,Z} + fY \metric{Z,X} \\
	&&&+ Y(f) \metric{Z,X} - fZ \metric{X,Y} \\
	&&&- Z(f)\metric{X,Y} - f\metric{X,[Y,Z]} + f\metric{Y, [Z,X]} \\
	&&&+ Z(f)\metric{Y, X} + f\metric{Z, [X,Y]} - Y(f)\metric{Z, X} \\
	&&=& fX \metric{Y,Z} + fY \metric{Z,X} - fZ \metric{X,Y} \\
	&&& - f\metric{X,[Y,Z]} + f\metric{Y,[Z,X]} + f\metric{Z,[X,Y]} \\
	&&=& 2f\kosz(X,Y,Z) \\
\end{array}
\end{equation*}
\end{proof}

\begin{remark}
\label{thm_Koszul_form_index}
Let $(E_a)_{a=1}^n\subset\fivect U$ a local frame of vector fields on an open set $U\subseteq M$.
Let $g_{ab} = \metric{E_a,E_b}$, and let $\ms C^c_{ab}$ so that $[E_a,E_b] = \ms C_{ab}^c E_c$.
Then,
\begin{equation}
\label{eq_Koszul_form_index}
\begin{array}{lll}
	\kosz_{abc}&:=&\kosz(E_a,E_b,E_c) \\
	&=&\ds{\frac 1 2} \{E_a(g_{bc}) + E_b(g_{ca}) - E_c(g_{ab})
	- g_{as} \ms C^s_{bc} + g_{bs} \ms C^s_{ca} + g_{cs} \ms C^s_{ab}\}.
\end{array}
\end{equation}

In the basis $(E_a)_{a=1}^n$, the equations (\ref{thm_Koszul_form_props_commutYZ} -- \ref{thm_Koszul_form_props_commutX2Y}) in Theorem \ref{thm_Koszul_form_props} become:
\begin{equation*}
\begin{array}{ll}
	(\ref{thm_Koszul_form_props_commutYZ}')
	& \kosz_{abc} + \kosz_{acb} = E_a (g_{bc}). \\
	(\ref{thm_Koszul_form_props_commutZX}')
	& \kosz_{abc} + \kosz_{cba} = (\lie_{E_b} g)_{ca}. \\
	(\ref{thm_Koszul_form_props_commutXY}')
	& \kosz_{abc} - \kosz_{bac} = g_{sc}\ms C^s_{ab}. \\
	(\ref{thm_Koszul_form_props_commutX2Y}')
	& \kosz_{abc} + \kosz_{bca} = E_b(g_{ca}) + g_{sc}\ms C^s_{ab}. \\
\end{array}
\end{equation*}

If $E_a=\partial_a:=\ds{\frac {\partial}{\partial x^a}}$ for all $a\in\{1,\ldots,n\}$, then $[\partial_a,\partial_b]=0$, and the coefficients of the {\koszulname} become Christoffel's symbols of the first kind,
\begin{equation}
\label{eq_Koszul_form_coord}
	\kosz_{abc}=\kosz(\partial_a,\partial_b,\partial_c)=\ds{\frac 1 2} (
	\partial_a g_{bc} + \partial_b g_{ca} - \partial_c g_{ab}).
\end{equation}
\end{remark}

\begin{corollary}
\label{thm_Koszul_form}
Let $X,Y\in\fivect{M}$ be two vector fields. The map $\kosz_{XY}:\fivect{M}\to\R$
\begin{equation}
	\kosz_{XY}(Z) := \kosz(X,Y,Z)
\end{equation}
for any vector field $Z\in\fivect{M}$, is a differential $1$-form.
\end{corollary}
\begin{proof}
Follows directly from Theorem \ref{thm_Koszul_form_props}, properties \eqref{thm_Koszul_form_props_linear} and \eqref{thm_Koszul_form_props_flinearZ}.
\end{proof}

\begin{corollary}
\label{thm_Koszul_null_props}
Let $X,Y\in \fivect M$, and $W\in\fivectnull{M}$. Then
\begin{equation}
	\kosz(X,Y,W) = \kosz(Y,X,W) = -\kosz(X,W,Y) = -\kosz(Y,W,X). \\
\end{equation}
\end{corollary}
\begin{proof}
The first identity follows from Theorem \ref{thm_Koszul_form_props}, property \eqref{thm_Koszul_form_props_commutXY}, and the other two from property \eqref{thm_Koszul_form_props_commutYZ}.
\end{proof}

\section{The covariant derivative}
\label{s_cov_der}

In this section we will show that, even when the metric is degenerate, it is possible to define in a canonical way the covariant derivative for differential forms which are radical annihilator. We can use the {\koszulname} for a sort of covariant derivative, named here lower covariant derivative, of vector fields, which associates to vector fields not vector fields, but $1$-forms. This means it doesn't have the properties of a connection \cite{kob:1963,teleman2000connections}, but for our purpose will do the same job.

\subsection{The lower covariant derivative of vector fields}
\label{s_l_cov_dev}

\begin{definition}
\label{def_l_cov_der}
Let $X,Y\in\fivect{M}$. The $1$-form $\lderb XY \in \fiformk 1{M}$, defined by
\begin{equation}
\label{eq_l_cov_der_vect}
\lderc XYZ := \kosz(X,Y,Z)
\end{equation}
for any $Z\in\fivect{M}$, is called the {\em lower covariant derivative} of the vector field $Y$ in the direction of the vector field $X$.
The operator
\begin{equation}
	\lder:\fivect{M} \times \fivect{M} \to \fiformk 1{M},
\end{equation}
which associates to each $X,Y\in\fivect{M}$ the differential $1$-form $\ldera XY$, is called the {\em lower covariant derivative operator}.
\end{definition}

\begin{remark}
The lower covariant derivative is well defined even if the metric is degenerate. When applying the lower covariant derivative to a vector field we don't obtain another vector field, but a differential $1$-form, so it is not exactly a covariant derivative. If the metric is {\nondeg}, one obtains the covariant derivative by $\derb XY=(\lderb XY)^\sharp$. In \citep{Koss85}{464--465} were used similar objects which map vector fields to $1$-forms.
\end{remark}

\begin{theorem}
\label{thm_l_cov_der_props}
Let $(M,g)$  be a a singular {\semiriem} manifold. The lower covariant derivative operator $\lder$ has the properties:
\begin{enumerate}
	\item \label{thm_l_cov_der_props_linear}
	It is additive and $\R$-linear in both of its arguments.
	\item \label{thm_l_cov_der_props_flinearX}
	It is $\fiscal M$-linear in the first argument:

	$\lderb{fX}Y = f\lderb XY.$
	\item \label{thm_l_cov_der_props_flinearY}
	Satisfies the {\em Leibniz rule}:

	$\lderb X{fY} = f\lderb XY + X(f) Y^\flat,$
	
	or, explicitly,

	$\lderc X{fY}Z = f\lderc XYZ + X(f) \metric{Y,Z}.$
	\item \label{thm_l_cov_der_props_flinearZ}
	It is {\em metric}:

	$\lderc XYZ + \lderc XZY = X \metric{Y,Z}$.
	\item \label{thm_l_cov_der_props_commutXY}
	It is {\em symmetric} or {\em torsionless}:

	$\lderb XY - \lderb YX = [X,Y]^\flat,$
	
	or, explicitly,
	
	$\lderc XYZ - \lderc YXZ = \metric{[X,Y],Z}$.
	\item \label{thm_l_cov_der_props_commutZX}
	Relation with the Lie derivative of $g$:

	$\lderc XYZ + \lderc ZYX = (\lie_Y g)(Z,X)$.
	\item \label{thm_l_cov_der_props_commutX2Y}

	$\lderc XYZ + \lderc YZX = Y\metric{Z,X} + \metric{[X,Y],Z}$,
	\end{enumerate}
	for any $X,Y,Z\in\fivect M$ and $f\in\fiscal M$.
\end{theorem}
\begin{proof}
Follows directly from Theorem \ref{thm_Koszul_form_props}.
\end{proof}

\subsection{{\rrstationary} singular {\semiriem} manifolds}
\label{s_radical_stationary_manifolds}

\begin{definition}[{\cf} \cite{Kup96} Definition 3.1.3]
\label{def_radical_stationary_manifold}
A singular {\semiriem} manifold $(M,g)$ is {\em {\rstationary}} if it satisfies the condition 
\begin{equation}
\label{eq_radical_stationary_manifold}
		\kosz(X,Y,\_)\in\annihforms M,
\end{equation}
for any $X,Y\in\fivect{M}$.
Equivalently, $\kosz(X,Y,W_p)=0$ for any $X,Y\in \fivect M$ and $W_p\in \fivectnull{M_p}$, $p\in M$.
\end{definition}

\begin{remark}
Kupeli introduced and studied the {\rstationary} singular {\semiriem} manifolds of constant signature in \citep{Kup87b}{259--260}. In \cite{Kup96} Definition 3.1.3, he called them ``stationary singular {\semiriem} manifolds''. We will call them here ``{\rstationary} {\semiriem} manifolds'', to avoid confusions with the more spread usage of the term ``stationary'' for manifolds admitting a Killing vector field, and particularly for spacetimes invariant at time translation. Kupeli needed them to ensure the existence of what he called ``Koszul derivative'', which is not needed in our invariant approach.
\end{remark}

\begin{corollary}
\label{thm_Koszul_null_props_rad_stat}
Let $(M,g)$ be {\rstationary}, and $X,Y\in \fivect M$ and $W\in\fivectnull{M}$. Then,
\begin{equation}
	\kosz(X,Y,W) = \kosz(Y,X,W) = -\kosz(X,W,Y) = -\kosz(Y,W,X) = 0. \\
\end{equation}
\end{corollary}
\begin{proof}
It is a direct application of Corollary \ref{thm_Koszul_null_props}.
\end{proof}

\begin{remark}
\label{rem_rad_stat_lower_der}
The manifold $(M,g)$ is {\rstationary} iff, for any $X,Y\in\fivect{M}$,
\begin{equation}
		\lderb X Y\in\annihforms M.
\end{equation}
\end{remark}

\subsection{The covariant derivative of differential forms}
\label{s_cov_der_forms}

When $g$ is {\nondeg}, the covariant derivative of a differential $1$-form $\omega$ is
\begin{equation}
	\left(\der_X\omega\right)(Y) = X\left(\omega(Y)\right) - \omega\left(\der_X Y\right).
\end{equation}
To generalize this definition to the case of degenerate metrics, we have to replace $\omega\left(\der_XY\right)$ with an expression which is well-defined even if the metric is degenerate, {\ie}
\begin{equation}
\label{eq_cov_der_form_raise}
	\omega\left(\der_XY\right) = \kosz(X,Y,\cocontr)\omega(\cocontr).
\end{equation}
This is defined on {\rstationary} {\semiriem} manifolds, if the $1$-form $\omega$ is {\rannih}.
This leads naturally to the definition:

\begin{definition}
\label{def_cov_der_covect}
Let $(M,g)$ be a {\rstationary} {\semiriem} manifold. The operator
\begin{equation}
\begin{array}{lll}
	\left(\der_X\omega\right)(Y) := X\left(\omega(Y)\right) - \annihprod{\lderb X Y,\omega},
\end{array}
\end{equation}
is called the {\em covariant derivative} of a {\rannih} $1$-form $\omega\in\annihforms{M}$ in the direction of a vector field $X\in\fivect{M}$.
\end{definition}

\begin{proposition}
\label{thm_cov_deriv_annih_smooth}
If $\der_X\omega$ is smooth, then $\der_X\omega\in\annihforms M$.
\end{proposition}
\begin{proof}
Follows from Definition \ref{def_cov_der_covect}: $\left(\der_X\omega\right)(W) = X\left(\omega(W)\right) - \annihprod{\lderb X W,\omega} = 0$.
\end{proof}

\begin{notation}
\label{def_cov_der_smooth}
If $(M,g)$ is a {\rstationary} {\semiriem} manifold, then
\begin{equation}
	\srformsk 1 M = \{\omega\in\annihforms M|(\forall X\in\fivect M)\ \der_X\omega\in\annihforms M\},
\end{equation}
\begin{equation}
	\srformsk k M := \bigwedge^k_M\srformsk 1 M,
\end{equation}
denote the vector spaces of differential forms having smooth covariant derivatives.
\end{notation}

\begin{theorem}
\label{thm_cov_der_covect_props}
Let $(M,g)$ be a {\rstationary} {\semiriem} manifold. Then, the covariant derivative operator $\der$ of differential $1$-forms has the following properties:
\begin{enumerate}
	\item \label{thm_cov_der_covect_props_linear}
	It is additive and $\R$-linear in both of its arguments.
	\item \label{thm_cov_der_covect_props_flinearX}
	It is $\fiscal M$-linear in the first argument:

	$\derb{fX}\omega = f\derb X\omega.$
	\item \label{thm_cov_der_covect_props_flinearY}
	It satisfies the {\em Leibniz rule}:

	$\derb X{f\omega} = f\derb X\omega + X(f) \omega.$
	\item \label{thm_cov_der_covect_props_flat_commut}
	It commutes with the lowering operator:

	$\derb X{Y^\flat} = \lderb XY$,
	\end{enumerate}
	for any $X,Y\in\fivect M$, $\omega\in\annihforms{M}$ and $f\in\fiscal M$.
\end{theorem}
\begin{proof}
Property \eqref{thm_cov_der_covect_props_linear} follows directly from Theorem \ref{thm_l_cov_der_props} and Definition \ref{def_cov_der_covect}.

\begin{equation}
\eqref{thm_cov_der_covect_props_flinearX}\ 
	\derc{fX}\omega Y = fX\left(\omega(Y)\right) - \annihprod{\lderb {fX} Y,\omega} = f \derc X\omega Y.
\end{equation}

\begin{equation}
\begin{array}{lll}
\eqref{thm_cov_der_covect_props_flinearY}\ 
	\derc{X}{f\omega}Y &=& X\left(f\omega(Y)\right) - \annihprod{\lderb {X} Y,f\omega} \\
	&=& X(f)\omega(Y) +fX\left(\omega(Y)\right) - f\annihprod{\lderb {X} Y,\omega}\\
	&=& f\derc X\omega Y + X(f) \omega(Y).
\end{array}
\end{equation}

\begin{equation}
\begin{array}{lll}
\eqref{thm_cov_der_covect_props_flat_commut}\ 
	\derc{X}{Y^\flat}Z &=& X\left(Y^\flat(Z)\right) - \annihprod{\lderb X Z,Y^\flat} \\
	&=& X\metric{Y,Z} - \lderc X Z Y \\
	&=& \lderc X Y Z. \\
\end{array}
\end{equation}
\end{proof}

Now we define the covariant derivative for tensors which are covariant and radical annihilator in all their slots (including differential forms). The obtained formulae generalize the corresponding ones from the {\nondeg} case, see {\eg} \citep{GHLF04}{70}).

\begin{definition}
\label{def_cov_der_cov_tensors}
Let $(M,g)$ be a {\rstationary} {\semiriem} manifold. The operator $\der:\fivect{M} \times \otimes^s_M\srformsk 1 M \to \otimes^s_M\annihformsk 1 M$,
\begin{equation}
	\der_X(\omega_1\otimes\ldots\otimes\omega_s) := \der_X(\omega_1)\otimes\ldots\otimes\omega_s +\ldots + \omega_1\otimes\ldots\otimes\der_X(\omega_s)
\end{equation}
is called the {\em covariant derivative} of tensors of type $(0,s)$.
In particular, for a differential $k$-form, $\der:\fivect{M} \times \srformsk k M \to \annihformsk k M$,
\begin{equation}
	\der_X(\omega_1\wedge\ldots\wedge\omega_s) := \der_X(\omega_1)\wedge\ldots\wedge\omega_s +\ldots + \omega_1\wedge\ldots\wedge\der_X(\omega_s).
\end{equation}
\end{definition}

\begin{theorem}
\label{thm_cov_der_cov_tensors}
Let $(M,g)$ be a {\rstationary} {\semiriem} manifold. Then,
\begin{equation}
\begin{array}{lll}
	\left(\nabla_X T\right)(Y_1,\ldots,Y_k) &=& X\left(T(Y_1,\ldots,Y_k)\right) \\
	&& - \sum_{i=1}^k\kosz(X,Y_i,\cocontr)T(Y_1,,\ldots,\cocontr,\ldots,Y_k).
\end{array}
\end{equation}
\end{theorem}
\begin{proof}
We will prove it for the case $T = \omega_1\otimes_M\ldots\otimes_M\omega_k$, and extend by linearity. The proof follows by applying the Definitions \ref{def_cov_der_cov_tensors} and \ref{def_cov_der_covect},
\begin{equation}
\begin{array}{lll}
	(\der_XT)(Y_1,\ldots,Y_k) &=& \derc X {\omega_1}{Y_1}\cdot\ldots\cdot\omega_k(Y_k) +\ldots \\
 && + \omega_1(Y_1)\cdot\ldots\cdot\derc X {\omega_k}{Y_k} \\
 &=& (X(\omega_1(Y_1)) - \annihprod{\lderb X Y_1,\omega_1})\cdot\ldots\cdot\omega_k(Y_k) +\ldots \\
 && + \omega_1(Y_1)\cdot\ldots\cdot(X(\omega_k(Y_k)) - \annihprod{\lderb X Y_k,\omega_k}) \\
 &=& X(\omega_1(Y_1))\cdot\ldots\cdot\omega_k(Y_k) + \ldots \\
 && +  \omega_1(Y_1)\cdot\ldots\cdot X(\omega_k(Y_k)) \\
 && - \annihprod{\lderb X Y_1,\omega_1}\cdot\ldots\cdot\omega_k(Y_k) \\
 && - \omega_1(Y_1)\cdot\ldots\cdot\annihprod{\lderb X Y_k,\omega_k} \\
 &=&  X\left(T(Y_1,\ldots,Y_k)\right) \\
 && - \ds{\sum_{i=1}^k}\kosz(X,Y_i,\cocontr)T(Y_1,,\ldots,\cocontr,\ldots,Y_k).
\end{array}
\end{equation}
\end{proof}

\begin{corollary}
Let $(M,g)$ be a {\rstationary} {\semiriem} manifold. Then, the metric $g$ is parallel:
\begin{equation}
	\nabla_Xg = 0.
\end{equation}
\end{corollary}
\begin{proof}
We apply Theorems \ref{thm_cov_der_cov_tensors} and \ref{thm_Koszul_form_props}, property \eqref{thm_Koszul_form_props_commutYZ}:
\begin{equation}
	(\nabla_Xg)(Y,Z) = X\metric{Y,Z} - \kosz(X,Y,\cocontr)g(\cocontr,Z) - \kosz(X,Z,\cocontr)g(Y,\cocontr) = 0.
\end{equation}
\end{proof}

\subsection{{\ssemireg} {\semiriem} manifolds}
\label{s_semi_regular}

\begin{definition}
\label{def_semi_regular_semi_riemannian}
Let $(M,g)$ be a singular {\semiriem} manifold. If 
\begin{equation}
	\ldera X Y \in\srformsk 1 M
\end{equation}
for any $X,Y\in\fivect{M}$, 
$(M,g)$ is called {\em {\semireg} {\semiriem} manifold}.
\end{definition}

\begin{remark}
Since $\srformsk 1 M \subseteq \annihforms M$, any {\semireg} {\semiriem} manifold is {\rstationary} ({\cf} Definition \ref{def_radical_stationary_manifold}).
\end{remark}

\begin{remark}
\label{rem_semi_regular_semi_riemannian}
From Definition \ref{def_cov_der_smooth}, $(M,g)$ is {\semireg} iff for any $X,Y,Z\in\fivect M$,
\begin{equation}
	\dera X {\ldera Y}Z \in \annihforms M.
\end{equation}
\end{remark}

\begin{proposition}
\label{thm_sr_cocontr_kosz}
A {\rstationary} {\semiriem} manifold $(M,g)$ is {\semireg} if and only if for any $X,Y,Z,T\in\fivect M$,
\begin{equation}
	\kosz(X,Y,\cocontr)\kosz(Z,T,\cocontr) \in \fiscal M.
\end{equation}
\end{proposition}
\begin{proof}
From Definition \ref{def_cov_der_covect} follows that
\begin{equation}
\begin{array}{lll}
	\derc X {\lderb YZ} T 
	&=& X\left(\lderc Y Z T\right) - \annihprod{\lderb X T,\lderb YZ} \\
	&=& X\left(\lderc Y Z T\right) - \kosz(X,T,\cocontr)\kosz(Y,Z,\cocontr). \\
\end{array}
\end{equation}
\end{proof}

\section{Curvature of {\semireg} {\semiriem} manifolds}
\label{s_riemann_curvature}

If the {\metricname} is {\nondeg}, the curvature is constructed from the Levi-Civita connection \cfeg{ONe83}{59}. But there is no Levi-Civita connection if $g$ is degenerate.

In this section we will propose another way to define the Riemann curvature tensor, which works and is invariant even if the {\metricname} is degenerate.

Further, in \sref{s_ricci_tensor_scalar}, we will obtain the Ricci curvature tensor and the scalar curvature, by using the metric contraction of the Riemann curvature tensor in two covariant indices, introduced in section \sref{s_tensors_contraction_sign_const}. For a degenerate metric, this covariant contraction requires the Riemann curvature tensor to be {\rannih} in all its slots. Luckily, this condition holds.

\subsection{Riemann curvature of {\semireg} {\semiriem} manifolds}
\label{ss_riemann_curvature}

\begin{definition}
\label{def_riemann_curvature_operator}
Let $(M,g)$ be a {\rstationary} {\semiriem} manifold. The operator 
\begin{equation}
\label{eq_riemann_curvature_operator}
	\curv XY Z := \dera X {\ldera Y}Z - \dera Y {\ldera X}Z - \ldera {[X,Y]}Z,
\end{equation}
for any vector fields $X,Y,Z\in\fivect{M}$, is called the {\em lower Riemann curvature operator}.
\end{definition}

\begin{theorem}
\label{thm_riemann_curvature_semi_regular}
Let $(M,g)$ be a {\semireg} {\semiriem} manifold. The object
\begin{equation}
\label{eq_riemann_curvature}
	R(X,Y,Z,T) := (\curv XY Z)(T),
\end{equation}
for any vector fields $X,Y,Z,T\in\fivect{M}$, is a smooth tensor field $R\in\tensors 0 4 M$.
\end{theorem}
\begin{proof}
The Riemann curvature $R$ is additive and $\R$-linear in all its arguments, due to Theorem \ref{thm_l_cov_der_props}, property \eqref{thm_l_cov_der_props_linear}, and Theorem \ref{thm_cov_der_covect_props}, property \eqref{thm_l_cov_der_props_linear}.

To show now that $R$ is $\fiscal M$-linear in all its arguments, we check by using the properties of the lower covariant derivative for vector fields (Theorem \ref{thm_l_cov_der_props} properties \eqref{thm_l_cov_der_props_flinearX}-\eqref{thm_l_cov_der_props_flinearZ}), and those of the covariant derivative for differential $1$-forms (Theorem \ref{thm_cov_der_covect_props}, properties \eqref{thm_cov_der_covect_props_flinearX}-\eqref{thm_cov_der_covect_props_flat_commut}). It follows that for any function $f\in\fiscal{M}$, $R(fX,Y,Z,T)=R(X,fY,Z,T)=R(X,Y,fZ,T)=R(X,Y,Z,fT)=fR(X,Y,Z,T)$. For example,
\begin{equation*}
\begin{array}{lll}
	R(fX,Y,Z,T) &=& \derc {fX} {{\ldera Y}Z}T - \derc Y {{\ldera {fX}}Z}T - \lderc {[fX,Y]}ZT \\
	&=& f\derc {X} {{\ldera Y}Z}T - \derc Y {{(f\ldera {X}}Z)}T \\
	&& - \lderc {f[X,Y]-Y(f)X}ZT \\
	&=& f\derc {X} {{\ldera Y}Z}T - f\derc Y {{\ldera {X}}Z}T \\
	&& - Y(f)\lderc X Z T - f\lderc {[X,Y]}ZT \\
	&&  + Y(f)\lderc {X}ZT \\
	&=& fR(X,Y,Z,T). \\
\end{array}
\end{equation*}

The smoothness of $R$ follows from that of ${\ldera X}{Z}$, ${\ldera Y}{Z}$, and ${\ldera {[X,Y]}}{Z}$.
\end{proof}

\begin{definition}
\label{def_riemann_curvature}
The object from equation \eqref{eq_riemann_curvature} is called the {\em Riemann curvature tensor}. It generalizes the Riemann curvature tensor $R(X,Y,Z,T) := \metric{R_{XY}Z,T}$ known from {\semiriem} geometry \cfeg{ONe83}{75}.
\end{definition}

\begin{notation}
We denote by $\curv{}{}$ the map $\curv{}{}: \fivect M ^2 \to \tensors 0 2 M$, 
\begin{equation}
	\curv XY := \dera X {\ldera Y} - \dera Y {\ldera X} - \ldera {[X,Y]},
\end{equation}
where, for any $Z,T\in\fivect M$,
\begin{equation}
	\curv XY(Z,T) := (\curv XYZ)(T).
\end{equation}
\end{notation}

\subsection{The symmetries of the Riemann curvature tensor}
\label{s_riemann_curvature_symmetries}

The well-known symmetry properties of the Riemann curvature tensor of a {\nondeg} metric \cfeg{ONe83}{75} can be extended to  {\semireg} metrics.

\begin{proposition}
\label{thm_curv_symm}
Let $(M,g)$ be a {\semireg} {\semiriem} manifold. Then, for any $X,Y,Z,T\in\fivect M$, the Riemann curvature has the following symmetry properties
\begin{enumerate}
	\item \label{thm_curv_symm_xy}
	$\curv XY = -\curv YX$
	\item \label{thm_curv_symm_zt}
	$\curv XY(Z,T) = -\curv XY(T,Z)$
	\item  \label{thm_curv_symm_xyz}
	$\curv YZ X + \curv ZX Y + \curv XY Z = 0$
	\item  \label{thm_curv_symm_xy_zt}
	$\curv XY(Z,T) = \curv ZT(X,Y)$
\end{enumerate}
\end{proposition}
\begin{proof}
\begin{equation*}
\begin{array}{lll}
\eqref{thm_curv_symm_xy}\
	\curv XY Z&=& \dera X {\ldera Y}Z - \dera Y {\ldera X}Z - \ldera {[X,Y]}Z \\
	&=&-\curv YX Z
\end{array}
\end{equation*}

\eqref{thm_curv_symm_zt}
It is enough to prove that, for any $V\in\fivect M$,
\begin{equation}
	\curv XY(V,V)=0.
\end{equation}
Definition \ref{def_cov_der_covect} and Theorem \ref{thm_l_cov_der_props}, property \eqref{thm_l_cov_der_props_flinearZ} implies that
\begin{equation}
	\derc X {\lderb YV} V  = \dsfrac 1 2 XY\metric{V,V} - \annihprod{\lderb X V,\lderb YV}.
\end{equation}
Also Theorem \ref{thm_l_cov_der_props}, property \eqref{thm_l_cov_der_props_flinearZ}, saids
\begin{equation*}
\lderc{[X,Y]}VV=\frac 1 2[X,Y]\metric{V,V}
\end{equation*}
Hence,
\begin{equation*}
\begin{array}{lll}
	\curv XY(V,V) 
	&=& \dsfrac 1 2X\left(\lderc Y V V\right) - \annihprod{\lderb X V,\lderb YV} \\
	&& - \dsfrac 1 2Y\left(\lderc X V V\right) + \annihprod{\lderb Y V,\lderb XV} \\
	&& - \frac 1 2[X,Y]\metric{V,V} = 0\\
\end{array}
\end{equation*}

\eqref{thm_curv_symm_xyz}
We define the cyclic sum for any $F:\fivect M^3\to\fiformk 1 M$ by
\begin{equation}
\begin{array}{l}
	\sum_{\cyclic}F(X,Y,Z):=F(X,Y,Z)+F(Y,Z,X)+F(Z,X,Y).
\end{array}
\end{equation}
Since it doesn't change at cyclic permutations of $X,Y,Z$, from the properties of the lower covariant derivative and from Jacobi's identity,
\begin{equation*}
\begin{array}{lll}
\sum_{\cyclic}\curv XY Z
&=& \sum_{\cyclic}\dera X {\ldera Y}Z - \sum_{\cyclic}\dera Y {\ldera X}Z - \sum_{\cyclic}\ldera {[X,Y]}Z\\
&=& \sum_{\cyclic}\dera X {\ldera Y}Z - \sum_{\cyclic}\dera X {\ldera Z}Y - \sum_{\cyclic}\ldera {[X,Y]}Z\\
&=& \sum_{\cyclic}\dera X \left({\ldera Y}Z - {\ldera Z}Y\right) - \sum_{\cyclic}\ldera {[X,Y]}Z\\
&=& \sum_{\cyclic}\dera X [Y,Z]^\flat - \sum_{\cyclic}\ldera {[X,Y]}Z\\
&=& \sum_{\cyclic}\dera X^\flat [Y,Z] - \sum_{\cyclic}\ldera {[Y,Z]}X\\
&=& \sum_{\cyclic}[X,[Y,Z]]^\flat = 0.\\
\end{array}
\end{equation*}

\eqref{thm_curv_symm_xy_zt} By applying \eqref{thm_curv_symm_xyz} four times (just like in the proof of the properties of the curvature for {\nondeg} metric),
\begin{equation*}
\begin{array}{lllllll}
	\curv XY(Z,T) &+& \curv YZ(X,T) &+& \curv ZX(Y,T) &=& 0 \\
	\curv YZ(T,X) &+& \curv ZT(Y,X) &+& \curv TY(Z,X) &=& 0 \\
	\curv ZT(X,Y) &+& \curv TX(Z,Y) &+& \curv XZ(T,Y) &=& 0 \\
	\curv TX(Y,Z) &+& \curv XY(T,Z) &+& \curv YT(X,Z) &=& 0. \\
\end{array}
\end{equation*}
We sum, then divide by $2$, and get:
\begin{equation*}
\curv XY(Z,T) = \curv ZT(X,Y).
\end{equation*}
\end{proof}

\begin{corollary}[see \citep{Kup87b}{270}]
\label{thm_curvature_tensor_radical}
Let $(M,g)$ be a {\rstationary} manifold. Then, for any $X,Y,Z\in \fivect{M}$ and $W\in\fivectnull{M}$,
\begin{equation}
	R(W,X,Y,Z) = R(X,W,Y,Z) = R(X,Y,W,Z) = R(X,Y,Z,W) = 0.
\end{equation}
\end{corollary}
\begin{proof}
Since, for any $X,Y,Z\in\fivect{M}$, $\dera X {\ldera Y}Z \in \annihforms M$ (Remark \ref{rem_semi_regular_semi_riemannian}), and $\lderb X Y\in\annihforms M$ (Remark \ref{rem_rad_stat_lower_der}), we get $R(X,Y,Z,W)=0$. The other identities follow from the symmetry properties \eqref{thm_curv_symm_xy} and \eqref{thm_curv_symm_xy_zt} from Theorem \ref{thm_curv_symm}.
\end{proof}

\subsection{Ricci curvature tensor and scalar curvature}
\label{s_ricci_tensor_scalar}

\begin{definition}
\label{def_ricci_curvature_tensor}
Let $(M,g)$ be a {\rstationary} {\semiriem} manifold. Then,
\begin{equation}
	\ric(X,Y):=R(X,\cocontr,Y,\cocontr),
\end{equation}
for any $X,Y\in\fivect{M}$, is called the {\em Ricci curvature tensor}. If the {\metricname} has constant signature, the Ricci tensor is smooth.
\end{definition}

\begin{proposition}
The Ricci curvature tensor on a {\rstationary} {\semiriem} manifold with constant signature is symmetric:
\begin{equation}
	\ric(X,Y)=\ric(Y,X)
\end{equation}
for any $X,Y\in\fivect{M}$.
\end{proposition}
\begin{proof}
From Proposition \ref{thm_curv_symm}, for any $X,Y,Z,T\in\fivect{M}$, $R(X,Y,Z,T)=R(Z,T,X,Y)$. Hence, $\ric(X,Y)=\ric(Y,X)$ (like in the {\nondeg} case \cfeg{ONe83}{87}).
\end{proof}

\begin{definition}
\label{def_scalar_curvature}
Let $(M,g)$ be a {\rstationary} {\semiriem} manifold. Then,
\begin{equation}
	s:=\ric(\cocontr,\cocontr)
\end{equation}
is called the {\em scalar curvature}. It is smooth if the {\metricname} has constant signature.
\end{definition}

\section{Relation with Kupeli's curvature function}
\label{s_riemann_curvature_ii}

Subsection \sref{s_riemann_curvature_koszul_formula} contains a useful formula for the Riemann curvature, in terms of the {\koszulname}.

Subsection \sref{s_koszul_deriv_curv_funct} contains a comparison of our Riemann curvature, and Kupeli's curvature function associated to the (non-unique) Koszul derivative $\der$ \cite{Kup87b}. We show that his curvature coincides to our Riemann curvature tensor, introduced in an invariant way in \sref{s_riemann_curvature}.

\subsection{Riemann curvature in terms of the {\koszulname}}
\label{s_riemann_curvature_koszul_formula}

\begin{proposition}
\label{thm_riemann_curvature_tensor_koszul_formula}
Let $(M,g)$ be a {\semireg} {\semiriem} manifold. For any vector fields $X,Y,Z,T\in\fivect{M}$,
\begin{equation}
\begin{array}{lll}
	R(X,Y,Z,T) &=& X\left(\lderc Y Z T\right) - Y\left(\lderc X Z T\right) - \lderc {[X,Y]}ZT \\
&& + \annihprod{\lderb XZ,\lderb Y T} - \annihprod{\lderb YZ,\lderb X T}. \\
\end{array}
\end{equation}
Equivalently,
\begin{equation}
\label{eq_riemann_curvature_tensor_koszul_formula}
\begin{array}{lll}
	R(X,Y,Z,T)&=& X \kosz(Y,Z,T) - Y \kosz(X,Z,T) - \kosz([X,Y],Z,T)\\
	&& + \kosz(X,Z,\cocontr)\kosz(Y,T,\cocontr) - \kosz(Y,Z,\cocontr)\kosz(X,T,\cocontr).
\end{array}
\end{equation}
\end{proposition}
\begin{proof}
From Definition \ref{def_cov_der_covect} we obtain
\begin{equation}
	\derc X {\lderb YZ} T  = X\left(\lderc Y Z T\right) - \annihprod{\lderb X T,\lderb YZ},
\end{equation}
hence, for any vector fields $X,Y,Z,T\in\fivect{M}$
\begin{equation}
\begin{array}{lll}
	R(X,Y,Z,T) &=& \derc X {{\ldera Y}Z}T - \derc Y {{\ldera X}Z}T -
 \lderc {[X,Y]}ZT \\
&=& X\left(\lderc Y Z T\right) - Y\left(\lderc X Z T\right) - \lderc {[X,Y]}ZT \\
&& + \annihprod{\lderb XZ,\lderb Y T} - \annihprod{\lderb YZ,\lderb X T}. \\
\end{array}
\end{equation}
The formula \eqref{eq_riemann_curvature_tensor_koszul_formula} follows from Definition \ref{def_l_cov_der}.
\end{proof}

\begin{remark}
In a coordinate basis, the components of Riemann's curvature tensor are
\begin{equation}
\label{eq_riemann_curvature_tensor_coord}
	R_{abcd}= \partial_a \kosz_{bcd} - \partial_b \kosz_{acd} + \annihg^{st}(\kosz_{acs}\kosz_{bdt} - \kosz_{bcs}\kosz_{adt}).
\end{equation}\end{remark}
\begin{proof}
\begin{equation}
\begin{array}{lll}
	R(\partial_a,\partial_b,\partial_c,\partial_d)
	&=& \partial_a \kosz(\partial_b,\partial_c,\partial_d) - \partial_b \kosz(\partial_a,\partial_c,\partial_d) - \kosz([\partial_a,\partial_b],\partial_c,\partial_d)\\
	&& + \kosz(\partial_a,\partial_c,\cocontr)\kosz(\partial_b,\partial_d,\cocontr) - \kosz(\partial_b,\partial_c,\cocontr)\kosz(\partial_a,\partial_d,\cocontr)\\
	&=& \partial_a \kosz_{bcd} - \partial_b \kosz_{acd} + \annihg^{st}(\kosz_{acs}\kosz_{bdt} - \kosz_{bcs}\kosz_{adt})
\end{array}
\end{equation}
\end{proof}

\subsection{Relation with Kupeli's curvature function}
\label{s_koszul_deriv_curv_funct}

Demir Kupeli  showed that for a {\rstationary} {\semiriem} manifold with constant signature $(M,g)$, there is always a {\em Koszul derivative} $\der$\cite{Kup87b}. From its curvature function $R_\der$ one can construct a tensor field $\metric{R_\der(\_,\_)\_,\_}$. We will see that, for a {\rstationary} {\semiriem} manifold, $\metric{R_\der(\_,\_)\_,\_}$ coincides to the Riemann curvature tensor from Definition \ref{def_riemann_curvature}.

\begin{definition}[Koszul derivative, {\cf} Kupeli \citep{Kup87b}{261}]
\label{def_Koszul_derivative}
Let $(M,g)$ be a {\rstationary} {\semiriem} manifold with constant signature.
A {\em Koszul derivative} on $(M,g)$ is an operator $\der:\fivect M\times\fivect M\to \fivect M$ which satisfies
\begin{equation}
\label{eq_Koszul_formula}
\begin{array}{llll}
	\metric{\der_X Y,Z} &=& \kosz(X,Y,Z).
\end{array}
\end{equation}
\end{definition}

\begin{remark}[{\cf} Kupeli \citep{Kup87b}{262}]
The Koszul derivative corresponds, for the {\nondeg} case, to the Levi-Civita connection. If $g$ is degenerate, the Koszul derivative is not unique.
\end{remark}

\begin{definition}[Curvature function, {\cf} Kupeli \citep{Kup87b}{266}]
\label{def_curvature_function}
Let $\nabla$ be a Koszul derivative on a {\rstationary} {\semiriem} manifold $(M,g)$ with constant signature.
Then, the map $R_\der : \fivect M\times\fivect M\times\fivect M\to \fivect M$, defined by
\begin{equation}
\label{eq_curvature_function}
R_\der(X,Y)Z:=\der_X\der_Y Z - \der_Y \der_X Z - \der_{[X,Y]}Z
\end{equation}
is called the {\em curvature function} of $\der$.
\end{definition}

\begin{remark}
In \citep{Kup87b}{266-268}, it is shown that $\metric{R_\der(\_,\_)\_,\_}\in\tensors 0 4 {M}$ and it has the same symmetry properties as the Riemann curvature tensor of a Levi-Civita connection.
\end{remark}

\begin{theorem}
Let $(M,g)$ be a {\rstationary} {\semiriem} manifold with constant signature, and $\der$ a Koszul derivative on $M$. The Riemann curvature tensor is related to the curvature function, for any $X,Y,Z,T\in\fivect{M}$, by
\begin{equation}
	\metric{R_\der(X,Y)Z,T} = R(X,Y,Z,T).
\end{equation}
\end{theorem}
\begin{proof}
We apply Theorem \ref{thm_Koszul_form_props}, Definition \ref{def_curvature_function}, Lemma \ref{thm_contraction_with_metric}, and the Koszul formula for the Riemann curvature tensor \eqref{eq_riemann_curvature_tensor_koszul_formula}, and we obtain
\begin{equation*}
\begin{array}{lll}
\metric{R_\der(X,Y)Z,T} &=&\metric{\der_X \der_Y Z,T} - \metric{\der_Y \der_X Z,T} - \metric{\der_{[X,Y]}Z,T}\\
&=&X \metric{\der_YZ,T} - \metric{\der_Y Z, \der_X T} \\
	&& - Y \metric{\der_X Z,T} + \metric{\der_X Z, \der_Y T} - \metric{\der_{[X,Y]} Z,T} \\
	&=& X \kosz(Y,Z,T) - \kosz(Y,Z,\cocontr)\kosz(X,T,\cocontr) \\
	&& - Y \kosz(X,Z,T) + \kosz(X,Z,\cocontr)\kosz(Y,T,\cocontr)\\
	&& - \kosz([X,Y],Z,T) \\
&=& R(X,Y,Z,T).
\end{array}
\end{equation*}
\end{proof}

\section{Examples of {\semireg} {\semiriem} manifolds}
\label{s_semi_reg_semi_riem_man_example}

\subsection{Diagonal metric}
\label{s_semi_reg_semi_riem_man_example_diagonal}

Let $(M,g)$ be a singular {\semiriem} manifold. We assume that, around each point $p\in M$, there is a local coordinate system in which the metric is diagonal, $g=\diag{(g_{11},\ldots,g_{nn})}$. From equation \eqref{eq_Koszul_form_coord}, $2\kosz_{abc}=\partial_a g_{bc} + \partial_b g_{ca} - \partial_c g_{ab}$. Because $g$ is diagonal, remain only the possibilities $\kosz_{baa} = \kosz_{aba} = -\kosz_{aab} = \frac 1 2\partial_b g_{aa}$, for $a\neq b$, and $\kosz_{aaa} = \frac 1 2\partial_a g_{aa}$. The condition that the manifold $(M,g)$ is {\rstationary} is equivalent to the condition that, if $g_{aa}(q)=0$, $\partial_b g_{aa}(q) = \partial_a g_{bb}(q) = 0$, for any $q\in M$.

By Proposition $\ref{thm_sr_cocontr_kosz}$, $(M,g)$ is a {\semireg} manifold if and only if
\begin{equation}
\label{eq_diag_g_contr_kosz}
	\sum_{\genfrac{}{}{0pt}{}{s\in\{1,\ldots,n\}}{g_{ss}\neq 0}} \dsfrac{\partial_a g_{ss}\partial_b g_{ss}}{g_{ss}},
	\sum_{\genfrac{}{}{0pt}{}{s\in\{1,\ldots,n\}}{g_{ss}\neq 0}} \dsfrac{\partial_s g_{aa}\partial_s g_{bb}}{g_{ss}},
	\sum_{\genfrac{}{}{0pt}{}{s\in\{1,\ldots,n\}}{g_{ss}\neq 0}} \dsfrac{\partial_a g_{ss}\partial_s g_{bb}}{g_{ss}}
\end{equation}
are smooth. It is easy to check that this can be ensured for example if $\sqrt{\abs{g_{aa}}}$ and the functions $u,v:M\to\R$ defined as
\begin{equation}
	u(p):=\left\{
\begin{array}{ll}
\dsfrac{\partial_b g_{aa}}{\sqrt{\abs{g_{aa}}}} & g_{aa}\neq 0 \\
0 & g_{aa}= 0 \\
\end{array}
\right.
\tn{ and }
	v(p):=\left\{
\begin{array}{ll}
\dsfrac{\partial_a g_{bb}}{\sqrt{\abs{g_{aa}}}} & g_{aa}\neq 0 \\
0 & g_{aa}= 0 \\
\end{array}
\right.
\end{equation}
are smooth for all $a,b\in\{1,\ldots,n\}$.

\movedfrom{1105.3404-body}

If the metric has the form $g=\sum_a\varepsilon_a\alpha_a^2\de x^a\otimes \de x^a$, where $\varepsilon_a\in\{-1,1\}$, then $g$ is {\semireg} if there is a function $f_{abc}\in\fiscal{M}$ with $\supp{f_{abc}}\subseteq\supp{\alpha_c}$ for any $c\in\{a,b\}\subset\{1,\ldots,n\}$, and
\begin{equation}
\label{eq_diagonal_metric:semireg}
\partial_a\alpha_b^2=f_{abc}\alpha_c.
\end{equation}
If $c=b$, from $\partial_a\alpha_b^2=2\alpha_b\partial_a\alpha_b$ follows that the function is  $f_{abb}=2\partial_a\alpha_b$. This has to satisfy, in addition, the condition $\partial_a\alpha_b=0$ whenever $\alpha_b=0$. The condition $\supp{f_{abc}}\subseteq\supp{\alpha_c}$ is required because in order to be {\semireg}, $(M,g)$ has to be {\rstationary}.

\subsection{Isotropic singularities}
\label{s_semi_reg_semi_riem_man_example_conformal}

\begin{definition}
Let $(M,g)$ be a singular {\semiriem} manifold. If there is a {\nondeg} {\semiriem} metric $\tilde g$ on $M$ and a smooth function $\Omega\in\fiscal M$, $\Omega\geq 0$, so that $g(X,Y)=\Omega^2\tilde g(X,Y)$ for any $X,Y\in\fivect M$, $(M,g)$ is said to be {\em conformally {\nondeg}}, and is alternatively denoted by $(M,\tilde g, \Omega)$.
\end{definition}

\begin{proposition}
\label{thm_conformal_koszul_form}[Generalizing a proposition \cfeg{HE95}{42} to the degenerate case]
Let $(M,\tilde g, \Omega)$ be a conformally {\nondeg} singular {\semiriem} manifold with the {\koszulname} of $g=\Omega^2\tilde g$ denoted by $\kosz$, and that of $\tilde g$ by $\tilde \kosz$. Then,
\begin{equation}
\label{eq_conformal_koszul_form}
\kosz(X,Y,Z) = \Omega^2\tilde \kosz(X,Y,Z) + \Omega\left[\tilde g(Y,Z)X + \tilde g(X,Z)Y - \tilde g(X,Y)Z\right](\Omega)
\end{equation}
\end{proposition}
\begin{proof}
Follows from the Koszul formula,
\begin{equation*}
\begin{array}{llll}
	\kosz(X,Y,Z) &=&\ds{\frac 1 2} \{ \Omega^2X (\tilde g(Y,Z)) + \tilde g(Y,Z)X(\Omega^2) + \Omega^2Y (\tilde g(X,Z)) \\
	&& + \tilde g(X,Z)Y(\Omega^2) - \Omega^2Z (\tilde g(X,Y)) - \tilde g(X,Y)Z(\Omega^2) \\
	&&\ - \Omega^2\tilde g(X,[Y,Z]) + \Omega^2\tilde g(Y, [Z,X]) + \Omega^2\tilde g(Z, [X,Y])\} \\
	&=& \Omega^2 \tilde \kosz(X,Y,Z) + \ds{\frac 1 2} \{ \tilde g(Y,Z)X(\Omega^2) \\
	&& + \tilde g(X,Z)Y(\Omega^2) - \tilde g(X,Y)Z(\Omega^2)\} \\
	&=& \Omega^2\tilde \kosz(X,Y,Z) + \Omega\big[\tilde g(Y,Z)X \\
	&& + \tilde g(X,Z)Y - \tilde g(X,Y)Z\big](\Omega)
\end{array}
\end{equation*}
\end{proof}

\begin{theorem}
\label{thm_conformal_semi_regular}
Let $(M,\tilde g, \Omega)$ be a conformally {\nondeg} singular {\semiriem} manifold. Then, $(M,g=\Omega^2\tilde g)$ is a {\semireg} {\semiriem} manifold.
\end{theorem}
\begin{proof}
At any point, the metric $g$ is either {\nondeg}, or $0$. Hence, $(M,g)$ is {\rstationary}.

Let $(E_a)_{a=1}^n$ be a local frame of vector fields orthonormal with respect to the {\nondeg} metric $\tilde g$, on an open $U\subseteq M$. Then, $g$ is diagonal in $(E_a)_{a=1}^n$. 

From proposition \ref{thm_conformal_koszul_form}, the {\koszulname} is of the form $\kosz(X,Y,Z) = \Omega h(X,Y,Z)$, where
\begin{equation}
	h(X,Y,Z) = \Omega \tilde \kosz(X,Y,Z) + \left[\tilde g(Y,Z)X + \tilde g(X,Z)Y - \tilde g(X,Y)Z\right](\Omega)
\end{equation}
is a smooth function of $X,Y,Z$. If $\Omega=0$, then $h(X,Y,Z)=0$, because $\Omega \tilde \kosz(X,Y,Z)=0$, and the other term is a sum of partial derivatives of $\Omega$, which vanish when $\Omega=0$, being a minimum.

From Theorem \ref{thm_contraction_orthogonal}, for any $X,Y,Z,T\in U$, on regions of constant signature,
\begin{equation}
\begin{array}{lll}
\kosz(X,Y,\cocontr)\kosz(Z,T,\cocontr)
&=& \sum_{a=r}^n \dsfrac{\kosz(X,Y,E_a)\kosz(Z,T,E_a)}{g(E_a,E_a)} \\
&=& \sum_{a=r}^n \dsfrac{\Omega^2 h(X,Y,E_a) h(Z,T,E_a)}{\Omega^2\tilde g(E_a,E_a)} \\
&=& \sum_{a=1}^n \dsfrac{h(X,Y,E_a) h(Z,T,E_a)}{\tilde g(E_a,E_a)}, \\
\end{array}
\end{equation}
where $r=n-\rank g+1$.
If $\Omega=0$, then $h(X,Y,Z)=0$, that's why the last member is independent on  $r$. Hence, $\kosz(X,Y,\cocontr)\kosz(Z,T,\cocontr)\in\fiscal M$, and Proposition \ref{thm_sr_cocontr_kosz} saids that $(M,g)$ is {\semireg}.
\end{proof}

\section{Cartan's structural equations for degenerate metric}
\label{s_cartan}

In {\semiriem} geometry (with {\nondeg} metric), there is an important relation between a connection and its curvature, in terms of the moving frames, captured in a compact way in Cartan's structural equations. 
Cartan's first structural equation expresses, by the means of the connection, the rotation of a moving coframe, due to the displacement in one direction.

But if the fundamental tensor becomes degenerate, we have to avoid the metric connection and its curvature operator, and the local orthonormal frames and coframes, which no longer exist.

In this section we show that we can construct in a canonical way geometric objects similar to Cartan's connection and curvature forms, based only on the metric. We obtain structure equations similar to those of Cartan, which are identical to them if the metric is {\nondeg}. Along the way to this goal, we will obtain a compact version of the Koszul formula.

In \sref{s_cartan_structure_i}, the connection forms are introduced, and from them is derived the first structural equation for {\rstationary} manifolds. In \sref{s_cartan_structure_ii}, the curvature forms are defined, and the second structural equation for {\rstationary} manifolds is obtained.

\subsection{The first structural equation}
\label{s_cartan_structure_i}

\subsubsection{The decomposition of the {\koszulname}}
\label{s_koszul_form_compact}

\begin{lemma}
\label{thm_Koszul_form_compact}
Let $(M,g)$ be a singular {\semiriem} manifold.
The {\koszulname} \eqref{eq_Koszul_form}
can be decomposed as
\begin{equation}
\label{eq_Koszul_form_compact}
	2\kosz(X,Y,Z) = (\de Y^{\flat})(X, Z) + (\lie_Y g)(X,Z).
\end{equation}
\end{lemma}
\begin{proof}
Follows immediately from the Lie derivative of $g$, and from the exterior derivative of $Y^\flat$,
\begin{equation*}
\begin{array}{lll}
(\de Y^{\flat})(X, Z) &=& X\(Y^{\flat}(Z)\) - Z\(Y^{\flat}(X)\) - Y^{\flat}([X,Z]) \\
&=&  X\metric{Y,Z} - Z\metric{X,Y} + \metric{Y,[Z,X]}. \\
\end{array}
\end{equation*}
\end{proof}

\begin{corollary} The following property of the {\koszulname} \eqref{eq_Koszul_form} holds
\begin{equation}
(\de Y^{\flat})(X, Z) =	\kosz(X,Y,Z) - \kosz(Z,Y,X).
\end{equation}
\end{corollary}
\begin{proof}
Follows immediately from Lemma \ref{thm_Koszul_form_compact} and Theorem \ref{thm_Koszul_form_props}.
\end{proof}

\subsubsection{The connection forms}
\label{s_conn_form}

Let $(M,g)$ be a {\semiriem} manifold. If $(E_a)_{a=1}^n$ is a local orthonormal frame on $M$ with respect to $g$, its dual $(\omega^b)_{b=1}^n$, where $\omega^b(E_a)=\delta^b_a$, is orthonormal. 
The $1$-forms $\omega_a{}^b$, $1\leq a,b\leq n$, where
\begin{equation}
\label{eq_connection_forms_nondeg}
	\omega_a{}^b(X) := \omega^b(\nabla_X E_a)
\end{equation}
are named the \textit{connection forms} (\cf \eg \cite{ONe95})~\footnote{Here the indices $a,b$ don't represent the components, they label the connection $1$-forms $\omega_a{}^b$.}.

For a degenerate {\metricname} $g$, the Levi-Civita connection $\nabla_X$ with respect to $g$, and hence $\nabla_X E_a$, don't exist. Moreover, a frame $(E_a)_{a=1}^n$ cannot be orthonormal (with respect to $g$), but it can be orthogonal. But even so, the dual frame $(\omega^b)_{b=1}^n$ cannot be orthogonal, because the {\metricname} $\annihg(\omega,\tau)$ is defined only for $\annih{T}M$, and not for the full $T^*M$. We will see here a way to define connection $1$-forms for the degenerate case.

\begin{definition}
\label{def_connection_form}
Let $(M,g)$ be a singular manifold, and $X,Y\in\fivect{M}$. The $1$-form
\begin{equation}
\label{eq_connection_forms}
	\omega_{XY}(Z) := \kosz(Z,X,Y)
\end{equation}
is named the \textit{connection form} associated to the metric $g$ and the vector fields $X,Y$.
We also define $\omega_{ab}$ by
\begin{equation}
\label{eq_connection_forms_frame}
	\omega_{ab}(Z) := \omega_{E_a E_b}(Z).
\end{equation}
\end{definition}

\begin{remark}
It is easy to see that $\omega_{XY}$ is $1$-form, because ${\kosz}$ is linear, and $\fiscal M$-linear in the first argument.
\end{remark}

\subsubsection{The first structural equation}
\label{ss_cartan_structure_i}

Let $(M,g)$ be a {\rstationary} manifold.

\begin{lemma}
\label{thm_cartan_structure_i}
The following equation holds
\begin{equation}
\label{eq_cartan_structure_i}
	\de X^{\flat} = \omega_{X\cocontr} \wedge \cocontr^{\flat},
\end{equation}
where $\omega_{X\cocontr} \wedge \cocontr^{\flat}$ is the metric contraction of $\omega_{XY}\wedge Z^{\flat}$ in $Y,Z$. Equation \eqref{eq_cartan_structure_i} is called \textit{the first structural equation} determined by the metric $g$.
\end{lemma}
\begin{proof}
From the formula \eqref{eq_Koszul_form} and the Lemma \ref{thm_Koszul_form_compact},
\begin{equation}
\label{eq_cartan_structure_i_kosz}
	(\de X^{\flat})(Y,Z) = \kosz(Y,X,Z) - \kosz(Z,X,Y).
\end{equation}
By substituting \eqref{eq_connection_forms_nondeg},
\begin{equation}
\label{eq_cartan_structure_i_expanded}
	(\de X^{\flat})(Y,Z) = \omega_{XZ}(Y) - \omega_{XY}(Z).
\end{equation}
The properties of the metric contraction, and the property of $(M,g)$ of being {\rstationary}, allow us to write
\begin{equation}
\begin{array}{lll}
\omega_{YZ}(X) &=& \omega_{Y\cocontr}(X)\metric{\cocontr,Z} = \omega_{Y\cocontr}(X)\(\cocontr^{\flat}(Z)\) \\
&=& \(\omega_{Y\cocontr}\otimes\cocontr^{\flat}\)(X,Z).
\end{array}
\end{equation}
Hence, equation \eqref{eq_cartan_structure_i_expanded} writes
\begin{equation}
\begin{array}{lll}
	(\de X^{\flat})(Y,Z) &=& \(\omega_{X\cocontr}\otimes\cocontr^{\flat}\)(Y,Z) - \(\omega_{X\cocontr}\otimes\cocontr^{\flat}\)(Z,Y) \\
	&=& \(\omega_{X\cocontr} \wedge \cocontr^{\flat}\)(Y,Z).
\end{array}
\end{equation}
\end{proof}

\begin{corollary}
\label{thm_cartan_structure_i_std}
Let $(M,g)$ be a {\semiriem} manifold, $(E_a)_{a=1}^n$ an orthonormal frame, and $(\omega^a)_{a=1}^n$ its dual. Then
\begin{equation}
\label{eq_cartan_structure_i_std}
	\de \omega^a = -\omega_s{}^a \wedge \omega^s.
\end{equation}
\end{corollary}
\begin{proof}
From Theorem \ref{thm_Koszul_form_props}:\eqref{thm_Koszul_form_props_commutYZ}, follows that
\begin{equation}
	\omega_{E_a E_b}(X) + \omega_{E_b E_a}(X) = X\metric{E_a,E_b} = X(\delta_{ab}) = 0,
\end{equation}
and hence
\begin{equation}
	\omega_{E_a E_b} = -\omega_{E_b E_a}.
\end{equation}
From equation \eqref{eq_cartan_structure_i},
\begin{equation}
	\de E_a^{\flat} = \omega_{E_a E_s} \wedge \omega^s.
\end{equation}
Equation \eqref{eq_cartan_structure_i_std} follows from $\omega_{E_a E_s} =-\omega_{E_s E_a}$ and $\omega^a=E_a^{\flat}$.
\end{proof}

\begin{remark}
At points where the metric changes its signature, continuity is not ensured, unless the manifold $(M,g)$ is {\semireg}.
\end{remark}

\subsection{The second structural equation}
\label{s_cartan_structure_ii}

\subsubsection{The curvature forms}
\label{s_curvature_form}

\begin{definition}
\label{def_curvature_form}
Let $(M,g)$ be a radical-stationary manifold, and $X,Y\in\fivect{M}$ two vector fields. Then, the $2$-form
\begin{equation}
\label{eq_curvature_form}
	\Omega_{XY}(Z,T) := R(X,Y,Z,T),
\end{equation}
where $Z,T\in\fivect{M}$, is called the \textit{Riemann curvature form} associated to the metric $g$ and the vector fields $X,Y$.
If $(E_a)_{a=1}^n$ is a frame field, we make the notation
\begin{equation}
\label{eq_curvature_form_frame}
	\Omega_{ab}(Z,T) := \Omega_{E_a E_b}(Z,T).
\end{equation}
\end{definition}

\subsubsection{The second structural equation}
\label{ss_cartan_structure_ii}

\begin{lemma}
\label{thm_cartan_structure_ii}
Let $(M,g)$ be a radical-stationary manifold, and $X,Y\in\fivect{M}$ two vector fields. Then, the following equation (which we call \textit{the second structural equation} determined by the metric $g$) holds
\begin{equation}
\label{eq_cartan_structure_ii}
	\Omega_{XY} = \de\omega_{XY} + \omega_{X\cocontr} \wedge \omega_{Y\cocontr}.
\end{equation}
\end{lemma}
\begin{proof}
The exterior derivative of $\omega_{XY}$ is
\begin{equation}
\label{eq_cartan_structure_ii_a}
\begin{array}{lll}
\de\omega_{XY}(Z,T) &=& Z\(\omega_{XY}(T)\) - T\(\omega_{XY}(Z)\) - \omega_{XY}([T,Z]) \\
&=& Z\kosz(T,X,Y) - T\kosz(Z,X,Y) - \kosz([T,Z],X,Y).
\end{array}
\end{equation}
Also, 
\begin{equation}
\label{eq_cartan_structure_ii_b}
\begin{array}{lll}
\(\omega_{X\cocontr} \wedge \omega_{Y\cocontr}\)(Z,T) &=& \omega_{X\cocontr}(Z) \omega_{Y\cocontr}(T) - \omega_{X\cocontr}(T) \omega_{Y\cocontr}(Z) \\
&=& \kosz(Z,X,\cocontr) \kosz(T,Y,\cocontr) - \kosz(T,X,\cocontr) \kosz(Z,Y,\cocontr).
\end{array}
\end{equation}
Equation \eqref{eq_cartan_structure_ii} follows by plugging the identities \eqref{eq_cartan_structure_ii_a} and \eqref{eq_cartan_structure_ii_b} in \eqref{eq_riemann_curvature_tensor_koszul_formula}.
\end{proof}

\section{Degenerate warped products}
\label{s_warped}

A large and important class of {\semireg} singularities is given by degenerate warped products of manifolds, with fundamental tensor allowed to be degenerate. We show that a degenerate warped product of semi-regular manifolds is a semi-regular manifold, provided that the warping function satisfies a certain condition. We express the main geometric objects on the warped product in terms of those of the factor manifolds. We provide examples of {\semireg} manifolds, obtained as as warped products. The techniques developed here will be applied to singularities in General Relativity, in the remaining part of the Thesis.

\subsection{Introduction}

The warped product is used to construct new examples of semi-Riemannian manifolds from known ones \cite{BON69,BEP82,ONe83,pripoae05srg}. In General Relativity, it is used in the study of black holes and cosmological models. When the warping function becomes $0$, the {\metricname} of the product manifold becomes degenerate, singularities appear, and semi-Riemannian geometry can't be applied. Luckily, we will see that the tools developed here can be applied at those singularities.

The degenerate warped products of singular manifolds are defined in \sref{ss_deg_wp_ssr_gen}. We derive simple properties of the {\koszulname} of the warped product in terms of the {\koszulname} of the factors. In \sref{s_semi_reg_semi_riem_man_warped} we prove that the warped products of {\rstationary} manifolds are also {\rstationary}, if the warping function satisfies a certain condition. Then, we show a similar result for {\semireg} manifolds. We express, in \sref{s_riemann_wp_deg}, the curvature of {\semireg} warped products in terms of the factor manifolds.

Before starting, let's recall some notions about the {\em product manifold} $B\times F$ of two differentiable manifolds, $B$ and $F$ \cfeg{ONe83}{24--25}.

Let $p=(p_1,p_2)\in M_1\times M_2$. At $p$, the tangent space decomposes as
\begin{equation}
	T_{(p_1,p_2)}(M_1\times M_2)\cong T_{(p_1,p_2)}(M_1)\oplus T_{(p_1,p_2)}(M_2),
\end{equation}
where $T_{(p_1,p_2)}(M_1):=T_{(p_1,p_2)}(M_1\times p_2)$ and $T_{(p_1,p_2)}(M_2):=T_{(p_1,p_2)}(p_1\times M_2)$.

Let $\pi_i:M_1\times M_2\to M_i$, for $i\in\{1,2\}$, be the canonical projections. The following definitions apply to each $i\in\{1,2\}$. Let $f_i\in\fiscal{M_i}$. The scalar field $\tilde f_i:=f_i\circ\pi_i\in\fivect{M_1\times M_2}$ is called the {\em lift of the scalar field} $f_i$.
Let $X_i\in\fivect{M_i}$ be a vector field. Then, the unique vector field $\tilde X_i$ on $M_1\times M_2$ satisfying $\de \pi_i(\tilde X_i)=X_i$ is called the {\em lift of the vector field}. Let $\fivectlift{M,M_i}$ denote the set of all vector fields $X\in \fivect{M_1\times M_2}$ which are lifts of vector fields $X_i\in \fivect{M_i}$.

\subsection{General properties}
\label{ss_deg_wp_ssr_gen}

\begin{definition}[generalizing \citep{ONe83}{204}]
\label{def_wp}
Let $(B,g_B)$ and $(F,g_F)$ be two singular manifolds, and $f\in\fiscal{B}$ a smooth function. The manifold
\begin{equation}
	B\times_f F:=\big(B\times F, \pi^*_B(g_B) + (f\circ \pi_B)\pi^*_F(g_F)\big),
\end{equation}
where $\pi_B:B\times F \to B$ and $\pi_F: B \times F \to F$ are the canonical projections, is called the {\em warped product} of $B$ and $F$ with {\em warping function} $f$. We call $B$ the {\em base} and $F$ the {\em fiber} of the warped product $B\times_f F$.

For all vector fields $X_B,Y_B\in\fivect{B}$ and $X_F,Y_F\in\fivect{F}$, we will use the notations $\metric{X_B,Y_B}_B := g_B(X_B,Y_B)$ and $\metric{X_F,Y_F}_F := g_F(X_F,Y_F)$.
For any point $p\in B\times F$ and for any pair of tangent vectors $x,y\in T_p(B\times F)$, the metric is
\begin{equation}
\label{eq_wp_metric}
	\metric{x,y}=\metric{\de \pi_B(x),\de \pi_B(y)}_B + f^2(p)\metric{\de \pi_F(x),\de \pi_F(y)}_F,
\end{equation}
or
\begin{equation}
\de s_{B\times F}^2 = \de s_B^2 + f^2\de s_F^2.
\end{equation}
\end{definition}

\begin{remark}
The Definition \ref{def_wp} generalizes the usual warped product definition, given for the case when both $g_B$ and $g_F$ are {\nondeg} and $f>0$ (see \cite{BON69}, \cite{BEP82} and \cite{ONe83}), to the case of singular {\semiriem} manifolds. We import some terms from \citep{ONe83}{204--205}. If $p_B\in B$, the {\semiriem} manifold $\pi_B^{-1}(p_B)=p_B\times F$ is named the {\em fiber} through $p_B$. If $p_F\in F$, the {\semiriem} manifold $\pi_F^{-1}(p_F)=B\times p_F$ is called the {\em leave} through $p_F$. $\pi_B|_{B\times p_F}$ is an isometry onto $B$, and, if $f(p_B)\neq 0$, $\pi_F|_{p_B\times F}$ is a homothety onto $F$. For any $(p_B,p_F)\in B\times F$, $B \times p_F$ and $p_B \times F$ are orthogonal at $(p_B,p_F)$. For simplicity, we will use sometimes the same notation for the vector and its lift, if they can be distinguished from the context. For example, $\metric{V,W}_F:=\metric{\pi_F(V),\pi_F(W)}_F$ for $V,W\in\fivectlift{B \times F,F}$.
\end{remark}

Some of the properties of the Levi-Civita connection for the warped product of ({\nondeg}) {\semiriem} manifolds \cfeg{ONe83}{206} can be rewritten in terms of the {\koszulname}, being thus extended to the degenerate case. We need this, because the Levi-Civita connection is not defined for degenerate {\metricname}.
\begin{proposition}
\label{thm_wp_deg_koszul}
Let $B \times_f F$ be a warped product, let $X,Y,Z\in\fivectlift{B \times F,B}$, and $U,V,W\in\fivectlift{B \times F,F}$. The {\koszulname} $\kosz$ on $B \times_f F$ can be expressed, in terms of the lifts $\kosz_B, \kosz_F$ of the {\koszulname}s on $B$, respectively $F$, by the following formulae
\begin{enumerate}
	\item \label{thm_wp_deg_koszul:BBB}
	$\kosz(X,Y,Z)=\kosz_B(X,Y,Z)$.
	\item \label{thm_wp_deg_koszul:BBF}
	$\kosz(X,Y,W) = \kosz(X,W,Y) = \kosz(W,X,Y) = 0$.
	\item \label{thm_wp_deg_koszul:BFF}
	$\kosz(X,V,W) = \kosz(V,X,W) = -\kosz(V,W,X) = f \metric{V,W}_F X(f)$.
	\item \label{thm_wp_deg_koszul:FFF}
	$\kosz(U,V,W)=f^2\kosz_F(U,V,W)$.
\end{enumerate}
\end{proposition}
\begin{proof}
For $X,Y,Z\in\fivectlift{B \times F,B}$ and $U,V,W\in\fivectlift{B \times F,F}$, it is easy to check that
$\metric{X,V}=0$, $[X,V]=0$, $V\metric{X,Y}=0$, and $X\metric{V,W}=2f\metric{V,W}_F X(f)$ (similar to \cite{ONe83}).

The properties \eqref{thm_wp_deg_koszul:BBB} and \eqref{thm_wp_deg_koszul:FFF} follow immediately from the properties of the lifts of vector fields, Definition \ref{def_Koszul_form}, and equation \eqref{eq_wp_metric}.

\eqref{thm_wp_deg_koszul:BBF}
From $\metric{Y,W}=\metric{W,X}=\metric{W,[X,Y]}=0$, $[Y,W]=[W,X]=0$, and $W \metric{X,Y}=0$, follows that $\kosz(X,Y,W)=0$.

Since $\kosz(X,W,Y) = X \metric{W,Y} - \kosz(X,Y,W)$, follows that $\kosz(X,W,Y) = 0$.

Because $\kosz(W,X,Y) = \kosz(X,W,Y) - \metric{[X,W],Y}$, we have $\kosz(W,X,Y) = 0$.

\eqref{thm_wp_deg_koszul:BFF} Property \eqref{thm_wp_deg_koszul:BBF} implies that $\kosz(X,V,W) = \dsfrac 1 2 X \metric{V,W} = f \metric{V,W}_F X(f)$. From $\kosz(V,X,W)=\kosz(X,V,W)-\metric{[X,V],W}$ and $[X,V]=0$, follows that $\kosz(V,X,W) = f \metric{V,W}_F X(f)$.

It is simple to see that $\kosz(V,W,X) = V\metric{W,X} - \kosz(V,X,W)$, but since $\metric{W,X}=0$, it follows that $\kosz(V,W,X) = -f\metric{V,W}_F X(f)$.
\end{proof}

In the following, we study some properties of the degenerate warped products, while allowing the warping function $f$ to vanish or change sign, and $(B,g_B)$ and $(F,g_F)$ to be singular and with variable signature.

\subsection{Warped products of {\semireg} manifolds}
\label{s_semi_reg_semi_riem_man_warped}

\begin{theorem}
\label{thm_rad_stat_semi_riem_man_warped}
Let $(B,g_B)$ and $(F,g_F)$ be {\rstationary} manifolds, and $f\in\fiscal{B}$ so that $\de f\in\annihforms B$. Then, the degenerate warped product $B \times_f F$ is a {\rstationary} manifold.
\end{theorem}
\begin{proof}
We need to prove that $\kosz(X,Y,W)=0$, for any $X,Y\in\fivect{B \times_f F}$ and $W\in\fivectnull{B \times_f F}$. It is enough to do this for the lifts of vector fields $X_B,Y_B,W_B\in \fivectlift{B \times F,B}$, $X_F,Y_F,W_F\in \fivectlift{B \times F,F}$, where $W_B,W_F\in\fivectnull{B \times_f F}$.

Proposition \ref{thm_wp_deg_koszul} implies
	
	\noindent$\bullet$ $\kosz(X_B,Y_B,W_B)=\kosz_B(X_B,Y_B,W_B)=0$,
	
	\noindent$\bullet$ $\kosz(X_B,Y_B,W_F) = \kosz(X_B,Y_F,W_B) = \kosz(X_F,Y_B,W_B) = 0$,
	
	\noindent$\bullet$ $\kosz(X_B,Y_F,W_F) = \kosz(Y_F,X_B,W_F) = f \metric{Y_F,W_F}_F X_B(f) = 0$, since $\metric{Y_F,W_F}_F=0$,
	
	\noindent$\bullet$ $\kosz(X_F,Y_F,W_B) = -f \metric{X_F,Y_F}_F W_B(f) = 0$, from $W_B(f)=0$,
	
	\noindent$\bullet$ $\kosz(X_F,Y_F,W_F)=f^2\kosz_F(X_F,Y_F,W_F) = 0$.
\end{proof}

\begin{theorem}
\label{thm_semi_reg_semi_riem_man_warped}
Let $(B,g_B)$ and $(F,g_F)$ be two {\semireg} manifolds, and let $f\in\fiscal{B}$ be a smooth function so that $\de f\in\srformsk 1 B$. Then, the warped product $B \times_f F$ is a {\semireg} manifold.
\end{theorem}
\begin{proof}
By Theorem \ref{thm_rad_stat_semi_riem_man_warped}, all contractions of the form $\kosz(X,Y,\cocontr)\kosz(Z,T,\cocontr)$ are well defined. 
Proposition \ref{thm_sr_cocontr_kosz} implies that it is enough to show that they are smooth.
It is enough to prove it for lifts of vector fields $X_B,Y_B,Z_B,T_B\in \fivectlift{B \times F,B}$,
$X_F,Y_F,Z_F,T_F\in \fivectlift{B \times F,F}$.
Let $\cocontr_B$ and $\cocontr_F$ denote the symbol for the contraction on $B$, respectively $F$.
Then, Proposition \ref{thm_wp_deg_koszul} implies:
\begin{equation*}
\begin{array}{lll}
\kosz(X_B,Y_B,\cocontr)\kosz(Z_B,T_B,\cocontr)&=&\kosz(X_B,Y_B,\cocontr_B)\kosz(Z_B,T_B,\cocontr_B) \\
&&+\kosz(X_B,Y_B,\cocontr_F)\kosz(Z_B,T_B,\cocontr_F) \\
&=&\kosz_B(X_B,Y_B,\cocontr_B)\kosz_B(Z_B,T_B,\cocontr_B) \in \fiscal{B \times_f F}, \\
\kosz(X_B,Y_B,\cocontr)\kosz(Z_F,T_B,\cocontr)&=&\kosz(X_B,Y_B,\cocontr)\kosz(Z_B,T_F,\cocontr)\\
&=&\kosz(X_B,Y_B,\cocontr_B)\kosz(Z_B,T_F,\cocontr_B) \\
&&+\kosz(X_B,Y_B,\cocontr_F)\kosz(Z_B,T_F,\cocontr_F) = 0. \\
\end{array}
\end{equation*}
The other cases are obtained similarly.
\end{proof}

\begin{corollary}
\label{thm_reg_warped}
Let $(B,g_B)$ be a {\nondeg} manifold, and $f\in\fiscal{B}$. If $(F,g_F)$ is a {\rstationary} ({\semireg}) manifold, then $B\times_f F$ is a {\rstationary} ({\semireg}) manifold. In particular, if $(B,g_B)$ and $(F,g_F)$ are {\nondeg}, and $f\in\fiscal{B}$, then $B\times_f F$ is {\semireg}.
\end{corollary}
\begin{proof}
If $(B,g_B)$ is {\nondeg}, any function $f\in\fiscal{B}$ satisfies $\de f\in\annihforms B$ and $\de f\in\srformsk 1 B$. Theorems \ref{thm_rad_stat_semi_riem_man_warped} and \ref{thm_semi_reg_semi_riem_man_warped} imply the desired result.
\end{proof}

The warped product of ({\nondeg}) {\semiriem} manifolds stays {\nondeg} for $f>0$. If $f\to 0$, the known formulate show that the connection $\nabla$ (\citep{ONe83}{206--207}), the Riemann curvature $R_\nabla$ (\citep{ONe83}{209--210}), the Ricci tensor $\ric$ and the scalar curvature $s$ (\citep{ONe83}{211}) diverge in general.

\subsection{Riemann curvature of {\semireg} warped products}
\label{s_riemann_wp_deg}

Let $(B,g_B)$ and $(F,g_F)$ be two {\semireg} manifolds, $f\in\fiscal{B}$ a smooth function so that $\de f\in\srformsk 1 B$, and $B\times_f F$ the warped product of $B$ and $F$.
In this section we will find the relation between the Riemann curvature of the warped product in terms of that on the factors. To work in the degenerate case, the resulting formulae have to be in terms of the curvature tensor $R_{abcd}$. They became, in the {\nondeg} case, similar to those in \citepcf{ONe83}{210--211} for the curvature operator $R^a{}_{bcd}$..

\begin{definition}
\label{def_hessian}
Let $(M,g)$ be a {\semireg} manifold, and $f\in\fiscal{B}$, so that $\de f\in\metricformsk 1 M$. 
The smooth tensor field $H^f\in\tensors 0 2{M}$ defined by $H^f(X,Y) := \left(\der_X\de f\right)(Y)$, for any $X,Y\in\fivect{M}$, is called the {\em Hessian} of $f$.
\end{definition}

\begin{theorem}
\label{thm_wp_nondeg_riemm_tens}
Let $(B,g_B)$ and $(F,g_F)$ be two {\semireg} manifolds, $f\in\fiscal{B}$ a smooth function so that $\de f\in\srformsk 1 B$. Let $R_B, R_F$ be the lifts of the curvature tensors of $B$ and $F$. Let $X,Y,Z,T\in\fivectlift{B \times F,B}$, $U,V,W,Q\in\fivectlift{B \times F,F}$, and let $H^f$ be the {\em Hessian} of $f$. Then:
\begin{enumerate}
	\item
	$R(X,Y,Z,T) = R_B(X,Y,Z,T)$,
	\item
	$R(X,Y,Z,Q) = 0$,
	$R(X,Y,W,Q) = 0$,
	$R(U,V,Z,Q) = 0$,
	\item
	$R(X,V,W,T) = -fH^f(X,T)\metric{V,W}_F$
	\item
	$\begin{aligned}[t]
          R(U,V,W,Q)=&R_F(U,V,W,Q) \\
					& + f^2 \annihprod{\de f,\de f}_B\big(\metric{U,W}_F\metric{V,Q}_F- \metric{V,W}_F\metric{U,Q}_F\big),
       \end{aligned}$
\end{enumerate}
the other cases following from the symmetries of the curvature tensor.
\end{theorem}
\begin{proof}
We will use the formula \eqref{eq_riemann_curvature_tensor_koszul_formula} for the curvature.
Let $\cocontr$ denote the covariant metric contraction on $B \times_f F$, and $\stackrel B{\cocontr}$, $\stackrel F{\cocontr}$, the contractions on $B$, respectively $F$.
From Proposition \ref{thm_wp_deg_koszul}, \eqref{thm_wp_deg_koszul:BBF},
\begin{equation*}
\begin{array}{lll}
	R(X,Y,Z,T)&=&  X \kosz(Y,Z,T) - Y \kosz(X,Z,T) - \kosz([X,Y],Z,T)\\
	&& + \kosz(X,Z,\stackrel B{\cocontr})\kosz(Y,T,\stackrel B{\cocontr}) - \kosz(Y,Z,\stackrel B{\cocontr})\kosz(X,T,\stackrel B{\cocontr}) \\
	&=& R_B(X,Y,Z,T).
	\end{array}
\end{equation*}
Similarly, $R(X,Y,Z,Q)=0$ and $R(X,Y,W,Q) = 0$.
From Proposition \ref{thm_wp_deg_koszul}, \eqref{thm_wp_deg_koszul:BFF} and  \eqref{thm_wp_deg_koszul:FFF}, equation \eqref{eq_Koszul_form}, and since the contraction on $F$ cancels the coefficient $f^2$ of $\kosz(U,V,W)_F$, we obtain
\begin{equation*}
\begin{array}{lll}
	R(U,V,Z,Q)&=& U \left(f\metric{V,Q}_FZ(f)\right) - V \left(f\metric{U,Q}_FZ(f)\right) \\
	&&- f\metric{[U,V],Q}_FZ(f) \\
	&& + \kosz(U,Z,\stackrel B{\cocontr})\kosz(V,Q,\stackrel B{\cocontr}) - \kosz(V,Z,\stackrel B{\cocontr})\kosz(U,Q,\stackrel B{\cocontr}) \\
	&& + \kosz(U,Z,\stackrel F{\cocontr})\kosz(V,Q,\stackrel F{\cocontr}) - \kosz(V,Z,\stackrel F{\cocontr})\kosz(U,Q,\stackrel F{\cocontr}) \\
	&=& f Z(f) \left(U \metric{V,Q}_F - V \metric{U,Q}_F - \metric{[U,V],Q}_F\right) \\
	&& + \kosz(U,Z,\stackrel F{\cocontr})\kosz(V,Q,\stackrel F{\cocontr})_F - \kosz(V,Z,\stackrel F{\cocontr})\kosz(U,Q,\stackrel F{\cocontr})_F \\
	&=& f Z(f) \left(U \metric{V,Q}_F - V \metric{U,Q}_F - \metric{[U,V],Q}_F\right) \\
	&& + f\metric{U,\stackrel F{\cocontr}}_FZ(f) \kosz(V,Q,\stackrel F{\cocontr})_F  -f\metric{V,\stackrel F{\cocontr}}_FZ(f)\kosz(U,Q,\stackrel F{\cocontr})_F \\
	&=& f Z(f) (U \metric{V,Q}_F - V \metric{U,Q}_F - \metric{[U,V],Q}_F ) \\
	&& + \kosz(V,Q,U)_F - \kosz(U,Q,V))_F = 0. \\
	\end{array}
\end{equation*}
From Definition \ref{def_hessian} and Proposition \ref{thm_wp_deg_koszul},
\begin{equation*}
\begin{array}{lll}
	R(X,V,W,T)&=& - X \left(f T(f) \metric{V,W}_F\right) \\
	&& + f\metric{V,W}_F\de f(\cocontr) \kosz(X,T,\stackrel B{\cocontr})_B + X(f) \metric{W,\stackrel F{\cocontr}}_F T(f) \metric{V,\stackrel F{\cocontr}}_F \\
	&=& - X(f) T(f) \metric{V,W}_F - fX(T(f)) \metric{V,W}_F \\
	&& + f\metric{V,W}_F \kosz(X,T,\stackrel B{\cocontr})_B\de f(\stackrel B{\cocontr}) + X(f) T(f) \metric{W,V}_F \\
	&=& f\metric{V,W}_F \left[\kosz(X,T,\stackrel B{\cocontr})_B\de f(\stackrel B{\cocontr}) - X\metric{T,\grad f}_B\right] \\
	&=& -f H^f(X,T)\metric{V,W}_F, \\
	\end{array}
\end{equation*}

\begin{equation*}
\begin{array}{lll}
	R(U,V,W,Q)&=& R_F(U,V,W,Q)\\
	&& + \kosz(U,W,\stackrel B{\cocontr})\kosz(V,Q,\stackrel B{\cocontr}) - \kosz(V,W,\stackrel B{\cocontr})\kosz(U,Q,\stackrel B{\cocontr}) \\		
	&=& R_F(U,V,W,Q) + f^2 \metric{U,W}_F\de f(\stackrel B{\cocontr})\metric{V,Q}_F\de f(\stackrel B{\cocontr}) \\
	&& - f^2 \metric{V,W}_F\de f(\stackrel B{\cocontr})\metric{U,Q}_F\de f(\stackrel B{\cocontr}) \\
	&=& R_F(U,V,W,Q)\\
	&& + f^2 \annihprod{\de f,\de f}_B\big(\metric{U,W}_F\metric{V,Q}_F - \metric{V,W}_F\metric{U,Q}_F\big). \\
\end{array}
\end{equation*}
\end{proof}

\begin{remark}
The Riemann tensor $R^{a}{}_{bcd}$ is divergent when the warping function $f\to 0$, even in the {\nondeg} case (\citep{ONe83}{209--210}). But, since the warped product manifold $B \times_f F$ is {\semireg},  the curvature tensor $R_{abcd}$ is smooth, and Theorem \ref{thm_wp_nondeg_riemm_tens} confirms this.
\end{remark}


\chap{ch_einstein_equation}{Einstein equation at singularities}

This chapter relies on author's original results, communicated in the papers \cite{Sto11a} and \cite{Sto12b}.
We apply the results from chapter \ref{ch_ssr} to  construct two versions of Einstein's equation, which are equivalent to Einstein's if the metric is {\nondeg}, but remain valid even in cases when the metric becomes degenerate.

The first version is constructed in \sref{s_einstein_semireg}, on {\semireg} spacetimes, where the Einstein tensor's density of weight 2 remains smooth even in the presence of {\semireg} singularities. We can thus write a densitized version of Einstein's equation, which is smooth, and which is equivalent to the standard Einstein equation if the metric is {\nondeg}.

Section \sref{s_einstein_quasireg} contains another version, which applies to a special class of {\semireg} spacetimes, which are called {\quasireg}, and whose Riemann curvature tensor admits smooth Ricci decomposition. This class of singularities will turn out to be important in problems involving the Weyl curvature tensor, in section \ref{s_wch} and chapter \ref{ch_quantum_gravity}.

The {\quasireg} singularities include, isotropic singularities, and a class of warped product singularities, the Schwarzschild singularity (see section \sref{s_black_hole_schw}), the Friedmann-Lema\^itre-Robertson-Walker Big Bang singularity (see section \sref{s_qreg_examples_warped}). This equation is constructed in terms of the Ricci part of the Riemann curvature (as the Kulkarni-Nomizu product between Einstein's equation and the metric tensor).

\section{Einstein equation at {\semireg} singularities}
\label{s_einstein_semireg}


\subsection{Einstein's equation on {\semireg} spacetimes}
\label{ss_einstein_tensor_densitized}

The Einstein tensor of a {\semireg} {\semiriem} manifold is defined as
\begin{equation}
\label{eq_einstein_tensor}
	G:=\ric-\frac 1 2 s g,
\end{equation}
in terms of the Ricci tensor and the scalar curvature, which can be defined even if the metric is degenerate metric, as long as its signature is constant (see \sref{s_ricci_tensor_scalar}). Unfortunately, at the points where signature changes, they may become infinite. Since the singularities from General Relativity are like this, we need to find something else.

\begin{definition}
\label{def_semi_reg_spacetime}
Let $(M,g)$ be a {\semireg} manifold of dimension $4$, having the signature $(0,3,1)$ at the points where it is {\nondeg}. Then, $(M,g)$ is called \textit{{\semireg} spacetime}.
\end{definition}

\begin{theorem}
\label{thm_densitized_einstein}
Let $(M,g)$ be a {\semireg} spacetime. Then, it has smooth Einstein density tensor of weight $2$, $G\det g$.
\end{theorem}
\begin{proof}
Let $p\in M$ be a point where the metric $g$ is {\nondeg}.
The Hodge dual of $R_{abcd}$ with respect to the first and the second pairs of indices \cfeg{PeR87}{234} is
\begin{equation}
	(\ast R\ast)_{abcd} = \varepsilon_{ab}{}^{st}\varepsilon_{cd}{}^{pq}R_{stpq},
\end{equation}
where $\varepsilon_{abcd}$ are the components of the volume form associated to the metric.

The Einstein tensor \eqref{eq_einstein_tensor} can be written using the Hodge $\ast$ operator as
\begin{equation}
\label{eq_einstein_tensor_hodge}
	G_{ab} = g^{st}(\ast R\ast)_{asbt}.
\end{equation}
Since, in coordinates, the volume form is, in terms of the Levi-Civita symbol,
\begin{equation}
\varepsilon_{abcd} = \epsilon_{abcd}\sqrt{-\det g},
\end{equation}
we can rewrite the Einstein tensor as 
\begin{equation}
\label{eq_einstein_tensor_hodge_lc}
	G^{ab} = \dsfrac{g_{kl}\epsilon^{akst}\epsilon^{blpq} R_{stpq}}{\det g}.
\end{equation}
If $\det g\to 0$, the Einstein tensor becomes divergent. But the tensor density of weight $2$
\begin{equation}
\label{eq_einstein_tensor_density}
	G^{ab}\det g = g_{kl}\epsilon^{akst}\epsilon^{blpq} R_{stpq},
\end{equation}
remains smooth, because it is constructed only from the Riemann curvature tensor, which is smooth (see Theorem \ref{thm_riemann_curvature_semi_regular}), and from the Levi-Civita symbol, which is constant in the particular coordinate system. Even if $G^{ab}$ diverges, $\det g\to 0$ as needed to make the tensor density $G^{ab}\det g$ smooth. The tensor density $G_{ab}\det g$, obtained by lowering its indices, is also smooth.
\end{proof}

\begin{remark}
The densitized curvature scalar is smooth, since
\begin{equation}
\label{eq_curvature_scalar_density}
	s\det g = -g_{ab}G^{ab}\det g,
\end{equation}
and the densitized Ricci tensor is also smooth, because
\begin{equation}
\label{eq_ricci_density}
	R_{ab}\det g = g_{as}g_{bt}G^{st}\det g + \dsfrac 1 2 s g_{ab}\det g.
\end{equation}
\end{remark}

\begin{definition}
In General Relativity on a {\semireg} spacetime, let $T$ be the stress-energy tensor. The equation 
\begin{equation}
\label{eq_einstein:densitized}
	G\det g + \Lambda g\det g = \kappa T\det g,
\end{equation}
or, in coordinates or local frames,
\begin{equation}
\label{eq_einstein_idx:densitized}
	G_{ab}\det g + \Lambda g_{ab}\det g = \kappa T_{ab}\det g,
\end{equation}
where $\kappa:=\dsfrac{8\pi \mc G}{c^4}$, $\mc G$ and $c$ being Newton's constant and the speed of light,
is called the \textit{densitized Einstein equation}.
\end{definition}

\section{Einstein equation at {\quasireg} singularities}
\label{s_einstein_quasireg}

In section \sref{s_einstein_exp_qreg} we introduce the \textit{expanded Einstein equation}, by taking the Kulkarni-Nomizu product between Einstein's equation and the metric tensor.
The expanded Einstein equation holds on a special type of {\semireg} spacetimes, named {\em \quasireg}. Given that we already have the densitized version of Einstein's equation, which holds on the full class of {\semireg} spacetimes, why would we need a new set of equations generalizing those of Einstein? The reason is that some applications require a smooth Weyl curvature tensor, which is ensured in {\quasireg} spacetimes (see section \sref{s_wch}, and chapter \ref{ch_quantum_gravity}). 

The {\quasireg} singularities, on which the extended Einstein equations hold, include isotropic singularities (section \sref{s_qreg_examples_isotropic}), and a class of warped product singularities, which includes the {\FLRW} Big Bang singularity (section \sref{s_qreg_examples_warped}). Also, the {\schw} singularity is {\quasireg} (section \sref{s_black_hole_schw}).

\subsection{Expanded Einstein equation and {\quasireg} spacetimes}
\label{s_einstein_exp_qreg}

The \textit{expanded Einstein equation} is
\begin{equation}
\label{eq_einstein_expanded}
	(G\circ g)_{abcd} + \Lambda (g\circ g)_{abcd} = \kappa (T\circ g)_{abcd}
\end{equation}
where the operation
\begin{equation}
\label{eq_kulkarni_nomizu}
	(h\circ k)_{abcd} := h_{ac}k_{bd} - h_{ad}k_{bc} + h_{bd}k_{ac} - h_{bc}k_{ad}
\end{equation}
is the \textit{Kulkarni-Nomizu product} of two symmetric bilinear forms $h$ and $k$.

If $g$ is {\nondeg}, it can be removed from equation \eqref{eq_einstein_expanded}, which is therefore equivalent to the Einstein equation. If $\det g\to 0$, it's inverse becomes singular, and in general the Einstein tensor $G_{ab}$, blows up (because the Ricci and scalar curvatures blow up). But the metric term from the Kulkarni-Nomizu product $G\circ g$ tends to $0$, and in some cases this is enough to cancel the blow up of the Einstein tensor.

Let $(M,g)$ be {\semiriem} manifold of dimension $n$. The Riemann curvature tensor can be decomposed algebraically (see {\eg} \cite{ST69,BESS87,GHLF04,pripoae2011connections}) as
\begin{equation}
	R_{abcd} = S_{abcd} + E_{abcd} + C_{abcd},
\end{equation}
where 
\begin{equation}
	S_{abcd} = \dsfrac{1}{2n(n-1)}R(g\circ g)_{abcd}
\end{equation}
is the scalar part of the Riemann curvature, and 
\begin{equation}
	E_{abcd} = \dsfrac{1}{n-2}(S \circ g)_{abcd}
\end{equation}
is the \textit{semi-traceless part} of the Riemann curvature. Here
\begin{equation}
\label{eq_ricci_traceless}
S_{ab} := R_{ab} - \dsfrac{1}{n}Rg_{ab}
\end{equation}
is the traceless part of the Ricci curvature.

The \textit{traceless part} of the Riemann curvature is called the \textit{Weyl curvature tensor},
\begin{equation}
\label{eq_weyl_curvature}
	C_{abcd} = R_{abcd} - S_{abcd} - E_{abcd}.
\end{equation}

Equations \eqref{eq_einstein_tensor}
 and \eqref{eq_ricci_traceless} allow us to write the Einstein tensor in terms of the traceless part of the Ricci tensor and the scalar curvature:
\begin{equation}
G_{ab} = S_{ab} - \dsfrac{1}{4}R g_{ab}.
\end{equation}
From this equation, we can define an \textit{expanded Einstein tensor}:
\begin{equation}
\label{eq_einstein_tensor_expanded}
\begin{array}{lrl}
G_{abcd} &:=& (G\circ g)_{abcd} \\
&=& (S \circ g)_{abcd} - \dsfrac{1}{4}R (g\circ g)_{abcd}\\
&=& 2 E_{abcd} - 6 S_{abcd},
\end{array}
\end{equation}
and rewrite the expanded Einstein equation in terms of it
\begin{equation}
\label{eq_einstein_expanded_explicit}
	2 E_{abcd} - 6 S_{abcd} + \Lambda (g\circ g)_{abcd} = \kappa (T\circ g)_{abcd}.
\end{equation}

\subsection{Quasi-regular spacetimes}
\label{s_qreg_spacetimes}

A {\semireg} manifold has a smooth Riemann curvature tensor $R_{abcd}$. In addition, we want to impose the condition that the tensors $E_{abcd}$ and $S_{abcd}$ are smooth, so that the expanded Einstein equation \eqref{eq_einstein_expanded_explicit} makes sense and is smooth.

\begin{definition}
\label{def_quasi_regular}
Let $(M,g_{ab})$ be a {\semireg} manifold. If the tensors $S_{abcd}$ and $E_{abcd}$ are smooth, we say that the manifold $(M,g)$ is \textit{{\quasireg}}, and that $g_{ab}$ is a \textit{{\quasireg} metric}.
\end{definition}

\subsection{Examples of {\quasireg} spacetimes}
\label{s_qreg_examples}

The {\quasireg} spacetimes include the ones with with {\nondeg} metric, but are more general than them. In the following we will see that this class is very large, and contains some relevant kinds of singularities encountered in General Relativity.

\subsubsection{Isotropic singularities}
\label{s_qreg_examples_isotropic}

\textit{Isotropic singularities} can be obtained by conformal rescalings of {\nondeg} metrics, when the scaling function vanishes \cite{Tod87,Tod90,Tod91,Tod92,Tod02,Tod03,CN98,AT99i,AT99ii}.

\begin{theorem}
\label{thm_quasireg_example_conformal}
Let $(M,\widetilde g_{ab})$ be a {\semiriem} manifold, and let $\Omega:M\to\R$ be a smooth function. Then, the manifold $(M,g_{ab} :=\Omega^2 \widetilde g_{ab})$ is {\quasireg}.
\end{theorem}
\begin{proof}
From Theorem \ref{thm_conformal_semi_regular}, $(M,g_{ab})$ is {\semireg}.

At points $p\in M$ where $\Omega\neq 0$, the Ricci and scalar curvatures are (\cite{HE95}, p. 42.):
\begin{equation}
\label{eq_conformal_ricci_curv_ud}
R^a{}_b = \Omega^{-2}\widetilde R^a{}_b + 2\Omega^{-1}(\Omega^{-1})_{;bs}\widetilde g^{as}-\dsfrac 1 2\Omega^{-4}(\Omega^2)_{;st}\widetilde g^{st}\delta^a{}_b
\end{equation}
\begin{equation}
\label{eq_conformal_scalar_curv}
R=\Omega^{-2}\widetilde R-6\Omega^{-3}\Omega_{;st}\widetilde g^{st}
\end{equation}
where the semicolon is the covariant derivative associated to $\widetilde g$.
From equation \eqref{eq_conformal_ricci_curv_ud},
\begin{equation}
R_{ab}=\Omega^2 \widetilde g_{as} R^s{}_b=\widetilde R_{ab} + 2\Omega(\Omega^{-1})_{;ab}-\dsfrac 1 2\Omega^{-2}(\Omega^2)_{;st}\widetilde g^{st}\widetilde g_{ab},
\end{equation}
which may tend to infinity as $\Omega\to 0$. Let's show that the Kulkarni-Nomizu product $\Ric\circ g$ is smooth. The term $g$ contributes with a factor $\Omega^2$, and
\begin{equation}
\Omega^2 R_{ab}=\Omega^2 \widetilde R_{ab} + 2\Omega^3(\Omega^{-1})_{;ab}-\dsfrac 1 2(\Omega^2)_{;st}\widetilde g^{st}\widetilde g_{ab},
\end{equation}
is smooth, as follows from
\begin{equation}
\begin{array}{lll}
\Omega^3(\Omega^{-1})_{;ab} &=& \Omega^3\((\Omega^{-1})_{;a}\)_{;b} = \Omega^3\(-\Omega^{-2}\Omega_{;a}\)_{;b} \\
&=& \Omega^3\(2\Omega^{-3}\Omega_{;b}\Omega_{;a} - \Omega^{-2}\Omega_{;ab}\) \\
&=& 2\Omega_{;a}\Omega_{;b} - \Omega\Omega_{;ab} \\
\end{array}
\end{equation}
Hence, the tensor $\Ric\circ g$ is smooth. The smoothness of $R g\circ g$ follows by noticing that $g\circ g$ introduces a factor $\Omega^4$, while in the expression \eqref{eq_conformal_scalar_curv} of $R$, the power of $\Omega$ is $\geq-3$.

Therefore, $E_{abcd}$ and $S_{abcd}$ are smooth, and the spacetime $(M,g_{ab})$ is {\quasireg}.
\end{proof}

\subsubsection{{\qquasireg} warped products}
\label{s_qreg_examples_warped}

We recall from Corollary \ref{thm_reg_warped} that the degenerate warped product of two {\semiriem} manifolds is a {\semireg} manifold. Now we show that, under some assumptions which correspond to the {\FLRW} spacetime, it is also {\quasireg}.

\begin{theorem}[{\qquasireg} warped product]
\label{thm_quasireg_example_wp}
Let $(B,g_B)$ and $(F,g_F)$ be two {\semireg} manifolds, with $\dim B=1$ and $\dim F=3$, and let $f\in\fiscal{B}$. Then, the degenerate warped product $(B\times_f F,g=g_B+f^2 g_F)$  is {\quasireg}.
\end{theorem}
\begin{proof}
From Corollary \ref{thm_reg_warped}, $B\times_f F$ is {\semireg}.

From \cite{ONe83}, p. 211, we know that, for horizontal vector fields $X,Y\in\fivectlift{B \times F,B}$ and vertical vector fields $V,W\in\fivectlift{B \times F,F}$,
\begin{enumerate}
	\item $\ric(X,Y) = \ric_B(X,Y) + \dsfrac{\dim F}{f}H^f(X,Y)$
	\item $\ric(X,V) = 0$
	\item $\ric(V,W) = \ric_F(V,W) + \(f\Delta f + (\dim F-1)g_B(\grad f,\grad f)\)g_F(V,W)$
\end{enumerate}
where $H^f$ is the Hessian, $\Delta f$ is the Laplacian, and $\grad f$ the gradient. It is clear that $\ric(X,V)$ and $\ric(V,W)$ are smooth, and we want to prove that $\ric(X,Y)$ is smooth too. Since $\dim B=1$, the only terms in $\Ric\circ g$ containing $\Ric(X,Y)$ have the form 
\begin{equation*}
\Ric(X,Y)g(V,W)=f^2\Ric(X,Y)g_F(V,W).
\end{equation*}
Hence, $\Ric\circ g$ is smooth.

The tensor $S_{abcd}$ is smooth too, because the scalar curvature is
\begin{equation}
	R = R_B + \frac {R_F}{f^2} + 2\dim F\dsfrac{\Delta f}{f} + \dim F(\dim F - 1)\dsfrac{g_B(\grad f,\grad f)}{f^2}.
\end{equation}
Hence, $B\times_f F$ is {\quasireg}.
\end{proof}

From Theorem \ref{thm_quasireg_example_wp} follows that the {\FLRW} spacetime with smooth {\em scale factor} (the warping function) is {\quasireg}.
In fact, from Theorem \ref{thm_quasireg_example_wp} this follows directly and more generally: 

\begin{corollary}
\label{thm_flrw}
The {\FLRW} spacetime, with smooth $a: I\to \R$, is {\quasireg}.
\end{corollary}
\begin{proof}
This is a direct consequence of Theorem \ref{thm_quasireg_example_wp}. 
\end{proof}

\chap{ch_big_bang}{The Big-Bang singularity}
This chapter is based on author's original results from the papers \cite{Sto11h}, \cite{Sto12a}, and \cite{Sto12c}.
\section{Introduction}
\label{s_big_bang_intro}

The \textit{cosmological principle} states that our expanding universe, can be considered at very large scale homogeneous and isotropic. This explains the success of the solution proposed by A. Friedmann \cite{FRI22de,FRI99en,FRI24}, which is an exact solution to Einstein's equation, describing a homogeneous, isotropic universe. It is known under the name of {\flrw} (\FLRW) spacetime, after the name of Georges Lema\^itre \cite{LEM27}, H. P. Robertson \cite{ROB35I,ROB35II,ROB35III} and A. G. Walker \cite{WAL37}, who rediscovered and made important contributions to this solution.

From the {\FLRW} model follows that the universe should be, at any time, either in expansion, or in contraction. Hubble's observations showed that the universe is currently expanding, and more recent observations showed that the expansion is accelerated \cite{PER99}. Modern cosmology saids that, long time ago, there was a huge concentration of matter which exploded, thus the name \textit{Big Bang}. It was suspected that General Relativity implies that the density of matter at the beginning of the universe was infinite, and Einstein's equation was singular. This was shown to be true, under general hypotheses, by Hawking's singularity theorem \cite{Haw66i,Haw66ii,Haw67iii,HP70,HE95} (who applied Penrose's method for the black hole singularities \cite{Pen65} backwards in time, to the past singularity of the Big Bang).

\image{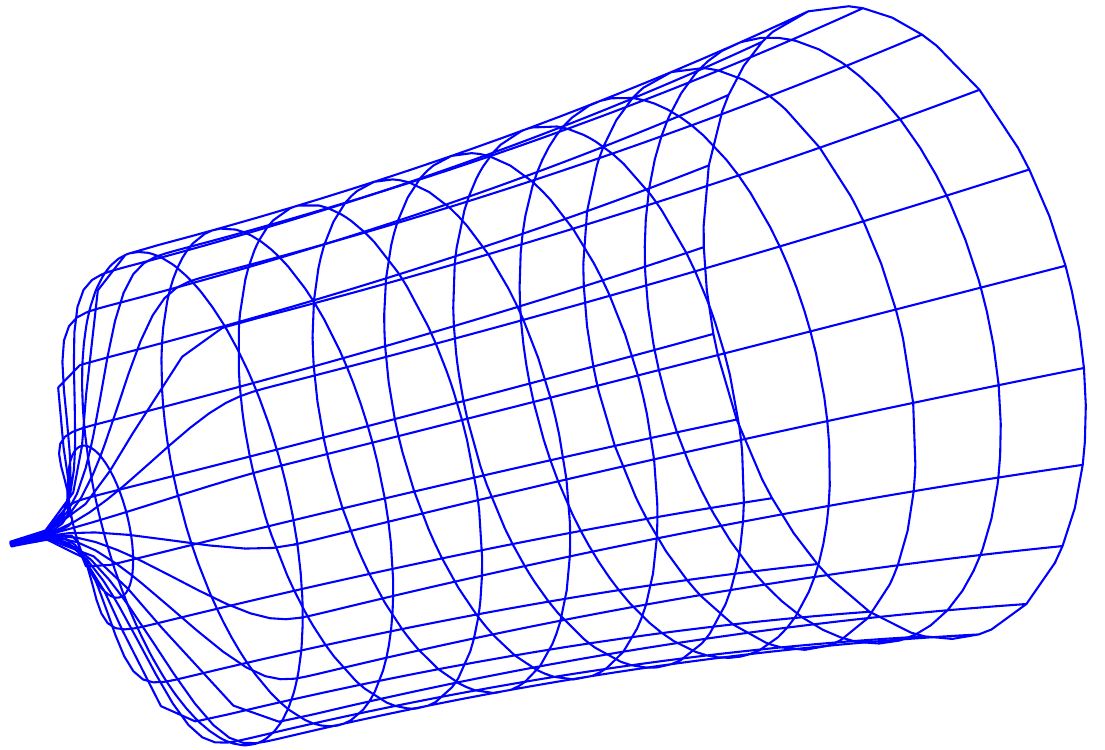}{0.6}{The universe originated from a very dense state, probably a singularity, and expanded, with a very brief period of very high acceleration.}

According to the standard cosmological model, the universe started with the Big Bang, possibly singular, and expanded. It begins with a very short period of exponentially accelerated expansion, called \textit{inflation} (Fig. \ref{flrw-std}).

The extreme conditions which present at the Big Bang are very far from the range of our experiments, and even of our theoretical models. Therefore, we don't know very well what happened then. Maybe the singularity was avoided, by some quantum effect which circumvented the energy condition from the hypothesis of the singularity theorem. Such a possibility is explored in the \textit{loop quantum cosmology} \cite{Boj01,ABL03,bojowald2003absenceLQC,Boj05,ashtekar2011LQC,visinescu2009bianchi,visinescu2012bianchi}, leading to a discrete model of the universe with a Big Bounce.

We will not explore here the possibility that the Big Bang singularity is avoided by quantum or other sort of effects, which may follow from a yet to be discovered unification of General Relativity and Quantum Theory. Instead, we will apply the tools of singular {\semiriem} geometry developed in this Thesis, to push the limits of General Relativity beyond the Big Bang singularity. We will see that the {\FLRW} singularities with smooth scale factor behave well, being {\semireg}, and even {\quasireg}.

In \sref{s_flrw_intro}, we recall some elements of the {\FLRW} spacetime.
In \sref{s_flrw_semireg}, we will prove that the {\FLRW} metric is {\semireg}, and it obeys a densitized version of Einstein equation, with density of weight $1$.
Section \sref{s_wch} discusses the Weyl Curvature Hypothesis, and the importance of the {\quasireg} singularities, which satisfy it automatically. It also shows that a large class of spacetimes which are not necessarily homogeneous and isotropic are {\quasireg}.

\section{The {\flrw} spacetime}
\label{s_flrw_intro}

In the {\FLRW} cosmological model, the $3$-space at any moment of time is modeled, up to a scale factor $a(t)$, by a three-dimensional Riemannian space $(\Sigma,g_\Sigma)$. Time is modeled by an interval $I\subseteq \R$, with the natural metric $-\de t^2$. At any moment $t\in I$, the space $\Sigma_t$ is obtained by scaling $(\Sigma,g_\Sigma)$ with $a^2(t)$, where $a: I\to \R$ is called the \textit{warping function}. The spacetime $I\times\Sigma$, with the metric
\begin{equation}
\label{eq_flrw_metric}
\de s^2 = -\de t^2 + a^2(t)\de\Sigma^2,
\end{equation}
is called the {\FLRW} spacetime.
It is the \textit{warped product} between the manifolds $(\Sigma,g_\Sigma)$ and $(I,-\de t^2)$, with the warping function $a: I\to \R$.

For the typical space $\Sigma$ one can use any Riemannian manifold one may need, but in most cases, because of the cosmological principle, one considers solutions which satisfy the homogeneity and isotropy conditions. Hence, in most cases, $\Sigma$ is taken to be, at least at large scale, one of the homogeneous spaces $S^3$, $\R^3$, and $H^3$, having, in spherical coordinates $(r,\theta,\phi)$, the metric
\begin{equation}
\label{eq_flrw_sigma_metric}
\de\Sigma^2 = \dsfrac{\de r^2}{1-k r^2} + r^2\(\de\theta^2 + \sin^2\theta\de\phi^2\),
\end{equation}
where $k=1$ for the $3$-sphere $S^3$, $k=0$ for the Euclidean space $\R^3$, and $k=-1$ for the hyperbolic space $H^3$.

\subsection{The Friedman equations}

After we choose the $3$-space $\Sigma$, the only unknown part of the {\FLRW} metric remains the function $a(t)$. This can be determined by making assumptions about matter. For simplicity, it is in general considered that the universe is filled with a fluid with mass density $\rho(t)$ and pressure density $p(t)$. These quantities depend on $t$ only, because we assume homogeneity and isotropy. The stress-energy tensor is
\begin{equation}
\label{eq_friedmann_stress_energy}
T^{ab} = \(\rho+p\)u^a u^b + p g^{ab},
\end{equation}
where $u^a$ is the timelike vector field $\partial_t$, normalized.

The energy density component of the Einstein equation leads to the \textit{Friedmann equation}
\begin{equation}
\label{eq_friedmann_density}
\rho = \dsfrac{3}{\kappa}\dsfrac{\dot{a}^2 + k}{a^2}.
\end{equation}
Here, $\kappa:=\dsfrac{8\pi \mc G}{c^4}$ \cite{osimion05gr}. We will use in the following units in which $\mc G=1$ and $c=1$.
The trace of the Einstein equation gives the \textit{acceleration equation}
\begin{equation}
\label{eq_acceleration}
\rho + 3p = -\dsfrac{6}{\kappa}\dsfrac{\ddot{a}}{a}.
\end{equation}
From these equations we get the \textit{fluid equation} (the conservation of mass-energy):
\begin{equation}
\label{eq_fluid}
\dot{\rho} = -3 \dsfrac{\dot{a}}{a}\(\rho + p\).
\end{equation}

Knowing the function $a$, determines uniquely $\rho$, from \eqref{eq_friedmann_density}, and then $p$, from \eqref{eq_acceleration}.

Recent observations on supernovae reveal that the expansion is accelerated, hence there is a positive cosmological constant $\Lambda$ \cite{RIE98,PER99}. The equations \eqref{eq_friedmann_density},  \eqref{eq_acceleration}, and  \eqref{eq_fluid} assume $\Lambda=0$, but the generality is not lost, because the version of the equations having $\Lambda\neq 0$ are equivalent to the above ones, by the substitution
\begin{equation}
\label{eq_friedmann_lambda}
\begin{array}{l}
\left\{
\begin{array}{lll}
\rho &\to& \rho + \kappa^{-1}\Lambda  \\
p &\to& p - \kappa^{-1}\Lambda  \\
\end{array}
\right..
\\
\end{array}
\end{equation}
To simplify calculations, we will ignore $\Lambda$ in the following, without any loss of generality.

\subsection{Distance separation {\vs} topological separation}

We can understand the singularities in the {\FLRW} model, by making distinction between topology and geometry. A manifold $M$ is a topological space. If, in addition, there is a metric tensor $g$ on $M$, we obtain a geometry. The metric $g$ defines a distance, and the distance between two points $p\neq q\in M$ is zero only if the metric is Riemannian. If the metric is {\semiriem}, it is possible that the distance between two distinct points $p$ and $q$ is zero, provided that they are separated by a lightlike interval. If the metric $g$ is allowed to be degenerate, then there are more possibilities to have zero distance between distinct points. For example, consider a surface in $\R^3$, defined locally as the image of a map $f:U\to\R^3$, where $U\in\R^2$ is an open subset of $\R^2$. If the function $f$ is not injective, the resulting surface has self-intersections. Sometimes the surface can be defined implicitly, as the set of solutions of an equation. In this case too it may have self-intersections. The typical example is the cone
\begin{equation}
\label{eq_cone}
	x^2-y^2-z^2=0.
\end{equation}
It has a singularity at $x=0$. To resolve it, we make the transformation
\begin{equation}
\label{eq_desing_cone}
\begin{array}{l}
\left\{
\begin{array}{ll}
	x&=u \\
	y&=uv \\
	z&=uw \\
\end{array}
\right.
\\
\end{array}
\end{equation}
which maps the cylinder $v^2+w^2=1$ to the cone \eqref{eq_cone}. This procedure was studied starting with Isaac Newton \cite{newt60}, and is very used in mathematics, especially in \textit{algebraic geometry}.

\image{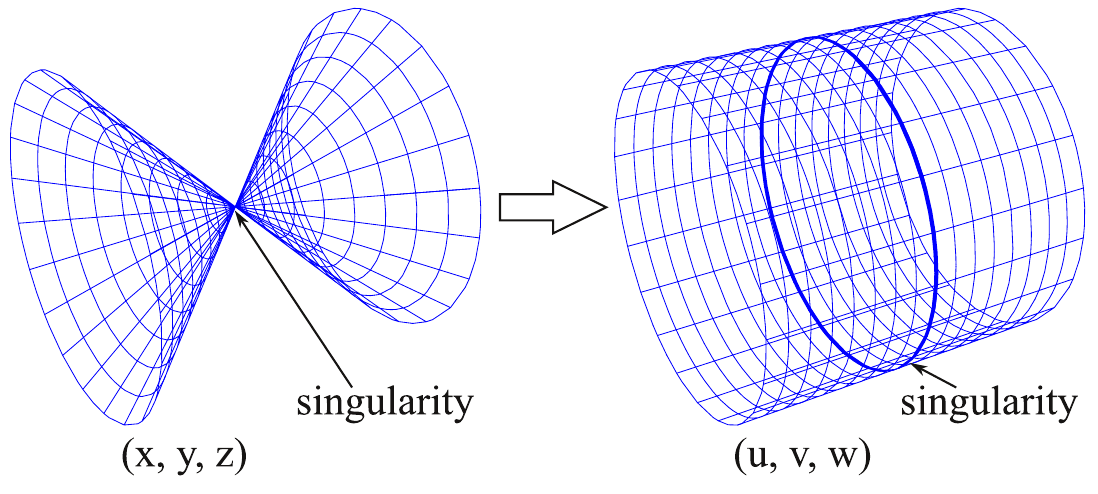}{0.8}{The old method of resolution of singularities ``unties'' the cone into a cylinder, removing the singularity.}

The natural metric on the space $(x,y,z)$ induces, by pull-back through the map \eqref{eq_desing_cone}, a metric on the the cylinder $v^2+w^2=1$. The induced metric is degenerate on the circle $u=0$ of the cylinder -- the distance between any two distinct points of the circle $u=0$ is zero.

\section{Densitized Einstein equation on the {\FLRW} spacetime}
\label{s_flrw_semireg}

We show that the Friedmann equations and the Einstein equation can be written in an equivalent form, which in addition avoids the infinities in a natural way, being thus valid at the singularity $a(t)=0$.

\subsection{What happens when the density becomes infinite?}
\label{s_infinite_density}

When $a\to 0$,  from \eqref{eq_friedmann_density} follows that $\rho\to\infty$, because a finite quantity of matter occupies a volume equal to $0$. From \eqref{eq_acceleration},  the pressure density $p$ may become infinite. So, is there a way to avoid the infinities? The answer is yes, provided that we realize that not $\rho$ and $p$ are the physical quantities, but the densities $\rho\sqrt{-g}$ and $p\sqrt{-g}$, which remain smooth.

In the Friedmann equations, $\rho$ and $p$ are scalar fields representing densities. But if we change the coordinates, $\rho$ and $p$ change, so they are not in fact scalars, but the components of another type of object. The adequate, invariant quantities representing densities involve the \textit{volume form}
\begin{equation}
\label{eq_dvol}
	\vol := \sqrt{-g}\de t\wedge\de x\wedge\de y\wedge\de z,
\end{equation}
where $\sqrt{-g}:=\sqrt{-\det g_{ab}}$.

The correct densities are not the scalars $\rho$ and $p$, but the differential $4$-forms $\rho\vol$ and $p\vol$. Their components in a coordinate system are $\rho\sqrt{-g}$ and respectively $p\sqrt{-g}$. They become identical to $\rho$ and $p$ only if $\det g=-1$, {\eg} in an orthonormal frame, like the comoving coordinate system of the {\FLRW} model. But an orthonormal frame doesn't make sense when $a\to 0$, because $\det g\to 0$. 

The determinant of the metric \eqref{eq_flrw_metric}, in the {\FLRW} coordinates, is
\begin{equation}
\label{eq_det_g_flrw}
	\det g = -a^6 \det{}_3 g_{\Sigma},
\end{equation}
where $\det_3 g_{\Sigma}$ is determinant of the metric of the $3$-dimensional typical space $\Sigma$, and is constant.
Hence, the metric's determinant in the comoving coordinates is 
\begin{equation}
\label{eq_sqrt_det_g_flrw}
	\sqrt{-g} = a^3 \sqrt{g_{\Sigma}}.
\end{equation}
Since $\sqrt{-g}\to 0$ when $a\to 0$, we will see that $\sqrt{-g}$ cancels the singularities introduced by $\rho$ and $p$ in $\rho\vol$, respectively $p\vol$.

The conservation of energy takes the form
\begin{equation}
	-a^3\dot\rho = 3 a^2 \dot a \rho + \dsfrac 3 {c^2} a^2\dot a p,
\end{equation}
and is valid even when the volume $a^3\to 0$.

\subsection{The Big Bang singularity resolution}
\label{s_bb_sing_resolved_generic}

Let's make the following substitution:
\begin{equation}
\label{eq_substitution_densities}
\begin{array}{l}
\left\{
\begin{array}{ll}
	\widetilde\rho = \rho \sqrt{-g} = \rho a^3 \sqrt{g_{\Sigma}} \\
	\widetilde p = p \sqrt{-g} = p a^3 \sqrt{g_{\Sigma}} \\
\end{array}
\right.
\\
\end{array}
\end{equation}

\begin{theorem}
\label{thm_bb_sing_resolved}
Let $a:I\to \R$ be a smooth function. Then, the densities $\widetilde\rho$, $\widetilde p$, and the densitized stress-energy tensor $T_{ab}\sqrt{-g}$ are smooth (and hence nonsingular), including when $a(t)=0$.
\end{theorem}
\begin{proof}
From the Friedmann equation \eqref{eq_friedmann_density}, which becomes now
\begin{equation}
\label{eq_friedmann_density_tilde}
\widetilde\rho = \dsfrac{3}{\kappa}a\(\dot a^2 + k\) \sqrt{g_{\Sigma}},
\end{equation}
if $a$ is smooth, $\widetilde\rho$ is smooth too.
Similarly, $\widetilde p$ is smooth too, from the acceleration equation \eqref{eq_acceleration}
\begin{equation}
\label{eq_acceleration_tilde}
\widetilde\rho + 3\widetilde p = -\dsfrac{6}{\kappa}a^2\ddot{a} \sqrt{g_{\Sigma}}.
\end{equation}

Since $\widetilde\rho$ and $\widetilde p$ are smooth functions, the densitized stress-energy tensor 
\begin{equation}
\label{eq_friedmann_stress_energy_densitized}
T_{ab}\sqrt{-g} = \(\widetilde\rho+\widetilde p\)u_a u_b + \widetilde p g_{ab},
\end{equation}
is smooth too. 
\end{proof}

Due to Theorem \ref{thm_bb_sing_resolved}, the smooth densitized version of the Einstein equation
\begin{equation}
\label{eq_einstein_idx:densitized_w1}
	G_{ab}\sqrt{-g} + \Lambda g_{ab}\sqrt{-g} = \kappa T_{ab}\sqrt{-g},
\end{equation}
makes sense and is smooth.

The Lagrangian density for General Relativity, proposed by Hilbert and Einstein, is
\begin{equation}
\label{eq_lagrangian}
	\dsfrac{1}{2\kappa}\(R\sqrt{-g} - 2\Lambda\sqrt{-g}\) + \mc L\sqrt{-g},
\end{equation}
where the Lagrangian density $\mc L\sqrt{-g}$ describes matter. While, for example, the scalar curvature $R$ is singular at $a(t)\to 0$, the density $R\sqrt{-g}$ is smooth.

For smooth $a$ so that $a(0)=0$, if $\dot a(0)\neq 0$, the {\FLRW} solution looks as in figure \ref{flrw-nonsing} \textbf{A}. If $\dot a(0)=0$, it behaves like a Big Bounce universe, except that the bounce is singular (fig. \ref{flrw-nonsing} \textbf{B}).

\image{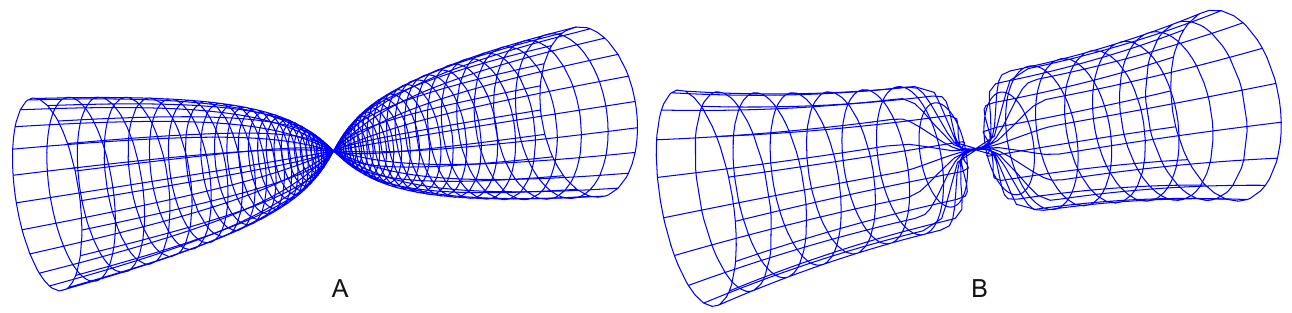}{1.0}{A schematic representation of a generic Big Bang singularity.\\\textbf{A.} $\dot a(0)\neq 0$. \textbf{B.} $\dot a(0)=0$, $\ddot a(0)>0$.}



\section{The Weyl curvature hypothesis}
\label{s_wch}

Penrose, while attempting to explain the high homogeneity and isotropy, and the very low entropy of the early universe, stated the {\em \Wch} (\WCH), conjecturing the vanishing of the Weyl tensor at the Big Bang singularity.

The {\quasireg} metrics provide a large class of singularities satisfying the \Wch. 
We find a very general cosmological model, which generalizes the {\FLRW} model, but also isotropic singularities, by dropping the isotropy and homogeneity constraints. We show that it is {\quasireg}, and therefore satisfies the {\WCH}.

\subsection{Introduction}
\label{s_intro}

R. Penrose's {\WCH} originates in several distinct problems, including the search for an explanation of the second law of Thermodynamics, and of the high homogeneity and isotropy of the universe \cite{Pen79}. He analyzed the flow of energy in the Universe, and concluded that the second law of Thermodynamics is caused by a very high homogeneity near the Big-Bang. To explain this homogeneity, he wrote (\cite{Pen79}, p. 614)

\begin{quote}
In terms of spacetime curvature, the absence of clumping corresponds, very roughly, to the absence of Weyl conformal curvature (since absence of clumping implies spatial-isotropy, and hence no gravitational principal null-directions).
\end{quote}

He then emitted the \textit{\Wch} (\cite{Pen79}, p. 630)
\begin{quote}
this restriction on the early geometry should be something like: the Weyl curvature $C_{abcd}$ vanishes at any initial singularity
\end{quote}

There is also a motivation from the {\WCH}, coming from Quantum Gravity. It is expected that, near the Big-Bang, the quantum effects of gravity become relevant, but gravity is perturbatively nonrenormalizable at two loops \cite{HV74qg,GS86uvgr}. If the Weyl tensor vanishes, local degrees of freedom, hence gravitons, vanish too, removing some problems of Quantum Gravity, at least at the Big-Bang \cite{Car95}.

The Weyl tensor $C_{abcd}$ is, from gravitational viewpoint, responsible for the tidal forces. It is the traceless part of the Riemann curvature tensor $R_{abcd}$. The tensor $C_{abc}{}^d$ is invariant at a conformal rescaling $g_{ab}\mapsto\Omega^2 g_{ab}$. It vanishes on an open set $U\subset M$, where $\dim M\geq 4$, if and only if the metric is conformally flat on $U$.

{\em Isotropic singularities}, studied by Tod \cite{Tod87,Tod90,Tod91,Tod92,Tod02,Tod03}, Claudel \& Newman \cite{CN98}, Anguige \& Tod \cite{AT99i,AT99ii}, are obtained from {\nondeg} metrics $\widetilde g$ on $M$, by a conformal rescaling $g_{ab}=\Omega^2 \widetilde g_{ab}$, when $\Omega\to 0$. They have finite Weyl curvature tensor $C_{abc}{}^d=\widetilde C_{abc}{}^d$, and 
\begin{equation}
	C_{abcd}=g_{sd}C_{abc}{}^s=\Omega^2 \widetilde g_{sd}\widetilde C_{abc}{}^s\to 0.
\end{equation}

Another simple example of vanishing Weyl tensor comes from the {\FLRW} cosmological model. In fact, the Weyl tensor of the {\FLRW} metric vanishes identically, and is not relevant for {\WCH}.

Both these types of singularities are particular cases of the {\quasireg} singularities (\sref{s_qreg_examples}). Moreover, any {\quasireg} singularity in spacetime satisfies {\WCH}, as we will show. From \eqref{eq_weyl_curvature} we know that the Weyl curvature $C_{abcd}$ is smooth. At {\semireg} singularities (therefore at {\quasireg} too), the metric behaves as if it loses one or more dimensions, and the Weyl curvature lives, from algebraic viewpoint, in a space of dimension lower than $4$, therefore it vanishes, because of its algebraic symmetries.

In \sref{s_wch_ex} we propose a very general cosmological model which, unlike the {\FLRW} model, does not assume that the space slices have constant metric up to the overall scaling factor $a^2(t)$. We will keep the overall scaling factor $a^2(t)$, but we will allow the space part of the metric to change freely in time. We allow this generality because the Universe is not perfectly isotropic and homogeneous. Due to the general conditions we assume, we will not be concerned at this point with the particular matter content of this universe.

\subsection{The Weyl tensor vanishes at {\quasireg} singularities}
\label{s_wch_thm}

\begin{theorem}
\label{thm_wch}
Let $(M,g)$ be a {\quasireg} manifold of dimension $4$. Then, the Weyl curvature tensor $C_{abcd}$ vanishes at singularities.
\end{theorem}
\begin{proof}
The smoothness of $C_{abcd}$ follows from the smoothness of $R_{abcd}$, $E_{abcd}$, $S_{abcd}$, and from equation \eqref{eq_weyl_curvature}.

From Corollary \ref{thm_curvature_tensor_radical}, the Riemann curvature tensor $R_{abcd}$ of a {\semireg} {\semiriem} manifold satisfies at any $p\in M$
\begin{equation}
\left(R_{abcd}\right)_p \in \otimes^4 \annih{T}_pM.
\end{equation}
If at a point $p$ the metric $g$ is degenerate, then $\dim\left(\annih{T}_pM\right)<4$.
But in dimension $\leq 3$, any tensor having the symmetries of the Weyl tensor vanishes (see {\eg} \cite{BESS87}), hence, at any point $p$ where $g$ is degenerate,
\begin{equation}
\left(C_{abcd}\right)_p=0.
\end{equation}
\end{proof}

\subsection{Example: a general cosmological model}
\label{s_wch_ex}

The universe is not homogeneous and isotropic at all scales, as it is assumed in the {\FLRW} model. Thus, we are motivated to study a spacetime $(M,g)$ which is allowed to be inhomogeneous and anisotropic.

We consider that the spacetime is a manifold of the form $M=I\times \Sigma$, where $I\subseteq \R$ is an interval, and $\Sigma$ is a three-dimensional manifold. Let $\tau:M\to I$, defined by $\tau(t,x)=t$, be a global time coordinate. Let's consider that on each slice $\Sigma_t=\tau^{-1}(t)$ there is a Riemannian (hence {\nondeg}) metric $h_{ij}(t,x)$ (where $1\leq i,j\leq 3$), which depends smoothly on $(t,x)\in I\times\Sigma$. Let's consider on $I$ a smooth function $N:I\to\R$, and the metric $N^2(t)\de t^2$, which is allowed to be degenerate.

We represent the metric $h(t)$ as an arc element by $\de\sigma_t^2$, and assume that the metric $g$ on the total manifold $M$ is
\begin{equation}
	g_{ij}(t,x):=-N^2(t)\de t^2 + a^2(t)h_{ij}(t,x),
\end{equation}
where $a:I\to\R$ is a smooth function. The function $a(t)$ is allowed to vanish, the Big-Bang singularity is obtained for $a(t)=0$.

As an arc element, the metric is
\begin{equation}
\label{eq_metric_a_Nh}
\de s^2 = -N^2(t)\de t^2 + a^2(t)\de\sigma_t^2.
\end{equation}

The {\FLRW} model is obtained if $N(t)=1$, and $h_{ij}(t)$ is time independent, and of constant curvature. But we will work in full generality, allowing $\de\sigma_t^2$ to vary freely in time, to obtain therefore a much more general solution. When $a(t)=0$, our Big-Bang singularity is much more general than the {\FLRW} one, because the geometry of space slices $\(\Sigma_t,h(t)\)$ may be inhomogeneous and variable in time.

If $N(t)\neq 0$ for any $t\in I$, then $I$ can be reparameterized to obtain a constant metric, so that we can consider $N(t)=1$. Hence, the relevant differences introduced by using a non-constant $N$ become important only when $N(t)$ vanishes, together with $a(t)$.
For reasons which will become apparent, we require that
\begin{equation}
\label{eq_f_aN}
	f(t):=\dsfrac{a(t)}{N(t)}
\end{equation}
is not singular. For example, if $f(t)=1$, then $N(t)=a(t)$, and the resulting singularities are just isotropic singularities.

We are here interested in the most general case.

\begin{theorem}
\label{thm_metric_a_Nh}
Let $(M,g)$ be a spacetime, where $M=I\times \Sigma$, and the metric is given by \eqref{eq_metric_a_Nh}. Then,  $(M,g)$ is {\quasireg}.
\end{theorem}
\begin{proof}
We will prove first that the metric is {\semireg}, by showing that the terms in the Riemann curvature tensor \eqref{eq_riemann_curvature_tensor_coord} are smooth. Then, we will show that the Ricci decomposition
\begin{equation}
\label{eq_ricci_decomposition}
	R_{abcd} = E_{abcd} + S_{abcd} + C_{abcd}.
\end{equation}
is smooth.

The metric $g$ is
\begin{equation}
	g(t,x) =
	\left(
\begin{array}{cc}
	-N^2(t) & 0 \\
	0 &  a^2(t) h_{ij}(t,x) \\
\end{array}
\right),
\end{equation}
and its reciprocal is
\begin{equation}
	g^{-1}(t,x) =
	\left(
\begin{array}{cc}
	-N^{-2}(t) & 0 \\
	0 &  a^{-2}(t) h^{ij}(t,x) \\
\end{array}
\right).
\end{equation}

The partial derivatives of the metric are
\begin{equation}
\label{eq_g_der}
\begin{array}{ll}
	g_{00,0} &= -2N\dot N, \\
	g_{00,k} &= 0, \\
	g_{ij,0} &= a(2\dot a h_{ij} + a \dot h_{ij}), \\
	g_{ij,k} &= a^2 \partial_k h_{ij}. \\
\end{array}
\end{equation}

The second order partial derivatives of the metric are

\begin{equation}
\label{eq_g_der_der}
\begin{array}{ll}
	g_{00,00} &= -2\(\dot N^2 + N\ddot N\), \\
	g_{00,k0} &= g_{00,0k} = g_{00,kl} =  0, \\
	g_{ij,00} &= 2\dot a^2 h_{ij} + 2a\ddot a h_{ij}  + 4a\dot a \dot h_{ij} + a^2 \ddot h_{ij}, \\
	g_{ij,k0} &= a\(2 \dot a \partial_k h_{ij} + a \partial_k \dot h_{ij}\), \\
	g_{ij,kl} &= a^2 \partial_k \partial_l h_{ij}.\\
\end{array}
\end{equation}

To check that $g$ is {\semireg}, it is enough to check the smoothness of the terms of the form $g_{ab,\cocontr}g_{cd,\cocontr}$. From the equations \eqref{eq_g_der},

\begin{equation}
\begin{array}{lll}
	g_{00,\cocontr}g_{00,\cocontr} &=& -N^{-2}g_{00,0}g_{00,0} + a^{-2}h^{cd}g_{00,c}g_{00,d} \\
	&=& -4\dot N^2, \\
\end{array}
\end{equation}

\begin{equation}
\begin{array}{lll}
	g_{00,\cocontr}g_{ij,\cocontr} &=& -N^{-2}g_{00,0}g_{ij,0} + a^{-2}h^{cd}g_{00,c}g_{ij,d} \\
	&=& 2\dsfrac{\dot N a}{N}\(2\dot a h_{ij} + a\dot h_{ij}\),  \\
	\end{array}
\end{equation}
and
\begin{equation}
\begin{array}{lll}
	g_{ij,\cocontr}g_{kl,\cocontr} &=& -N^{-2}g_{ij,0}g_{kl,0} + a^{-2}h^{cd}g_{ij,c}g_{kl,d} \\
	&=& -\dsfrac{a^2}{N^2}(2\dot a h_{ij} + a \dot h_{ij})(2\dot a h_{kl} + a \dot h_{kl})
	+ a^2h^{cd}\partial_c h_{ij}\partial_d h_{kl}. \\
	\end{array}
\end{equation}

But for a smooth function $f(t,x)$
\begin{equation}
a(t,x) = f(t,x)N(t),
\end{equation}
therefore, the terms calculated above are smooth, as we can see:
\begin{equation}
\label{eq_g_cocontr}
\begin{array}{ll}
	g_{00,\cocontr}g_{00,\cocontr} &= -4\dot N^2 \\
	g_{00,\cocontr}g_{ij,\cocontr} &= 2\dot N f \(2\dot a h_{ij} + a\dot h_{ij}\)  \\
	g_{ij,\cocontr}g_{kl,\cocontr} &= f^2\( - (2\dot a h_{ij} + a \dot h_{ij})(2\dot a h_{kl} + a \dot h_{kl}).
		+ N^2h^{cd}\partial_c h_{ij}\partial_d h_{kl}\)\\
\end{array}
\end{equation}
Hence, the metric $g$ is {\semireg}.

To prove that $\Ric\circ g$ and $R g \circ g$ are smooth, we have to contract the terms from \eqref{eq_g_der_der} and \eqref{eq_g_cocontr}, and see what happens when taking Kulkarni-Nomizu products with $g$. The tensor $\Ric\circ g$ is a sum of products between $g_{ab,cd}$ or $g_{ab,\cocontr}g_{cd,\cocontr}$, and $g^{ef}$, and $g_{gh}$. The tensor $R g \circ g$ is a sum of products between $g_{ab,cd}$ or $g_{ab,\cocontr}g_{cd,\cocontr}$, $g^{ef}g_{gh}$, and $g\circ g$, which is of the form $N^4(t)f^2(t)q_{abcd}$. We resume these products in the tables \ref{tab_kn_ricci} and \ref{tab_kn_scalar}, where we can see that they are smooth. To write these tables, we used the following facts:
\begin{itemize}
	\item 
only some particular combinations of indices are allowed when contracting, and in the Kulkarni-Nomizu products,
	\item 
$g\circ g$ is of the form $N^4f^2q_{abcd}$, with $q_{abcd}$ smooth,
	\item 
$g_{ab}=N^2\tilde{g}_{ab}$, where $\tilde{g}_{ab}$ is smooth; $g^{ab}=N^{-2}f^{-2}\hat{g}^{ab}$, with $\hat{g}^{ab}$ smooth
	\item 
$g_{ij}=a^2h_{ij}=f^2N^2h_{ij}$; $g^{ij}=a^{-2}h^{ij}=f^{-2}N^{-2}h^{ij}$.
\end{itemize}

\begin{table}[!htbp]
\centering
\caption{Terms from $\Ric\circ g$.}
\label{tab_kn_ricci}
\begin{tabular}{|c|c|c|c|}
\toprule
Term &  \multicolumn{2}{c|}{Multiply with} & Term from $\Ric\circ g$ \\
\midrule
$u_{abcd}$ &  $g^{ef}$ & $g_{pq}$ & $u_{abcd} g^{ef} g_{pq}$ \\
{} &  ${}_{\{e,f\}\subset\{a,b,c,d\}}$ & ${}_{\{p,q\}\neq \{a,b,c,d\}-\{e,f\}}$ & {} \\\toprule
$g_{00,00}$ or $g_{00,\cocontr}g_{00,\cocontr}$ & $-N^{-2}$ & $a^2h_{ij}$ &  $-u_{0000}f^2h_{ij}$ \\\midrule
$g_{ij,kl}=a^2 \partial_k \partial_l h_{ij}$ & $a^{-2}h^{ef}$ & $g_{pq}$ & $h^{ef}\partial_k \partial_l h_{ij}g_{pq}$ \\\midrule
$g_{ij,\cocontr}g_{kl,\cocontr}=f^2 \tilde{u}_{ijkl}$ & $f^{-2}N^{-2}h^{ef}$ & $N^2\tilde{g}_{pq}$ & $\tilde{u}_{ijkl}h^{ef}\tilde{g}_{pq}$ \\\midrule
$g_{ij,00}$ or $g_{ij,\cocontr}g_{00,\cocontr}$ & $g^{i0}=0$ & $g_{pq}$ & $0$ \\\midrule
$g_{ij,00}$ or $g_{ij,\cocontr}g_{00,\cocontr}$ & $g^{00}=N^{-2}$ & $N^2\tilde{g}_{pq}$ & $\tilde{g}_{pq}u_{abcd}$ \\\midrule
$g_{ij,00}$ or $g_{ij,\cocontr}g_{00,\cocontr}$ & $g^{ij}=a^{-2}h^{ij}$ & $a^2h_{pq}$ ($p,q\neq 0$) & $h^{ik}h_{pq}u_{abcd}$ \\\midrule
$g_{ij,k0}=a v_{ijk}$ & $g^{ef}=a^{-2}h^{ef}$ & $a^2h_{pq}$ ($p,q\neq 0$) & $a h^{ef}h_{pq} v_{ijk}$ \\\bottomrule
\end{tabular}
\end{table}

\begin{table}[!htbp]
\centering
\caption{Terms from $R g\circ g$.}
\label{tab_kn_scalar}
\begin{tabular}{|c|c|c|c|}
\toprule
Term &  \multicolumn{2}{c|}{Multiply with} & Term from $R(g\circ g)$ \\
\midrule
$u_{abcd}$ &  $g^{ef}g^{gh}$ & $(g\circ g)_{pqrs}$ & $u_{abcd} g^{ef}g^{gh} (g\circ g)_{pqrs}$ \\
{} &  ${}_{\{e,f,g,h\}=\{a,b,c,d\}}$ & {} & {} \\\toprule
$g_{00,00}$ or $g_{00,\cocontr}g_{00,\cocontr}$ & $-N^{-4}$ & $N^4f^2q_{abcd}$ &  $-u_{0000}f^2q_{abcd}$ \\\midrule
$g_{ij,kl}=a^2 \partial_k \partial_l h_{ij}$ & $a^{-4}h^{ef}h^{gh}$ & $N^4f^2q_{abcd}$ & $N^2q_{abcd} h^{ef}h^{gh}\partial_k \partial_l h_{ij}$ \\\midrule
$g_{ij,\cocontr}g_{kl,\cocontr}=f^2 \tilde{u}_{ijkl}$ & $f^{-4}N^{-4}h^{ef}h^{gh}$ & $N^4f^2q_{abcd}$ & $\tilde{u}_{ijkl}h^{ef}h^{gh}q_{abcd}$ \\\midrule
$g_{ij,00}$ or $g_{ij,\cocontr}g_{00,\cocontr}$ & $g^{00}g^{ij}=N^{-4}f^{-2}h^{ij}$ & $N^4f^2q_{abcd}$ & $h^{ij}u_{abcd}q_{abcd}$ \\\midrule
$g_{ij,k0}=a v_{ijk}$ & $0$ & $N^4f^2q_{abcd}$ & $0$ \\\bottomrule
\end{tabular}
\end{table}

By inspecting the tables, we see that all terms contained in $R_{abcd}$, $E_{abcd}$, and $S_{abcd}$ are smooth. Hence, the metric \eqref{eq_metric_a_Nh} is {\quasireg}.
\end{proof}

\subsection{Conclusion}

The {\quasireg} singularities offer a nice surprise, since in dimension $4$ they have vanishing Weyl curvature $C_{abcd}$. Therefore, any {\quasireg} Big-Bang singularity also satisfies the {\Wch} (\sref{s_wch_thm}).

As a main application, we studied in \sref{s_wch_ex} a cosmological model which is much more general than {\FLRW}, because it drops the isotropy and homogeneity conditions. This generality is more realistic from physical viewpoint, since our Universe appears homogeneous and isotropic only at very large scales. This model contains as particular cases, in addition to {\FLRW}, also the isotropic singularities.

\chap{ch_black_hole}{Black hole singularity resolution}

In this chapter are presented the original results published by the author in the papers \cite{Sto11e}, \cite{Sto11f}, \cite{Sto11g}, and \cite{Sto12e}.

At least in classical General Relativity, black hole singularities are usually considered to break the time evolution.

For each of the standard black hole solutions, we find coordinates which make the metric smooth. This is somehow similar to the method used by Eddington \cite{eddington1924comparison} and Finkelstein \cite{finkelstein1958past}, to show that the event horizon singularity is due to the coordinates. In our case, the singularities remain, but the metric becomes degenerate and analytic, without singular components.

We show that the {\schw} singularity can be made analytic, and in fact {\semireg}, by such a coordinate transformation, in section \sref{s_black_hole_schw}. Sections \sref{s_black_hole_rn} and \sref{s_black_hole_kn} present coordinates in which the {\rn}, respectively the {\kn} singularities are made analytic. The latter involve the presence of an electromagnetic field, and we show that this field and its potential are analytic in our coordinates, thus being {\nonsing}, as they appear in the standard, but singular coordinates.

In section \sref{s_black_hole_hyper} we show that these solutions can be utilized to construct spacetimes with more general black hole singularities, which are created and then vanish by Hawking evaporation.

The results are applied to non-stationary black holes, including evaporating ones.

\section{Introduction}

\subsection{The singularity theorems}

Despite the successes of General Relativity, one of its own consequences seems to question it: the occurrence of singularities in the black holes. It is often said that General Relativity predicts, because of these singularities, it's own breakdown \cite{HP70,Haw76,ASH91,HP96,Ash08,Ash09}. Such singularities follow from the \textit{singularity theorems} of Penrose and Hawking \cite{Pen65,Haw66i,Haw66ii,Haw67iii,HP70,HE95}. The conditions leading to singularities were found to be common (Christodoulou \cite{Chr09}), and then even more common (Klainerman and Rodnianski \cite{KR09}).

Initially, there was some confusion regarding the singularities. In 1916, when {\schw} proposed \cite{Scw16a,Scw16b} his solution to Einstein's equation, representing a black hole, it was believed that the event horizon is singular. Only after 1924, when Eddington proposed another coordinate system which removed the singularity at the event horizon \cite{eddington1924comparison}, and 1958, with the work of D. Finkelstein \cite{finkelstein1958past}, it was understood that the event horizon's singularity was only apparent,  being due to the choice of the coordinate system. But the singularity at the center of the black hole remained independent of the particular coordinates, and the singularity theorems showed that any black hole would have such a singularity. In \cite{Sto11e,Sto11f,Sto11g} we have shown that, although the genuine singularities cannot be removed, they can at least be made manageable -- there are coordinate changes which make the metric degenerate, but smooth.

\subsection{The black hole information paradox}

Black holes have interesting properties similar to the entropy and temperature in thermodynamics, which were studied in \cite{bardeen1973four,bekenstein1973black,udriste2010black,udricste2013controllability,udriste2013black,udricste2013optimal}. 

Soon (in its proper time) after an object passes through the event horizon, it reaches the singularity of the black hole. All the information contained in it seems to vanish in the singularity.

On the other hand, the equations governing the physical laws are in general reversible, guaranteeing that no information can be lost. But according to Hawking \cite{Haw75,Haw76} the black hole may emit radiation and evaporate. If the black hole evaporates completely, it seems to leave behind no trace of the information it swallowed. Moreover, it seems to be possible for an originally pure state to end up being mixed, because the density matrix of the particles in the black hole's exterior is obtained by tracing over the particles lost in the black hole with which they were entangled. This means that the unitarity appears to be violated, and the problem becomes even more acute.

\subsection{The meaning of singularities}

The singularities in General Relativity are places where the evolution equations cannot work, because the involved fields become infinite. If the Cauchy surface on which the fields are defined is affected by singularities, then the equations cannot be developed in time.

From geometric viewpoint, these singularities are points where the metric becomes singular, and the geodesics become \textit{incomplete}. Since we don't know to extend the fields to such points, normally we remove them from the spacetime.

Actually, we can rewrite the fields involved, and the equations defining them, so that the fields remain finite at any point \cite{ER35, Sto11a, Sto11e, Sto11f, Sto11g}. At the points where the metric is non-singular, the equations remain equivalent to Einstein's equation. 

With this modification, the spacetime can be extended to the singular points.

Once we have the fields and the topology repaired, we have to check that we can choose a maximal globally hyperbolic spacetime (or equivalently, admitting a Cauchy foliation) so that the evolution equations can be defined.

Consequently, a natural interpretation of the singularities emerges, which makes them harmless for the physical law, in particular for the information conservation.

We will illustrate this approach on the black hole solutions known from the literature.

\section{Schwarzschild singularity is semi-regularizable}
\label{s_black_hole_schw}

We show that the {\schw} metric can be made analytic at the singularity, by a proper coordinate transformation. The singularity remains, but it is made degenerate and smooth, and the infinities are removed. We find a family of analytic extensions, and one of them is {\semireg}. A degenerate singularity doesn't break the topology. We prove that the metric is {\semireg}, hence the densitized version of Einstein's equation can be used, avoiding thus the infinities. Moreover, we show that the singularity is {\quasireg}. In the new coordinates, the {\schw} solution extends beyond the singularity, suggesting the possibility that the information is not destroyed by the singularity, and can be restored after the evaporation.

\subsection{Introduction}


The {\schw} metric, expressed in the {\schw} coordinates, is
\begin{equation}
\label{eq_schw_schw}
\de s^2 = -\(1-\dsfrac{2m}{r}\)\de t^2 + \(1-\dsfrac{2m}{r}\)^{-1}\de r^2 + r^2\de\sigma^2,
\end{equation}
where $m$ the mass of the body, and the units were chosen so that $c=1$ and $G=1$, and
\begin{equation}
\label{eq_sphere}
\de\sigma^2 = \de\theta^2 + \sin^2\theta \de \phi^2
\end{equation}
is the metric of the unit sphere $S^2$ (see {\eg} \citep{HE95}{149}).

From \eqref{eq_schw_schw} follows that this spacetime is a warped product between a two-dimensional {\semiriem} space and the sphere $S^2$ with the metric \eqref{eq_sphere}. This allows us to change the coordinates $r$ and $t$ independently on the other coordinates, and to ignore in calculations the term $r^2\de\sigma^2$, which we reintroduce at the end.

As $r\to 2m$, $g_{tt}=-\(1-\dsfrac{2m}{r}\)^{-1}\to\infty$. This is a coordinate singularity, and not a genuine one, as shown by the Eddington-Finkelstein coordinates (\cite{eddington1924comparison,finkelstein1958past}, \citep{HE95}{150}).

On the other hand, $r\to 0$ is a genuine singularity, because $R_{abcd}R^{abcd}\to\infty$. This seems to suggest that the {\schw} metric cannot be made smooth at $r=0$. We will see that, in fact, there are coordinate systems in which the components of the metric are are analytic, hence finite, even $r=0$. The singularity remains, because the metric is degenerate. We find analytic and \textit{\semireg} extension of the {\schw} spacetime.

\subsection{Analytic extension of the {\schw} spacetime}
\label{s_schw_analytic}

\begin{theorem}
\label{thm_schw_analytic}
Let's consider the {\schw} spacetime, with the metric \eqref{eq_schw_schw}. There is a singular semi-Riemannian spacetime, which extends it analytically beyond the singularity $r=0$.
\end{theorem}
\begin{proof}
It suffices to find new coordinates in the region $r<2m$, which is a neighborhood of the singularity. On the region $r<2m$, the coordinate $t$ is spacelike, and $r$ is timelike. We use the transformation

\begin{equation}
\label{eq_coordinate_analytic}
\begin{array}{l}
\left\{
\begin{array}{ll}
r &= \tau^2 \\
t &= \xi\tau^\CT \\
\end{array}
\right.,
\\
\end{array}
\end{equation}
where $\CT\in\N$, to be conveniently determined.
Then, we have
\begin{equation}
\label{eq_coordinate_jacobian_schw}
\dsfrac{\partial r}{\partial \tau} = 2\tau,\,
\dsfrac{\partial r}{\partial \xi} = 0,\,
\dsfrac{\partial t}{\partial \tau} = \CT\xi\tau^{\CT-1},\,
\dsfrac{\partial t}{\partial \xi} = \tau^\CT.
\end{equation}
Recall that the metric coefficients in the {\schw} coordinates are
\begin{equation}
\label{eq_metric_coeff_schw}
g_{tt} = \dsfrac{2m-r}{r},\,
g_{rr} = \dsfrac{r}{r-2m},\,
g_{tr} = g_{rt} = 0.
\end{equation}
Let's calculate the metric coefficients in the new coordinates. We will do explicitly one such calculation:
\begin{equation*}
\begin{array}{lll}
g_{\tau\tau} &=& \(\dsfrac{\partial r}{\partial \tau}\)^2\dsfrac{r}{r-2m} + \(\dsfrac{\partial t}{\partial \tau}\)^2\dsfrac{2m-r}{r} \\
&=& 4\tau^2\dsfrac{\tau^2}{\tau^2 - 2m} + \CT^2\xi^2\tau^{2\CT-2}\dsfrac{2m-\tau^2}{\tau^2}\\
\end{array}
\end{equation*}
Hence
\begin{equation}
\label{eq_metric_coeff_analytic_tau_tau_tau_xi}
g_{\tau\tau} = -\dsfrac{4\tau^4}{2 m - \tau^2} + \CT^2\xi^2(2m - \tau^2) \tau^{2\CT-4}
\end{equation}
Similarly,
\begin{equation}
\label{eq_metric_coeff_analytic_tau_xi_tau_xi}
g_{\tau\xi} = \CT\xi(2m - \tau^2) \tau^{2\CT-3},
\end{equation}
\begin{equation}
\label{eq_metric_coeff_analytic_xi_xi_tau_xi}
g_{\xi\xi} = (2m - \tau^2) \tau^{2\CT-2}.
\end{equation}
The determinant is
\begin{equation}
\label{eq_metric_analytic_det_g_tau_xi}
\det g = - 4 \tau^{2\CT+2}.
\end{equation}

The four-metric is
\begin{equation}
\label{eq_schw_analytic_tau_xi}
\de s^2 = -\dsfrac{4\tau^4}{2m-\tau^2}\de \tau^2 + (2m-\tau^2)\tau^{2\CT-4}\(\CT\xi\de\tau + \tau\de\xi\)^2 + \tau^4\de\sigma^2,
\end{equation}
and it is easy to see that it is analytic for $\CT\geq 2$.
\end{proof}

\subsection{{\ssemireg} extension of the {\schw} spacetime}
\label{s_schw_semireg}

While proving Theorem \ref{thm_schw_analytic}, we found an infinite family of coordinates which make the metric analytic. We prove now that among these solutions there is a {\semireg} one.

\begin{theorem}
\label{thm_schw_semireg}
The {\schw} metric admits an analytic extension in which the singularity at $r=0$ is {\semireg}.
\end{theorem}
\begin{proof}
We start with the coordinate transformation \eqref{eq_coordinate_analytic}.
To show that the metric is {\semireg}, {\cf} Proposition \ref{thm_sr_cocontr_kosz}, it is enough to find $\CT\in\N$ so that there is a coordinate system in which the products of the form
\begin{equation}
\label{eq_semireg_condition_coord}
g^{st}\Gamma_{abs}\Gamma_{cdt}
\end{equation}
are all smooth \cite{Sto11a}, where $\Gamma_{abc}$ are Christoffel's symbols of the first kind. In a coordinate system in which the metric is smooth, as in \sref{s_schw_analytic}
, Christoffel's symbols of the first kind are also smooth. But the inverse metric $g^{st}$ is not smooth for $r=0$. We will show that the products from the expression \eqref{eq_semireg_condition_coord} are smooth.

We use the solution from Theorem \ref{thm_schw_analytic}, and try to find a value for $\CT$, so that the metric is {\semireg}.

The components of the inverse of the metric are given by $g^{\tau\tau} = g_{\xi\xi}/{\det g}$, $g^{\xi\xi} = g_{\tau\tau}/{\det g}$, and $g^{\tau\xi} = g^{\xi\tau} = -g_{\tau\xi}/{\det g}$. From (\ref{eq_metric_coeff_analytic_tau_tau_tau_xi}--\ref{eq_metric_coeff_analytic_xi_xi_tau_xi}), follows that

\begin{equation}
\label{eq_metric_inv_coeff_semireg_tau_tau}
g^{\tau\tau} = -\dsfrac{1}{4}(2m - \tau^2)\tau^{-4},
\end{equation}
\begin{equation}
\label{eq_metric_inv_coeff_semireg_tau_xi}
g^{\tau\xi} = \dsfrac{1}{4}\CT\xi(2m - \tau^2)\tau^{-5},
\end{equation}
\begin{equation}
\label{eq_metric_inv_coeff_semireg_xi_xi}
g^{\xi\xi} = \dsfrac{\tau^{-2\CT+2}}{2 m - \tau^2} - \dsfrac{1}{4}\CT^2 \xi^2 (2m - \tau^2)\tau^{-6}.
\end{equation}
Christoffel's symbols of the first kind are given by
\begin{equation}
\label{eq_christoffel_schw}
\Gamma_{abc} = \dsfrac 1 2 \(\partial_a g_{bc} + \partial_b g_{ca} - \partial_c g_{ab}\),
\end{equation}
so we have to calculate the partial derivatives of the coefficients of the metric.

From \eqref{eq_coordinate_jacobian_schw} and (\ref{eq_metric_coeff_analytic_tau_tau_tau_xi}--\ref{eq_metric_coeff_analytic_xi_xi_tau_xi}) we have:
\begin{equation*}
\begin{array}{lll}
\partial_{\tau}g_{\tau\tau} &=& \partial_{\tau}\(-\dsfrac{4\tau^4}{2 m - \tau^2} + \xi^2\CT^2(2m - \tau^2) \tau^{2\CT-4}\) \\
&=& -4\dsfrac{4\tau^3(2 m - \tau^2) + 2\tau^5}{(2 m - \tau^2)^2} + 2\CT^2(2\CT-4) m \xi^2\tau^{2\CT-5} \\
&&- \CT^2 (2\CT-2)\xi^2\tau^{2\CT-3}, \\
\end{array}
\end{equation*}
hence
\begin{equation}
\label{eq_pd_tau_tau_tau}
\partial_{\tau}g_{\tau\tau} = 8\dsfrac{\tau^5 - 4m\tau^3}{(2 m - \tau^2)^2} + 2\CT^2(2\CT-4) m \xi^2\tau^{2\CT-5}- \CT^2 (2\CT-2)\xi^2\tau^{2\CT-3}.
\end{equation}
Similarly,
\begin{equation}
\label{eq_pd_tau_tau_xi}
\partial_{\tau}g_{\tau\xi} = 2 \CT (2\CT-3) m\xi\tau^{2\CT-4} - \CT (2\CT-1) \xi\tau^{2\CT-2},
\end{equation}
\begin{equation}
\label{eq_pd_tau_xi_xi}
\partial_{\tau}g_{\xi\xi} = 2 m (2\CT-2)\tau^{2\CT-3} - 2\CT\tau^{2\CT-1},
\end{equation}
\begin{equation}
\label{eq_pd_xi_tau_tau}
\partial_{\xi}g_{\tau\tau} = 2 \CT^2 \xi (2m - \tau^2)\tau^{2\CT-4},
\end{equation}
\begin{equation}
\label{eq_pd_xi_tau_xi}
\partial_{\xi}g_{\tau\xi} = \CT(2m - \tau^2)\tau^{2\CT-3},
\end{equation}
and
\begin{equation}
\label{eq_pd_xi_xi_xi}
\partial_{\xi}g_{\xi\xi} = 0.
\end{equation}

To ensure that the expression \eqref{eq_semireg_condition_coord} is smooth, we try to find a value of $\CT$ for which it doesn't contain negative powers of $\tau$. The least power of $\tau$ in equations (\ref{eq_pd_tau_tau_tau}--\ref{eq_pd_xi_xi_xi}) is $\min(3,2\CT-5)$, as we can see by inspection. The least power of $\tau$ in equations (\ref{eq_metric_inv_coeff_semireg_tau_tau}--\ref{eq_metric_inv_coeff_semireg_xi_xi}) is $\min(-6,-2\CT+2)$. Since $\min(-6,-2\CT+2) = -3 - \max(3,2\CT-5)$, the conjunction of the two conditions is
\begin{equation}
	-1 - 2\CT + 3\min(3,2\CT-5) \geq 0
\end{equation}
with the unique solution
\begin{equation}
	\CT = 4.
\end{equation}
which ensures the smoothness of \eqref{eq_semireg_condition_coord}, and by this, the {\semireg}ity of the metric in two dimensions $(\tau,\xi)$.
The coordinate transformation becomes
\begin{equation}
\label{eq_coordinate_semireg}
\begin{array}{l}
\left\{
\begin{array}{ll}
r &= \tau^2 \\
t &= \xi\tau^\CT \\
\end{array}
\right.,
\\
\end{array}
\end{equation}

When going back to four dimensions, we have to take the {\semireg} warped product between the two-dimensional extension $(\tau,\xi)$ and the sphere $S^2$, with warping function $\tau^2$, which is {\semireg}, {\cf} Theorem \ref{thm_semi_reg_semi_riem_man_warped}.
\end{proof}

The metric in the proof is
\begin{equation}
\label{eq_schw_semireg}
\de s^2 = -\dsfrac{4\tau^4}{2m-\tau^2}\de \tau^2 + (2m-\tau^2)\tau^4\(4\xi\de\tau + \tau\de\xi\)^2 + \tau^4\de\sigma^2,
\end{equation}
which is a {\semireg} analytic extension of the {\schw} metric.

It is easy to see that it is also {\quasireg}.
\begin{corollary}
\label{thm_schw_quasireg}
The {\schw} spacetime is {\quasireg} (in any atlas compatible with the coordinates \eqref{eq_coordinate_semireg}).
\end{corollary}
\begin{proof}
We know from \cite{Sto11e} that the {\schw} spacetime is {\semireg}. It is also Ricci flat, hence $S_{abcd}$ and $E_{abcd}$ are smooth too, and the metric is {\quasireg}.
\end{proof}

\section{Analytic Reissner-Nordstrom singularity}
\label{s_black_hole_rn}

In this section, we derive an analytic extension of the Reissner-Nordstr\"om solution beyond the singularity, by using new coordinates. The metric's components are made finite and analytic, but degenerate at $r=0$.

\subsection{Introduction}

The {\rn} is a solution to the Einstein-Maxwell equations, describing a static, spherically symmetric, electrically charged, non-rota\-ting black hole \cite{reiss16,nord18}. The metric is
\begin{equation}
\label{eq_rn_metric}
\de s^2 = -\left(1-\dsfrac{2m}{r} + \dsfrac{q^2}{r^2}\right)\de t^2 + \left(1-\dsfrac{2m}{r} + \dsfrac{q^2}{r^2}\right)^{-1}\de r^2 + r^2\de\sigma^2,
\end{equation}
where $q$ is the electric charge of the body, $m$ the mass of the body, $\de\sigma^2$ is that from equation \eqref{eq_sphere}, and the units were chosen so that $c=1$ and $G=1$ ({\cf} \citep{HE95}{156}).

The number of the event horizons of a {\rn} black hole equals the number of real zeros of $r^2 - 2mr + q^2$. These are apparent singularities, that can be removed for example by Eddington-Finkelstein coordinates \cite{eddington1924comparison,finkelstein1958past}. But in all cases, there is an irremovable singularity at $r=0$, which at least we will make degenerate.

\subsection{Extending the {\rn} spacetime at the singularity}
\label{s_rn_ext_ext}

The main result of this section is contained in the following theorem.

\begin{theorem}
\label{thm_rn_ext_ext}
Let's consider the {\rn} spacetime, with the metric \eqref{eq_rn_metric}. There is a larger spacetime, which extends it analytically beyond the singularity $r=0$.
\end{theorem}
\begin{proof}
The metric \eqref{eq_rn_metric} is a warped product between a two-dimensional {\semiriem} space of coordinates $(t,r)$, and the sphere $\CS^2$, parameterized by $\phi$ and $\theta$. As in the {\schw} case, we will use coordinate transformations which affect only the coordinates $r$ and $t$, in the neighborhood $r\in[0,M)$ of the singularity, where $M=\infty$ if the black hole is naked, and $M$ is the smallest real zero of $r^2 - 2mr + q^2$ otherwise. We take the coordinates $\rho$ and $\tau$ given by
\begin{equation}
\label{eq_coordinate_ext_ext}
\begin{array}{l}
\left\{
\begin{array}{ll}
t &= \tau\rho^\CT \\
r &= \rho^\CS \\
\end{array}
\right.
\\
\end{array}
\end{equation}
and search for $\CS,\CT\in\N$ which make the metric analytic. This choice is motivated by the need to smoothen the divergent components of the metric.
Then,
\begin{equation}
\label{eq_coordinate_jacobian_rn}
\dsfrac{\partial t}{\partial \tau} = \rho^\CT,\,
\dsfrac{\partial t}{\partial \rho} = \CT\tau\rho^{\CT-1},\,
\dsfrac{\partial r}{\partial \tau} = 0,\,
\dsfrac{\partial r}{\partial \rho} = \CS\rho^{\CS-1}.
\end{equation}
The following notation is standard
\begin{equation}
	\Delta := r^2 - 2m r + q^2 \tn{ (hence $\Delta = \rho^{2\CS} - 2m \rho^{\CS} + q^2$)}.
\end{equation}
We see that $\Delta> 0$ for $\rho\in[0,M)$.
The metric components from \eqref{eq_rn_metric} become
\begin{equation}
\label{eq_metric_coeff_rn}
g_{tt} = - \dsfrac{\Delta}{\rho^{2\CS}},\,
g_{rr} = \dsfrac{\rho^{2\CS}}{\Delta},\,
g_{tr} = g_{rt} = 0.
\end{equation}
We calculate the metric components in the new coordinates. 
\begin{equation*}
\begin{array}{lll}
g_{\tau\tau} &=& \left(\dsfrac{\partial r}{\partial \tau}\right)^2\dsfrac{\rho^{2\CS}}{\Delta} - \left(\dsfrac{\partial t}{\partial \tau}\right)^2\dsfrac{\Delta}{\rho^{2\CS}} \\
&=& 0 - \rho^{2\CT}\dsfrac{\Delta}{\rho^{2\CS}}. \\
\end{array}
\end{equation*}
Hence
\begin{equation}
\label{eq_metric_coeff_ext_ext_tau_tau}
g_{\tau\tau} = -\Delta\rho^{2\CT-2\CS}.
\end{equation}
\begin{equation*}
\begin{array}{lll}
g_{\rho\tau} &=& \dsfrac{\partial r}{\partial \rho}\dsfrac{\partial r}{\partial \tau}\dsfrac{\rho^{2\CS}}{\Delta} - \dsfrac{\partial t}{\partial \rho}\dsfrac{\partial t}{\partial \tau}\dsfrac{\Delta}{\rho^{2\CS}} \\
&=& 0 - \CT\tau\rho^{2\CT-1}\dsfrac{\Delta}{\rho^{2\CS}}.
\end{array}
\end{equation*}
Then
\begin{equation}
\label{eq_metric_coeff_ext_ext_rho_tau}
g_{\rho\tau} = - \CT\Delta\tau\rho^{2\CT-2\CS-1},
\end{equation}
and
\begin{equation}
\label{eq_metric_coeff_ext_ext_rho_rho}
g_{\rho\rho} = \CS^2\dsfrac{\rho^{4\CS-2}}{\Delta} - \CT^2\Delta\tau^2\rho^{2\CT-2\CS-2}.
\end{equation}

From the term $\CS^2\dsfrac{\rho^{4\CS-2}}{\Delta}$ in equation \eqref{eq_metric_coeff_ext_ext_rho_rho}, to have smooth $g_{\rho\rho}$, $\CS$ has to be an integer so that $\CS\geq 1$. This condition also makes this term analytic at $\rho=0$.

To avoid negative powers of $\rho$, the other term in equation \eqref{eq_metric_coeff_ext_ext_rho_rho}, and those from \eqref{eq_metric_coeff_ext_ext_tau_tau} and \eqref{eq_metric_coeff_ext_ext_rho_tau}, require that $2\CT - 2\CS - 2 \geq 0$. Hence, to remove the infinity of the metric at $r=0$ by a coordinate transformation like \eqref{eq_coordinate_ext_ext}, $\CS$ and $\CT$ has to be integers, and
\begin{equation}
\label{eq_metric_smooth_cond}
\begin{array}{l}
\left\{
\begin{array}{ll}
\CS \geq 1 \\
\CT \geq \CS + 1	
\end{array}
\right.
\\
\end{array}
\end{equation}
These conditions ensure that all metric's components are analytic at the singularity.

Now we take the warped product between the space $(\tau,\rho)$, and the sphere $S^2$, with warping function $\rho^{\CS}$, to return to four dimensions. This degenerate warped product results in a $4D$ manifold whose metric is analytic and degenerate at $\rho=0$.
\end{proof}

\begin{corollary}
In the coordinates from theorem \ref{thm_rn_ext_ext}, the {\rn} metric is
\begin{equation}
\label{eq_rn_ext_ext}
\de s^2 = - \Delta\rho^{2\CT-2\CS-2}\left(\rho\de\tau + \CT\tau\de\rho\right)^2 + \dsfrac{\CS^2}{\Delta}\rho^{4\CS-2}\de\rho^2 + \rho^{2\CS}\de\sigma^2.
\end{equation}
\end{corollary}
\begin{proof}
From \eqref{eq_coordinate_jacobian_rn} follows
\begin{equation}
\label{eq_de_t}
	\de t = \dsfrac{\partial t}{\partial \tau}\de\tau + \dsfrac{\partial t}{\partial \rho}\de\rho = \rho^\CT\de\tau + \CT\tau\rho^{\CT-1}\de\rho = \rho^{\CT-1}(\rho\de\tau + \CT\tau\de\rho)
\end{equation}
and
\begin{equation}
\label{eq_de_r}
	\de r = \dsfrac{\partial r}{\partial \tau}\de\tau + \dsfrac{\partial r}{\partial \rho}\de\rho = \CS\rho^{\CS-1}\de\rho.
\end{equation}
We substitute them in the {\rn} equation \eqref{eq_rn_metric}, and obtain
\begin{equation*}
\begin{array}{lll}
	\de s^2 &=& - \dsfrac{\Delta}{\rho^{2\CS}}\de t^2 + \dsfrac{\rho^{2\CS}}{\Delta}\de r^2 + r^2\de\sigma^2 \\
	&=& -\Delta\rho^{2\CT-2\CS-2}(\rho\de\tau + \CT\tau\de\rho)^2 + \dsfrac{\CS^2}{\Delta}\rho^{4\CS-2}\de\rho^2 + \rho^{2\CS}\de\sigma^2.
\end{array}
\end{equation*}
\end{proof}

\subsection{Null geodesics in the proposed solution}
\label{s_rn_ext_null_geodesics}

Let's discuss the geometric meaning of the extension proposed here, from the viewpoint of the lightcones and the null geodesics. In the coordinates $(\tau,\rho)$, the metric is analytic near the singularity $\rho=0$ and has the form
\begin{equation}
\label{eq_rn_metric_tau_rho_matrix}
g = -\Delta\rho^{2\CT-2\CS-2}\left(
\begin{array}{ll}
    \rho^2 & \CT\tau\rho \\
    \CT\tau\rho & \CT^2\tau^2 - \dsfrac{\CS^2}{\Delta^2}\rho^{6\CS-2\CT} \\
\end{array}
\right)
\end{equation}

Let's find the null directions at each point, {\ie} the tangent vectors $u\neq 0$ satisfying $g(u, u)=0$. If we consider $u=(\sin\alpha,\cos\alpha)$, we have to find $\alpha$. We obtain the equation
\begin{equation}
	\rho^2\sin^2\alpha + 2\CT\tau\rho\sin\alpha\cos\alpha + \left(\CT^2\tau^2 - \dsfrac{\CS^2}{\Delta^2}\rho^{6\CS-2\CT}\right)\cos^2\alpha = 0,
\end{equation}
which is quadratic in $\tan\alpha$
\begin{equation}
	\rho^2 \tan^2\alpha + 2\CT\tau\rho \tan\alpha + \left(\CT^2\tau^2 - \dsfrac{\CS^2}{\Delta^2}\rho^{6\CS-2\CT}\right) = 0,
\end{equation}
and has to the solutions
\begin{equation}
\label{eq_null_vectors}
	\tan\alpha_\pm = -\dsfrac{\CT\tau}{\rho} \pm \dsfrac{\CS}{\Delta}\rho^{3\CS-\CT-1}.
\end{equation}
Hence, the incoming and outgoing null geodesics satisfy the differential equation
\begin{equation}
\label{eq_null_geodesics}
	\dsfrac{\de\tau}{\de\rho} = -\dsfrac{\CT\tau}{\rho} \pm \dsfrac{\CS}{\Delta}\rho^{3\CS-\CT-1}.
\end{equation}
For the coordinate $\rho$ to remain spacelike, it has to satisfy $g_{\rho\rho}>0$. From equation \eqref{eq_rn_metric_tau_rho_matrix}, this requires that
\begin{equation}
\label{eq_rho_spacelike_condition}
	\dsfrac{\CS^2}{\Delta^2}\rho^{6\CS-2\CT} > \CT^2\tau^2.
\end{equation}
To ensure \eqref{eq_rho_spacelike_condition} in a neighborhood of $(0,0)$, it is required that
\begin{equation}
\label{eq_rho_spacelike_condition_T}
	\CT \geq 3\CS.
\end{equation}

The null geodesics are integral curves of the null vectors \eqref{eq_null_vectors}. In coordinates $(\tau,\rho)$, the null geodesics are oblique everywhere, except at $\rho=0$, where they become tangent to the axis $\rho=0$. The degeneracy of the metric stretch the lightcones as approaching $\rho=0$, where they become degenerate (figure \ref{lightcones}). At $\rho=0$, the incoming and outgoing null geodesics become tangent (figure \ref{null-geodesics}).

\image{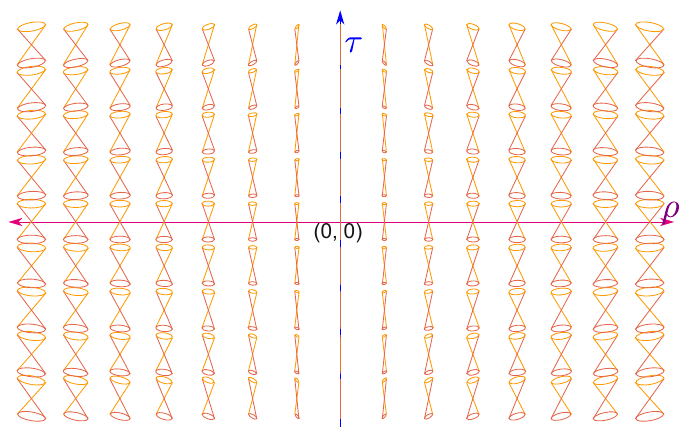}{0.7}{For $\CT\geq 3\CS$ and even $\CS$, as $\rho\to 0$, the lightcones become more and more degenerate along the axis $\rho=0$.}

\image{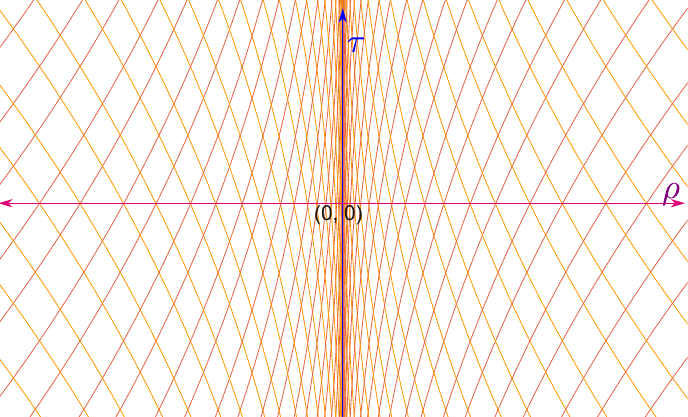}{0.7}{The null geodesics, in the $(\tau,\rho)$ coordinates, for $\CT\geq 3\CS$ and even $\CS$.}

\subsection{The electromagnetic field}
\label{s_rn_ext_ext_electromagnetic}

In the standard {\rn} solution, in coordinates $(t,r,\phi,\theta)$, the potential of the electromagnetic field is singular at $r=0$, being
\begin{equation}
A = -\dsfrac q r \de t.
\end{equation}
On the other hand, in our solution both the electromagnetic potential, and the electromagnetic field, have better behavior.
\begin{corollary}
\label{rem_rn_electromagnetic_potential}
In the coordinates $(\tau,\rho,\phi,\theta)$, given by the coordinate change \eqref{eq_coordinate_ext_ext}, under the condition \eqref{eq_metric_smooth_cond}, the electromagnetic potential is
\begin{equation}
A = -q\rho^{\CT-\CS-1}\left(\rho\de\tau + \CT\tau\de\rho\right),
\end{equation}
the electromagnetic field is
\begin{equation}
F = q(2\CT-\CS)\rho^{\CT-\CS-1}\de\tau \wedge\de\rho.
\end{equation}
They are finite, and analytic everywhere, including at the singularity $\rho=0$.
\end{corollary}
\begin{proof}
The electromagnetic potential is obtained directly from the proof of theorem \ref{thm_rn_ext_ext} and equation \eqref{eq_de_t}. To obtain the electromagnetic field, we apply the exterior derivative:
\begin{equation*}
\begin{array}{lll}
F &=& \de A = -q\de\left(\rho^{\CT-\CS}\de\tau + \CT\tau\rho^{\CT-\CS-1}\de\rho\right) \\
&=& -q \left(\dsfrac{\partial\rho^{\CT-\CS}}{\partial \rho} \de\rho\wedge\de\tau + \CT\dsfrac{\partial\tau\rho^{\CT-\CS-1}}{\partial\tau}\de\tau\wedge\de\rho\right) \\
&=& -q \left((\CT-\CS)\rho^{\CT-\CS-1} \de\rho\wedge\de\tau + \CT \rho^{\CT-\CS-1}\de\tau\wedge\de\rho\right) \\
&=& q(2\CT-\CS)\rho^{\CT-\CS-1}\de\tau\wedge\de\rho
\end{array}
\end{equation*}
\end{proof}

\section{Kerr-Newman solutions with analytic singularity}
\label{s_black_hole_kn}

The stationary and axisymmetric solutions of the Einstein-Maxwell equations, representing charged rotating black holes, are called {\kn} solutions \cite{new65,Wal84}.
According to the {\em no-hair theorem}, the non-stationary black holes tend asymptotically to {\kn} solutions, which makes them representative for all the black holes.
Therefore, it is important to understand them, especially the properties in general considered undesirable. The first is the existence of the singularity, which is in general ring-shaped. Another problem is that, passing through the ring, one can reach inside another universe, in which there are {\em closed timelike curves}, {\ie} time machines (which fortunately don't affect the causality in the region $r>0$). A no less important problem is the {\em black hole information paradox}, which will be discussed later.

We show that the {\kn} solution, representing a charged and rotating stationary black hole, admits analytic extension beyond the singularity. We use again new coordinates, in which the metric becomes smooth, although degenerate, on the ring singularity. Unlike fort the standard solution, one can choose the extension so that there are no closed timelike curves.

\subsection{Introduction}

The {\kn} space is $\R\times\R^3$, where $\R$ is the time coordinate, and $\R^3$ is parameterized by the spherical coordinates $(r,\phi,\theta)$. The parameter $a\geq 0$ characterizes the rotation, $m\geq 0$ the mass, $q\in\R$ the charge. Let's define the functions
\begin{equation*}
\Sigma(r,\theta) := r^2 + a^2 \cos^2 \theta,\,
\Delta(r) := r^2 - 2 m r + a^2 + q^2.
\end{equation*}
The {\kn} {\metricname} is
\begin{equation*}
g_{tt} = - \frac{\Delta(r) - a^2\sin^2\theta}{\Sigma(r,\theta)},\,	
g_{rr} = \frac{\Sigma(r,\theta)}{\Delta(r)},\,
g_{\theta\theta} = \Sigma(r,\theta),
\end{equation*}
\begin{equation}
\label{eq_metric_kn}
g_{\phi\phi} = \frac{(r^2 + a^2)^2 -\Delta(r) a^2\sin^2\theta}{\Sigma(r,\theta)}\sin^2\theta,\,
g_{t\phi} = g_{\phi t} = - \frac{2a\sin^2\theta(r^2 + a^2 - \Delta(r))}{\Sigma(r,\theta)}.
\end{equation}
All other components of the {\metricname} vanish \cite{Wal84}.

If $q=0$, we obtain the Kerr solution \cite{kerr63,ks65}; if $a=0$ we get the {\rn} solution \cite{reiss16,nord18}. If both $q=0$ and $a=0$, we obtain the {\schw} solution; if, in addition, $m=0$, the result is the empty Minkowski spacetime ({\cf} Table \ref{tab_static_black_holes}).
\begin{table}[htb!]
\centering
\begin{tabular}{|l|l|l|}
\hline
& $\mathbf{a>0}$ & $\mathbf{a=0}$ \\\hline
$\mathbf{q\neq 0}$ & {\kn} & {\rn} \\\hline
$\mathbf{q=0}$ & Kerr & {\schw} \\\hline
\end{tabular}
\vspace{0.15in}
\caption{Other stationary black hole solutions are obtained from the {\kn}.}
\label{tab_static_black_holes}
\end{table}

\subsection{Extending the {\kn} spacetime at the singularity}
\label{s_kn_ext_ext}

\begin{theorem}
\label{thm_kn_ext_ext}
Let's consider the {\kn} spacetime, with the metric \eqref{eq_metric_kn}. There is a larger spacetime, which extends it analytically beyond the singularity $r=0$.
\end{theorem}
\begin{proof}
We will find a coordinate system in the neighborhood of the singularity usually called the block III (\cite{ONe95}, p. 66), defined as follows. If the equation $\Delta(r)=0$ has no real solutions, the region is $r<\infty$, otherwise, it is $r<r_-$, where $r_-$ is the smallest solution.

The new coordinates are $\tau$, $\rho$, and $\mu$, defined by
\begin{equation}
\label{eq_coordinate_ext_ext_kn}
\begin{array}{l}
\left\{
\begin{array}{ll}
t &= \tau\rho^{\CT}, \\
r &= \rho^{\CS}, \\
\phi &= \mu\rho^{\CM}, \\
\theta &= \theta, \\
\end{array}
\right.
\\
\end{array}
\end{equation}
where ${\CS},{\CT},{\CM}\in\N$ are natural numbers, chosen so that the {\metricname} is analytic.
When passing from coordinates $(x^{a})$ to new coordinates $(x^{a'})$, the metric becomes
\begin{equation}
\label{eq_metric_coord_change}
g_{a'b'} = \dsfrac{\partial x^{a}}{\partial x^{a'}}\dsfrac{\partial x^{b}}{\partial x^{b'}}	g_{ab}.
\end{equation}
For our transformation \eqref{eq_coordinate_ext_ext_kn}, the Jacobian is
\begin{equation}
\label{eq_coordinate_change}
\begin{array}{l}
\dsfrac{\partial(t,r,\phi,\theta)}{\partial(\tau,\rho,\mu,\theta)}=
\left(
\begin{array}{llll}
\dsfrac{\partial t}{\partial \tau} & \dsfrac{\partial t}{\partial \rho} & \dsfrac{\partial t}{\partial \mu} & \dsfrac{\partial t}{\partial \theta}  \\
\dsfrac{\partial r}{\partial \tau} & \dsfrac{\partial r}{\partial \rho} & \dsfrac{\partial r}{\partial \mu} & \dsfrac{\partial r}{\partial \theta}  \\
\dsfrac{\partial \phi}{\partial \tau} & \dsfrac{\partial \phi}{\partial \rho} & \dsfrac{\partial \phi}{\partial \mu} & \dsfrac{\partial \phi}{\partial \theta}  \\
\dsfrac{\partial \theta}{\partial \tau} & \dsfrac{\partial \theta}{\partial \rho} & \dsfrac{\partial \theta}{\partial \mu} & \dsfrac{\partial \theta}{\partial \theta}  \\
\end{array}
\right)
=
\left(
\begin{array}{llll}
\rho^{\CT} & {\CT} \tau \rho^{{\CT}-1} & 0 & 0 \\
0 & {\CS}\rho^{{\CS}-1} & 0 & 0  \\
0 & {\CM} \mu \rho^{{\CM}-1} & \rho^{\CM} & 0 \\
0 & 0 & 0 & 1 \\
\end{array}
\right).
\end{array}
\end{equation}
We arrange the components in Table \ref{tab_jacobian}.
\begin{table}[htb!]
\centering
\begin{tabular}{|p{1.5 cm}||p{1.5 cm}|p{1.5 cm}|p{1.5 cm}|p{1.5 cm}|}
\hline
& $\cdot /\partial \tau$ & $\cdot /\partial \rho$ & $\cdot /\partial \mu$ & $\cdot /\partial \theta$ \\\hline
\hline
$\partial t/\cdot$ & $\rho^{\CT}$ & ${\CT} \tau \rho^{{\CT}-1}$ & $0$ & $0$ \\\hline
$\partial r/\cdot$ & $0$ & ${\CS}\rho^{{\CS}-1}$ & $0$ & $0$ \\\hline
$\partial \phi/\cdot$ & $0$ & ${\CM} \mu \rho^{{\CM}-1}$ & $\rho^{\CM}$ & $0$ \\\hline
$\partial \theta/\cdot$ & $0$ & $0$ & $0$ & $1$ \\\hline
\end{tabular}
\vspace{0.15in}
\caption{The Jacobian components.}
\label{tab_jacobian}
\end{table}

To make sure that the new expression of the {\metricname} becomes smooth even on the ring singularity, we want that all the terms in the right hand side of equation \eqref{eq_metric_coord_change} are smooth. For this, that the Jacobian components has to cancel the singularities of the {\metricname}'s components, even when $\cos \theta=0$.

Here are the least powers of $\rho$ on the ring singularity, in each of the {\metricname}'s components
\begin{equation}
\label{eq_ord_g_t_t}
\rhord{g_{tt}} = - 2{\CS},
\end{equation}
\begin{equation}
\label{eq_ord_g_phi_phi}
\rhord{g_{\phi\phi}} = -2{\CS},
\end{equation}
\begin{equation}
\label{eq_ord_g_t_phi}
\rhord{g_{t\phi}} = \rhord{g_{\phi t}} = - 2{\CS},
\end{equation}
obtained by dividing polynomial expressions in $\rho$ by $\Sigma(r,\theta)$. All other components are non-singular on the ring singularity.

Table \ref{tab_jacobian_ord} contains the least power of $\rho$ present in Table \ref{tab_jacobian}.
\begin{table}[htb!]
\centering
\begin{tabular}{|p{1.5 cm}||p{1.5 cm}|p{1.5 cm}|p{1.5 cm}|p{1.5 cm}|}
\hline
& $\cdot /\partial \rho$ & $\cdot /\partial \tau$ & $\cdot /\partial \mu$ & $\cdot /\partial \theta$ \\\hline
\hline
$\partial t/\cdot$ & ${\CT}$ & ${\CT}-1$ & $0$ & $0$ \\\hline
$\partial r/\cdot$ & $0$ & ${\CS}-1$ & $0$ & $0$ \\\hline
$\partial \phi/\cdot$ & $0$ & ${\CM}-1$ & ${\CM}$ & $0$ \\\hline
$\partial \theta/\cdot$ & $0$ & $0$ & $0$ & $0$ \\\hline
\end{tabular}
\vspace{0.15in}
\caption{The least power of $\rho$ in the Jacobian components of the coordinate change.}
\label{tab_jacobian_ord}
\end{table}

We now check the {\metricname}'s components, to see if the singularities are canceled by the components of the Jacobian.

Each component $g_{ab}$ of the {\metricname} is checked by looking up the rows labeled by $\partial x^a/\cdot$ and $\partial x^b/\cdot$ in Table \ref{tab_jacobian_ord}.
For example,
\begin{equation}
	\dsfrac{\partial t}{\partial \rho}\dsfrac{\partial t}{\partial \tau}g_{tt}
\end{equation}
satisfies
\begin{equation}
	\rhord{\dsfrac{\partial t}{\partial \rho}\dsfrac{\partial t}{\partial \tau}g_{tt}} = ({\CT}  - 1) + {\CT} - 2{\CS},
\end{equation}
hence ${\CT}$ has to satisfy $2{\CT}\geq 2{\CS}+1$.

From \eqref{eq_ord_g_t_t}, \eqref{eq_ord_g_phi_phi}, and \eqref{eq_ord_g_t_phi}, follows that we have to do this only for the components of the {\metricname} with indices $t$ and $\phi$. From the equation \eqref{eq_metric_coord_change}, we need to check only the rows $\partial t/\cdot$ and $\partial \phi/\cdot$ from Table \ref{tab_jacobian_ord}. Hence, to cancel the singularities of the {\metricname}'s components, each component of the Jacobian of the form $\partial t/\cdot$ and $\partial \phi/\cdot$ has to contain $\rho$ to at least the power ${\CS}$. Therefore, the smoothness (and the analyticity for that matter) of the {\metricname} on the ring singularity, in the new coordinates, is ensured by the conditions
\begin{equation}
\label{eq_metric_smooth_cond_kn}
\begin{array}{l}
\left\{
\begin{array}{ll}
{\CS} &\geq 1 \\
{\CT} &\geq {\CS} + 1 \\
{\CM} &\geq {\CS} + 1,
\end{array}
\right.
\\
\end{array}
\end{equation}
where ${\CS},{\CT},{\CM}\in\N$.
Hence, in the new coordinates, all of the {\metricname}'s components are analytic at the singularity.
\end{proof}

\subsection{Electromagnetic field in \kn}
\label{s_kn_em}

Another feature of our extension is that the electromagnetic potential and electromagnetic field become smooth.

In standard coordinates, the electromagnetic potential of the {\kn} solution is the $1$-form
\begin{equation}
\label{eq_kn_electromagnetic_potential}
A = -\dsfrac{qr}{\Sigma(r,\theta)}(\de t - a\sin^2\theta\de\phi).
\end{equation}
In our coordinates, it is
\begin{equation}
\label{eq_kn_electromagnetic_potential_smooth}
A = -\dsfrac{q\rho^{\CS}}{\Sigma(r,\theta)}(\rho^{\CT}\de\tau + {\CT}\tau\rho^{{\CT}-1}\de\rho - a\sin^2\theta\rho^{\CM}\de\mu),
\end{equation}
because from the Table \ref{tab_jacobian} it follows that
\begin{equation}
\label{eq_de_t_kn}
	\de t = \rho^{\CT}\de\tau + {\CT}\tau\rho^{{\CT}-1}\de\rho
\end{equation}
\begin{equation}
\label{eq_de_r_kn}
	\de r = {\CS}\rho^{{\CS}-1}\de\rho
\end{equation}
and
\begin{equation}
\label{eq_de_phi}
	\de \phi = {\CM}\mu\rho^{{\CM}-1}\de\rho + \rho^{\CM}\de\mu.
\end{equation}

In our coordinates, there is no singularity of the electromagnetic potential $A$ at $\rho=0$ and $\cos\theta=0$, because ${\CT}> {\CS}$ and ${\CM}> {\CS}$, from the conditions \eqref{eq_metric_smooth_cond}. Because the electromagnetic field $F = \de A$, it is smooth too.

\chap{ch_global_hyperbolic}{Global hyperbolicity and black hole singularities}

The source of this chapter is author's original conference talk \cite{Sto12e}.
\label{s_black_hole_hyper}
Black hole singularities are usually considered to break the time evolution, in (classical) General Relativity. It is shown here that, in fact, the standard black hole solutions are compatible with global hyperbolicity. For this, globally hyperbolic spacetimes containing singularities are constructed.

The first step was to extend solutions of the Einstein equation to the singularities, in a way which avoids infinities. The stationary black hole solutions are then analytically extended beyond the singularities. Next, the topology of the spacetime at the singularity is repaired, by using the new extensions. Then, solutions are constructed, which are shown to be globally hyperbolic by foliating the spacetime with Cauchy hypersurfaces. The foliations are constructed by applying a Schwarz-Christoffel mapping to the Penrose-Carter representations of the black hole solutions.

The results are applied to non-stationary black holes, including evaporating ones.

\section{Canceling the singularities of the field equations}
\label{s_field_equations_singularity_removal}

Some of the tensor fields involved in Einstein's field equation become infinite at the singularity points. We need a method to replace them with other fields which obey equations equivalent to Einstein's at the non-singular points, but remain in the same time smooth at the singularity. Our proposed approach generalizes the {\semiriem} geometry, so that it works at an important class of singularity points \cite{Sto11a,Sto11b,Sto11d}.

Other method, used by Einstein and Rosen (for which they credited Mayer \cite{ER35}) is to multiply the Einstein equations with by a well-chosen power of $\det g$. This way, all instances of $g^{ab}$ in the expression of $\det g^2\ric$ are replaced by $\det g g^{ab}$, the adjugate matrix of $g_{ab}$ (\cite{ER35}, p. 74), and the singularities of the Riemann and the Ricci tensors can be canceled out for certain cases. The resulting equations still make sense as partial differential equations in the components of the metric tensor. Of course, the geometric objects of {\ssemiriem} Geometry have geometric meanings so long as the metric tensor is {\nonsing}. But once we allow it to become singular, the geometric objects like covariant derivative and curvature become undefined. The quantities which replace them if we multiply the equations with $\det g$ or other factors cannot give them this meaning. Maybe a new kind of geometry is required to restore the ideas of covariant derivative and curvature for singular metrics.

We preferred to have a generalization of  {\ssemiriem} Geometry, which allows us to define geometric objects like covariant derivative, Riemann curvature and Ricci tensor even when the metric becomes degenerate.

In the following, we will explore the arena on which are defined the fields -- the spacetime -- to make sure that the singularities don't alter the topology in a way which breaks down the evolution equations.

\subsection{The topology of singularities}
\label{s_singularities_topology}

In the following we shall see that the main black hole solutions of Einstein's equations can be interpreted so that the time evolution is not jeopardized.
To do this, we
\begin{enumerate}
	\item Remove the singularities from the field equations.
	\item Make sure that the solution is globally hyperbolic.
\end{enumerate}
In this section, we shall see that these conditions can be ensured for the typical black hole solutions of Einstein's equation.

\section{Globally hyperbolic spacetimes}
\label{s_globally_hyperbolic_spacetimes}

A global solution to the Einstein equation is well-behaved when the initial data at a given moment of time determine the solution for the entire future and past. This condition is ensured by the \textit{global hyperbolicity}, which is expressed by the requirements that
\begin{enumerate}
	\item 
for any two points $p$ and $q$, the intersection between the causal future of $p$, and the causal past of $q$, $J^+(p)\cap J^-(q)$, is a compact subset of the spacetime;
	\item 
	there are no closed timelike curves (\citep{HE95}{206}).
\end{enumerate}

The property of global hyperbolicity is equivalent to the existence of a \textit{Cauchy hypersurface} -- a spacelike hypersurface $\mathfrak S$ that, for any point $p$ in the future (past) of $\mathfrak S$, is intersected by all past-directed (future-directed) inextensible causal ({\ie} timelike or null) curves through the point $p$ (\citep{HE95}{119, 209--212}).

A method we will use to construct hyperbolic spacetimes out of our solutions describing black holes is the following. We will use Penrose-Carter coordinates, and prove the global hyperbolicity by explicitly constructing foliations in Cauchy hypersurfaces.

To obtain explicitly the foliations of the Penrose-Carter diagrams, we map to our solutions represented in coordinates $(\tau,\rho)$ the product $(0,1)\times\R$. To do this, we will use a version of the Schwarz-Christoffel mapping that maps the strip
\begin{equation}
\label{eq_strip}
\mc S:=\{z\in\CC|\tn{Im}(z)\in[0,1]\}
\end{equation}
to a polygonal region from $\CC$, with the help of the formula
\begin{equation}
	\label{eq_s_c_map}
	f(z)=A + C\int^{\mc S}\exp\left[\frac\pi 2(\alpha_--\alpha_+)\zeta\right]\prod_{k=1}^n\left[\sinh \frac\pi 2(\zeta-z_k)\right]^{\alpha_k-1}\de\zeta,
\end{equation}
where $z_k\in\partial\mc S:=\R\times\{0,i\}$ are the prevertices of the polygon, and $\alpha_-,\alpha_+,\alpha_k$ are the measures of the angles of the polygon, divided by $\pi$ (\cf\ \eg\ \cite{dri02}). The vertices having the angles $\alpha_-$ and $\alpha_+$ correspond to the ends of the strip, which are at infinity. The level curves $\{\tn{Im}(z)=\tn{const.}\}$ give our foliation \cite{Sto12e}.

\section{Globally hyperbolic {\schw} spacetime}
\label{s_globally_hyperbolic_spacetimes_schw}

\subsection{Penrose-Carter diagram for the {\semireg} solution}
\label{s_schw_semireg_penrose_carter}

The coordinates for the Penrose-Carter diagram (fig. \ref{std-schwarzschild}) of the {\schw} metric in its standard form are
\begin{equation}
\label{eq_schw_penrose_coord}
\begin{array}{l}
\left\{
\begin{array}{lll}
v''&=&\arctan\(\phantom{-}(2m)^{-1/2}\exp\(\phantom{-}\dsfrac{v}{4m}\)\) \\
w''&=&\arctan\(-(2m)^{-1/2}\exp\(-\dsfrac{w}{4m}\)\) \\
\end{array}
\right.
\\
\end{array}
\end{equation}
where $v,w$ are the Eddington-Finkelstein lightlike coordinates \cite{eddington1924comparison,finkelstein1958past}
\begin{equation}
\label{eq_schw_lightlike_coord}
\begin{array}{l}
\left\{
\begin{array}{lll}
v &=& t+r+2m\ln(r-2m) \\
w &=& t-r-2m\ln(r-2m). \\
\end{array}
\right.
\\
\end{array}
\end{equation}

\image{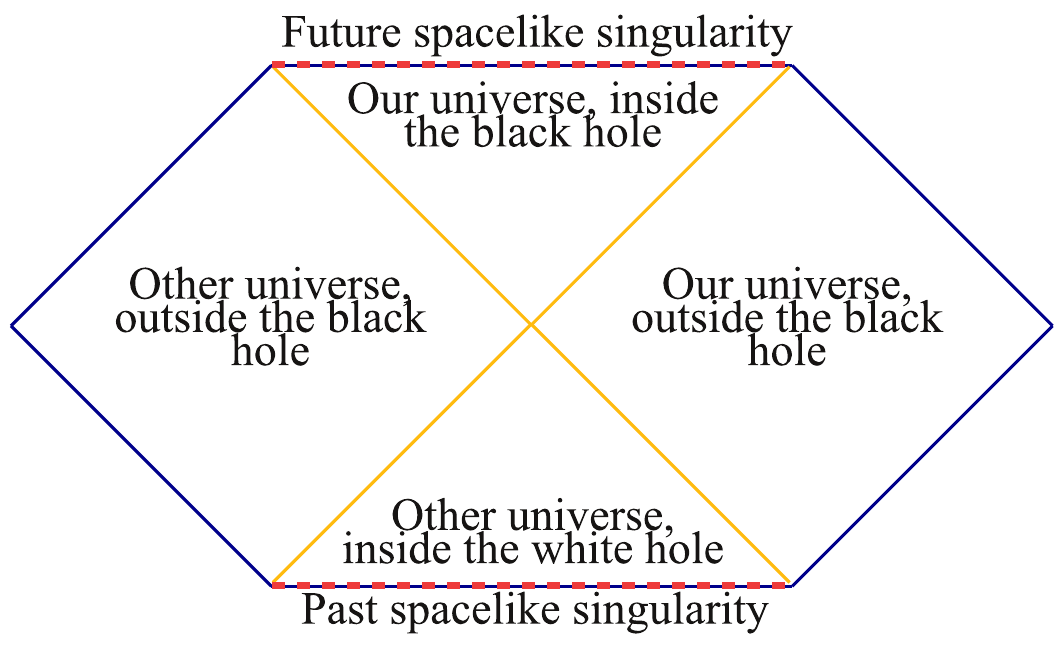}{0.5}{Penrose-Carter diagram for the standard maximally extended {\schw} solution.}

The maximal analytic extension is limited by the conditions $v''+w''\in(-\pi,\pi)$ and $v'',w''\in\(-\dsfrac{\pi}{2},\dsfrac{\pi}{2}\)$ (see Fig. \ref{std-schwarzschild}), because we have to stop at the singularity $r=0$, where we obtain infinite values, which prevent the analytic continuation.

But out coordinates allow the extension beyond these boundaries. To get the Penrose-Carter coordinates corresponding to our solution, we substitute \eqref{eq_coordinate_semireg} and get
\begin{equation}
\label{eq_schw_semireg_lightlike_coord}
\begin{array}{l}
\left\{
\begin{array}{lll}
v &=& \xi\tau^4 + \tau^2 + 2m\ln(2m - \tau^2) \\
w &=& \xi\tau^4 - \tau^2 - 2m\ln(2m - \tau^2). \\
\end{array}
\right.
\\
\end{array}
\end{equation}

From equation \eqref{eq_schw_semireg}, we see that our solution extends to negative values for $\tau$, and from \eqref{eq_schw_semireg_lightlike_coord}, the extension is symmetric with respect to the hypersurface $\tau=0$. We obtain the Penrose-Carter diagram from Fig. \ref{semireg-schwarzschild}.

\image{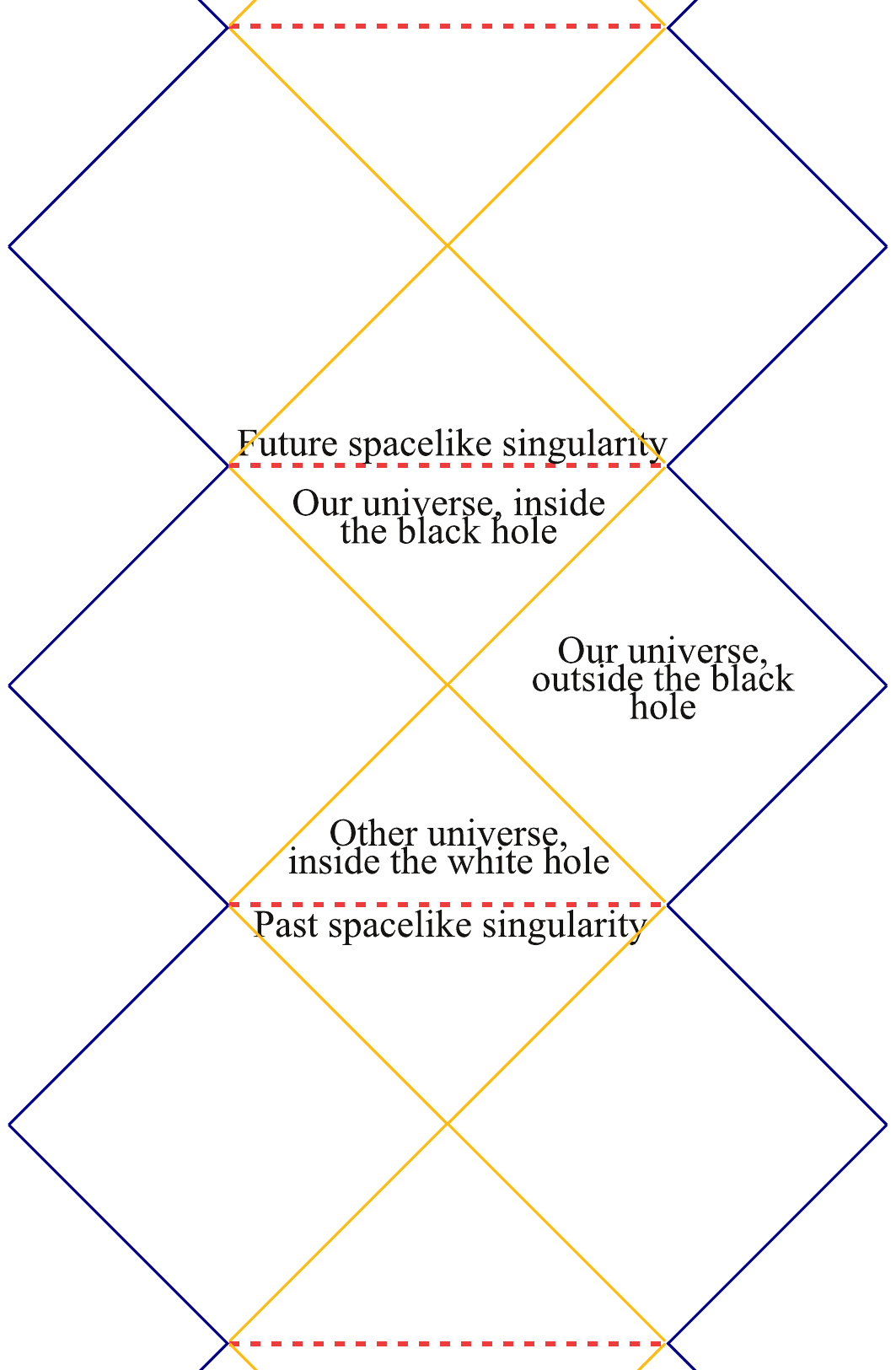}{0.5}{Penrose-Carter diagram for the {\semireg} maximally extended {\schw} solution.}

\subsection{Globally hyperbolic {\schw} spacetime}

The {\schw} singularity is spacelike. Although in the original {\schw} coordinates the singularity is apparently timelike and one dimensional, it is in fact spacelike, as we can see from the Kruskal-Szekeres coordinates or Penrose-Carter coordinates.

Let's take the Penrose-Carter diagram of a {\schw} black hole (Figure \ref{std-schwarzschild}).
This diagram actually represents the maximally extended {\schw} solution, in Penrose-Carter coordinates. This extended solution is interpreted to include, together with the universe containing the black hole, another universe, in which there is a white hole.

We can foliate this spacetime with Cauchy hypersurfaces, as we can see in Figure \ref{hexagon}.

\image{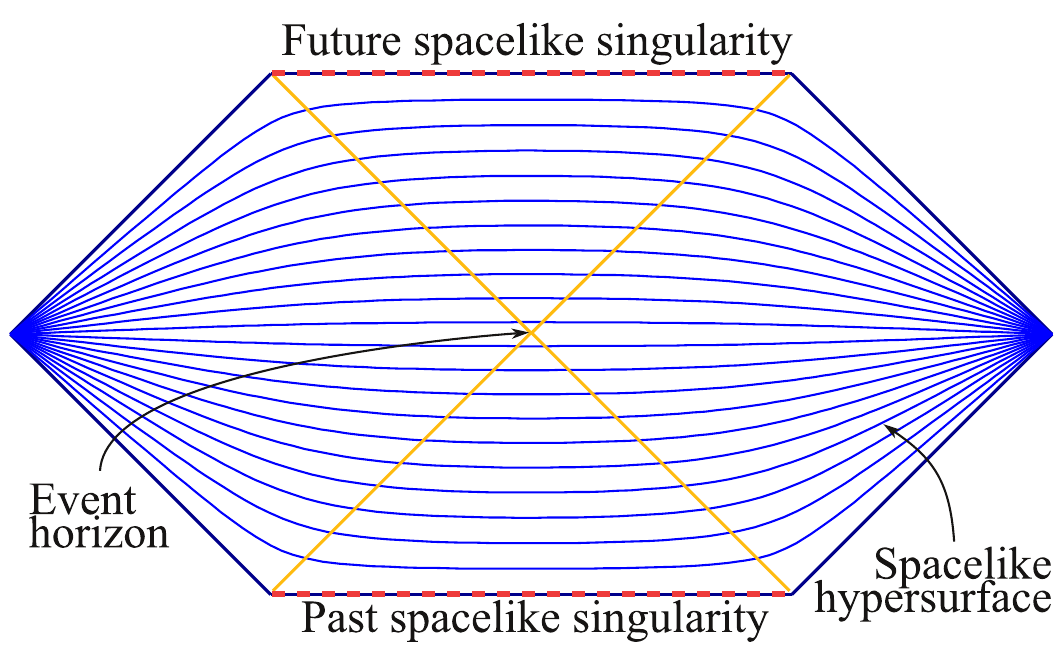}{0.5}{Space-like foliation of the maximally extended {\schw} solution.}

This foliation is obtained with the help of the version of the Schwarz-Christoffel mapping from equations \eqref{eq_strip} and \eqref{eq_s_c_map} (\cf \eg \cite{dri02}).

To obtain the foliation from Figure \ref{hexagon}, we take the prevertices to be
\begin{equation}
\label{eq_prevertices_hexagon}
	\left(-\infty,-a, a, +\infty, a+i,-a+i\right),
\end{equation}
where $a>0$ is a real number. The angles are respectively
\begin{equation}
\label{eq_angles_hexagon}
\left(\frac {\pi}{2},\frac {3\pi}{4},\frac {3\pi}{4},\frac {\pi}{2},\frac {3\pi}{4},\frac {3\pi}{4}\right).
\end{equation}

This foliation contains the past white hole, which one may consider unphysical. We can make a similar foliation, but without the white hole (fig. \ref{3oo3s}), if we use the prevertices
\begin{equation}
\label{eq_prevertices_3oo3s}
\left(-\infty,-a, 0, a, +\infty, b+i,-b+i\right),
\end{equation}
where $0<b<a$ are positive real numbers. The angles are respectively
\begin{equation}
\label{eq_angles_3oo3s}
\left(\frac {\pi}{2},\frac {\pi}{2},\frac {3\pi}{2},\frac {\pi}{2},\frac {\pi}{2},\frac {3\pi}{4},\frac {3\pi}{4}\right).
\end{equation}

\image{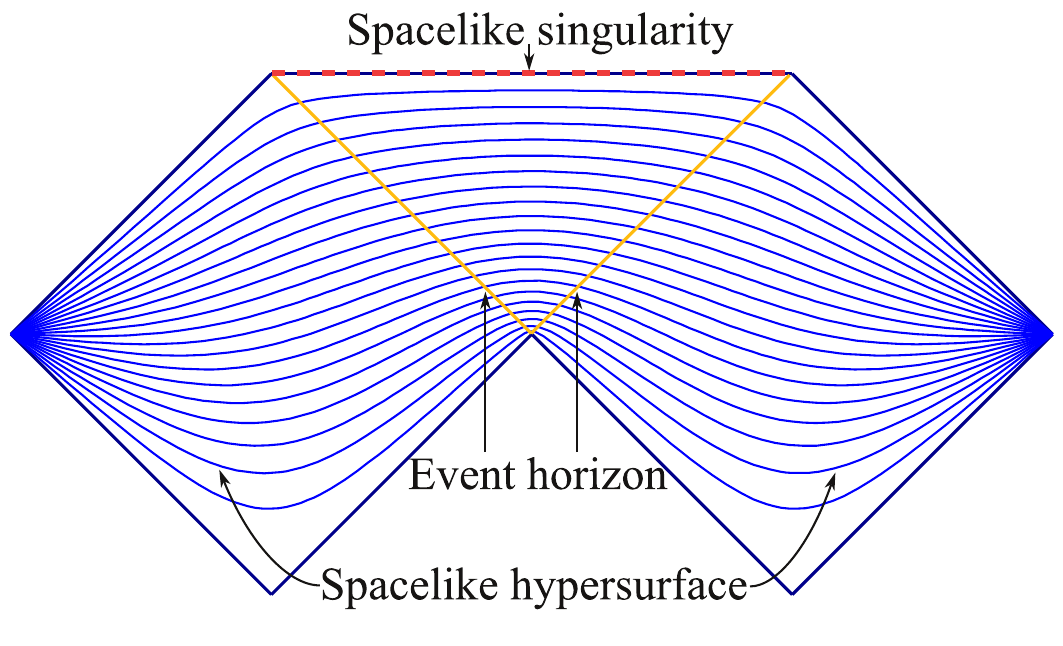}{0.5}{Space-like foliation of the {\schw} solution, without white hole.}

\section{Globally hyperbolic {\rn} spacetime}

\subsection{The Penrose-Carter diagrams for our solution}
\label{s_rn_ext_ext_penrose_carter}

The Penrose-Carter diagrams for the standard {\rn} spacetimes are presented in figure \ref{std-rn} (\citep{HE95}{157-161}).

\image{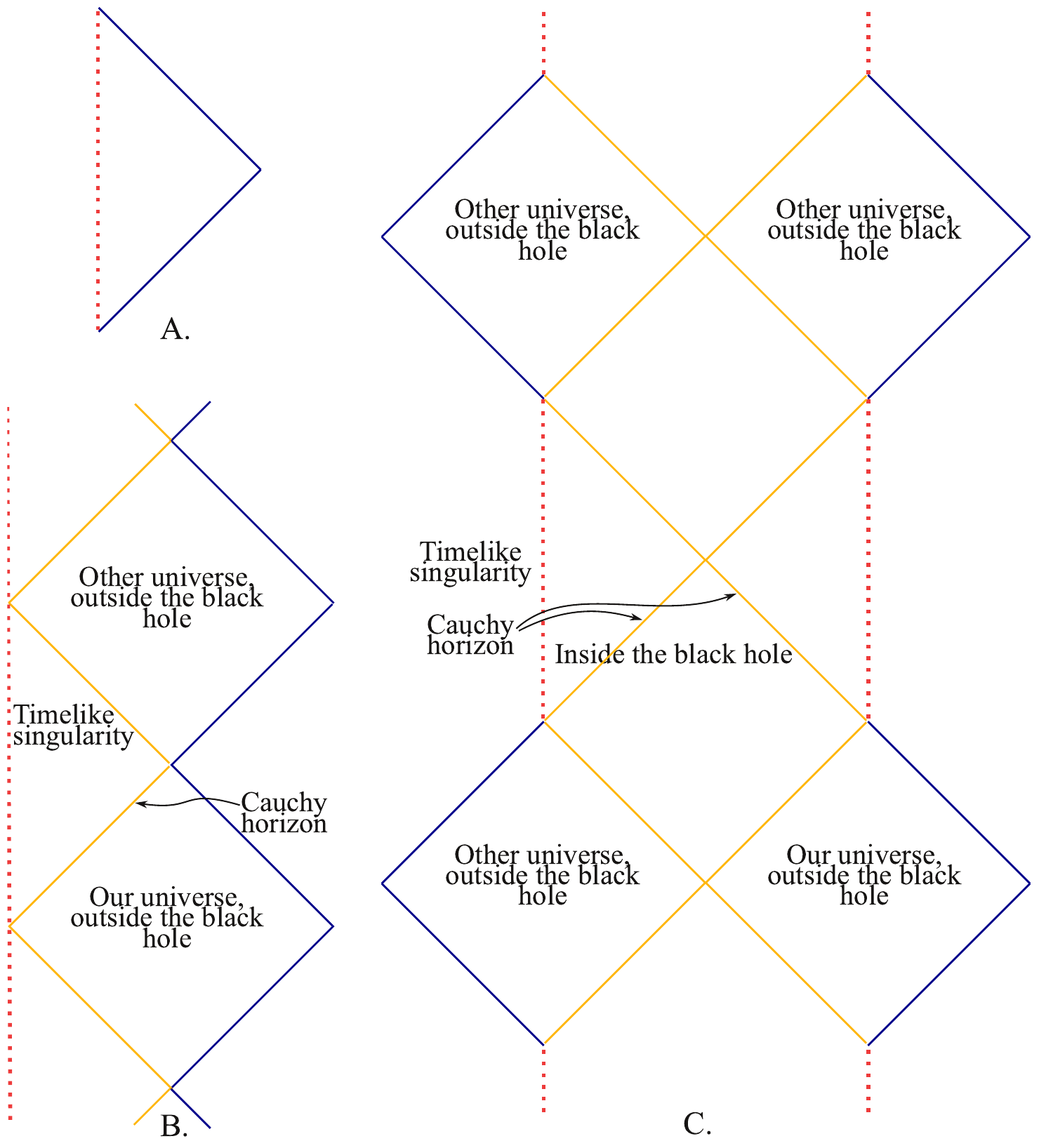}{0.6}{Reissner-Nordstr\"om black holes. A. Naked solutions ($q^2>m^2$). B. Extremal solution ($q^2=m^2$). C. Solutions with $q^2<m^2$.}

To obtain the Penrose-Carter coordinates for our solution from section \sref{s_black_hole_rn}, we add our coordinate transformation before the steps leading to the Penrose-Carter coordinates. For odd $\CS$, there is a region $\rho<0$, and the Penrose-Carter diagrams are the same as the standard diagrams ({\eg} \citep{HE95}{165}). If $\CS$ is even, the diagram repeats periodically, not only vertically, but also horizontally, symmetrical to the singularity.

The diagram for the naked {\rn} black hole ($q^2>m^2$) is obtained by taking the symmetric of the standard {\rn} diagram with respect to the singularity (figure \ref{ext-ext-rn-naked}).

\image{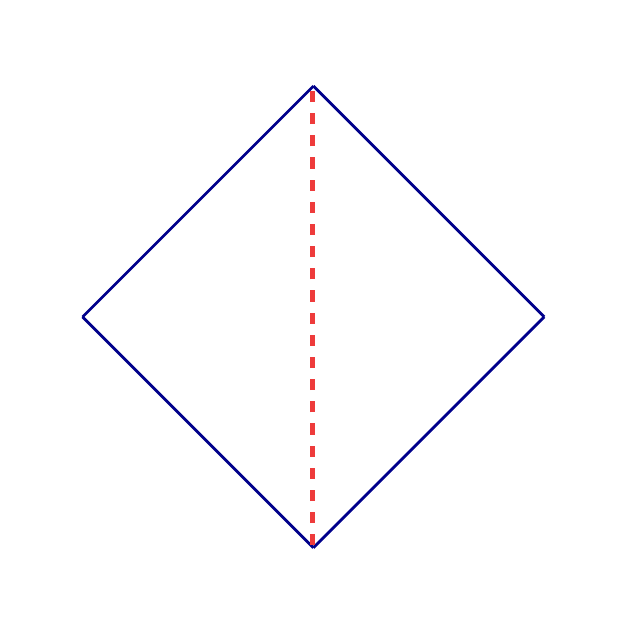}{0.3}{Penrose-Carter diagram for our naked {\rn} black hole ($q^2>m^2$), which is the analytic extension beyond the singularity. It is symmetric with respect to the timelike singularity.}

The resulting diagram for the extremal {\rn} black hole ($q^2=m^2$) is a strip symmetric about the singularity (figure \ref{ext-ext-rn-extremal}). 

\image{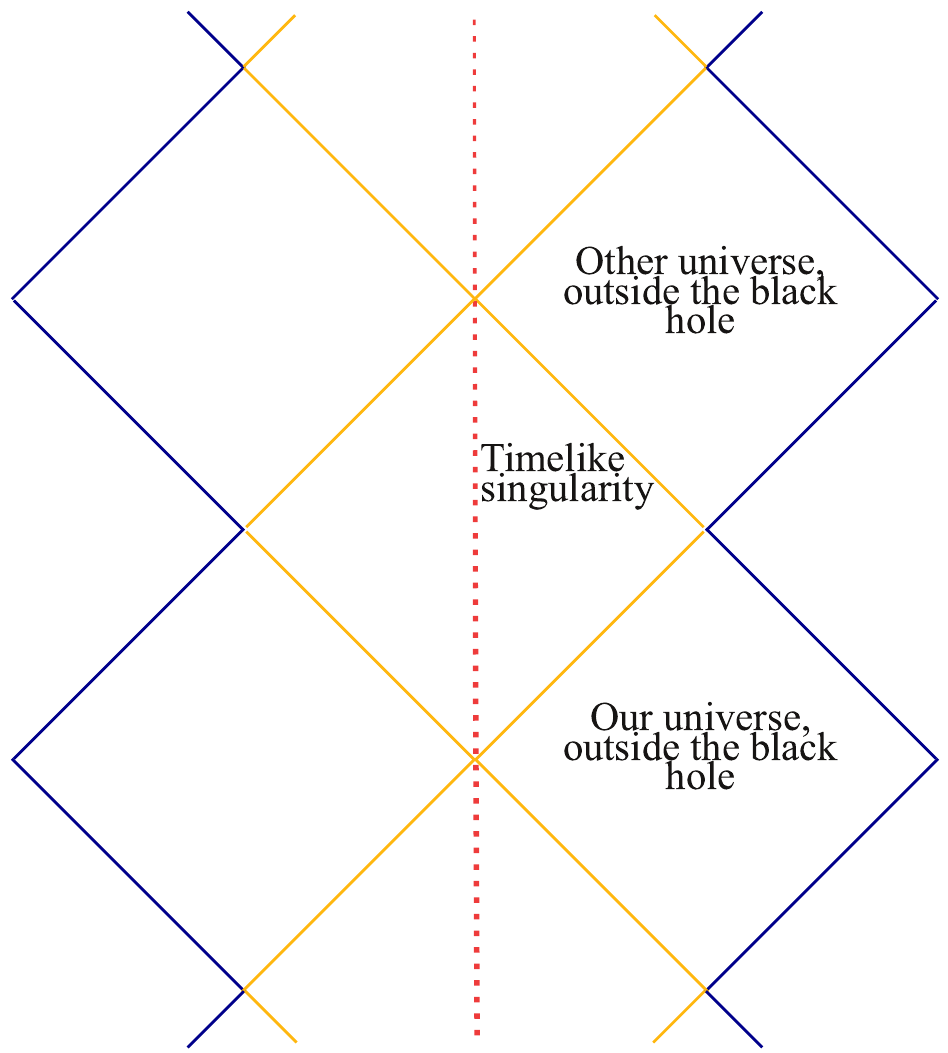}{0.5}{Penrose-Carter diagram for our analitic extension of the extremal {\rn} black hole ($q^2=m^2$) beyond the singularity. It is symmetric with respect to the timelike singularity, and repeats periodically in time.}

The Penrose-Carter diagram for the non-extremal {\rn} black hole ($q^2<m^2$) repeats, in the planar representation, both vertically and horizontally, and has overlapping parts (figure \ref{ext-ext-rn}).

\image{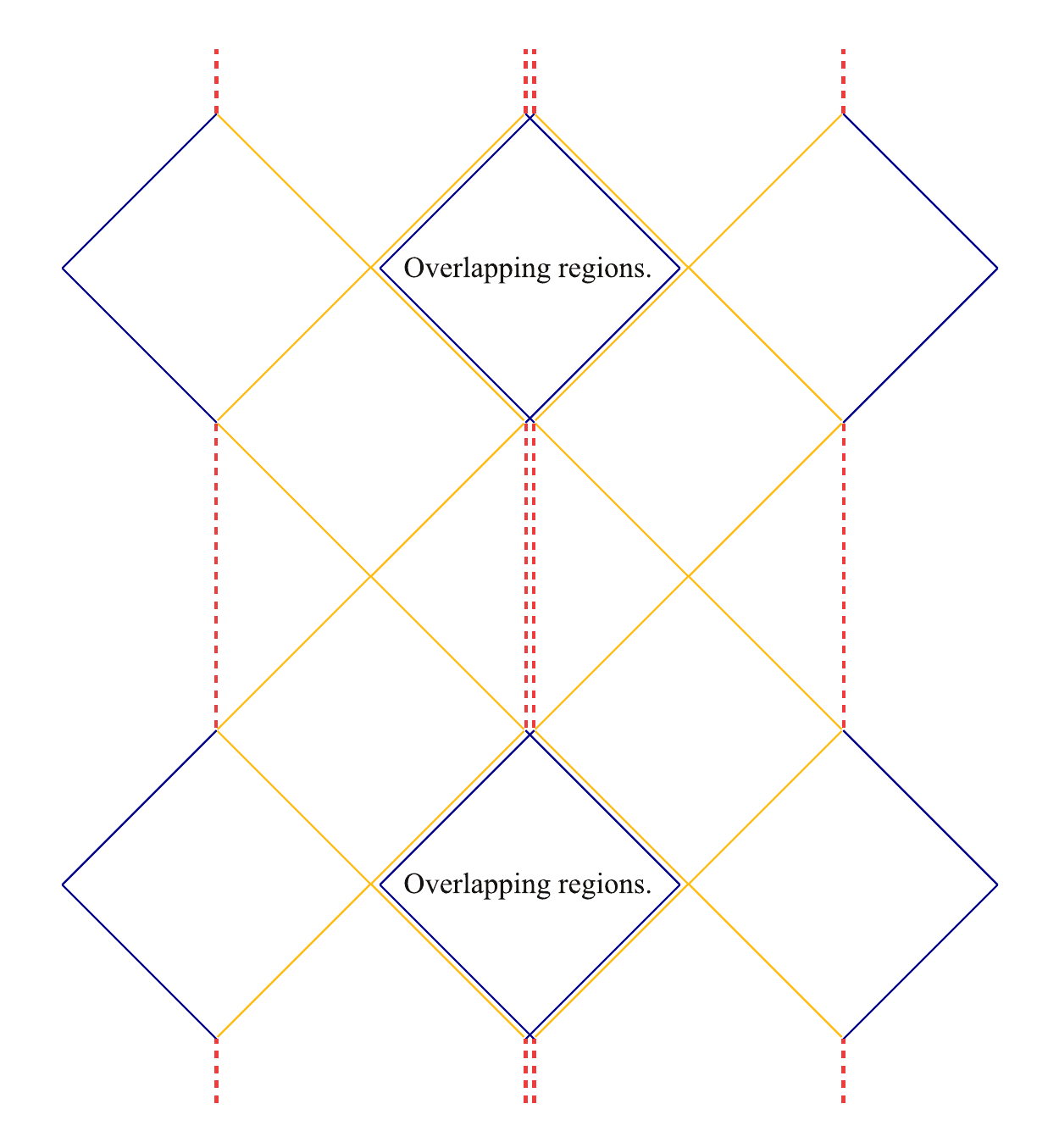}{0.6}{The Penrose-Carter diagram our extension of  for the non-extremal {\rn} black hole with $q^2<m^2$, analytically extended beyond the singularity. When represented in plane, it repeats periodically along both the vertical and the horizontal directions, and it has overlaps. In the diagram, there is an intentional small shift between the two copies, to make the overlapping visible.}

In the Penrose-Carter diagrams of the extensions of the {\rn} solution beyond the singularities, the null geodesics continue through the singularity, because they always make angle equal to $\pm \frac {\pi}{4}$.

\subsection{Globally hyperbolic {\rn} spacetime}
\label{s_rn_globally_hyperbolic}

Two reasons prevent the standard {\rn} solution to admit a Cauchy hypersurface. First, it cannot be extended beyond the singularity. Second, it has Cauchy horizons.

Our solution solves the first problem, because extends analytically the {\rn} spacetime beyond the singularity. We will use this to construct solutions that admit \textit{foliations} with Cauchy hypersurfaces.

Let's consider the coordinates $(\tau,\rho)$ as in \eqref{eq_coordinate_ext_ext}, with $\CT$ so that the condition \eqref{eq_rho_spacelike_condition_T} holds. They give a spacelike foliation, by the hypersurfaces $\tau=\tn{const}$. This foliation is global only for naked singularities; otherwise it is defined locally, in a neighborhood of $(\tau,\rho)=(0,0)$ given by $r<r_-$. Equation \eqref{eq_rn_metric} saids that the solution is \textit{stationary}, {\ie}, it is time independent, in the coordinates $(t,r)$. Hence, we can choose any value for the time origin. By this, we can cover a neighborhood of the entire axis $\rho=0$ with coordinate patches like $(\tau,\rho)$. To obtain global foliations, we use the Penrose-Carter diagrams of section \sref{s_rn_ext_ext_penrose_carter}, figures \ref{ext-ext-rn-naked}, \ref{ext-ext-rn-extremal} and \ref{ext-ext-rn}. Recall that in the Penrose-Carter diagram, the null directions are represented as straight lines inclined at an angle of $\pm \frac {\pi}{4}$.

It is easy to see that the naked {\rn} solution (figure \ref{ext-ext-rn-naked}) has a global foliation, because the Penrose-Carter diagram is the same as for the Minkowski spacetime (fig. \ref{diamond-rn-naked}). Hence, we use the same foliation as that for the Minkowski spacetime will.

To foliate it, we use the Schwarz-Christoffel mapping \eqref{eq_s_c_map}. By taking as prevertices
\begin{equation}
\label{eq_prevertices_diamond-s}
\left(-\infty,0, +\infty, i\right),
\end{equation}
and as angles
\begin{equation}
\label{eq_angles_diamond-s}
\left(\frac {\pi}{2},\frac {\pi}{2},\frac {\pi}{2},\frac {\pi}{2}\right),
\end{equation}
we obtain the foliation represented in fig. \ref{diamond-rn-naked}.

\image{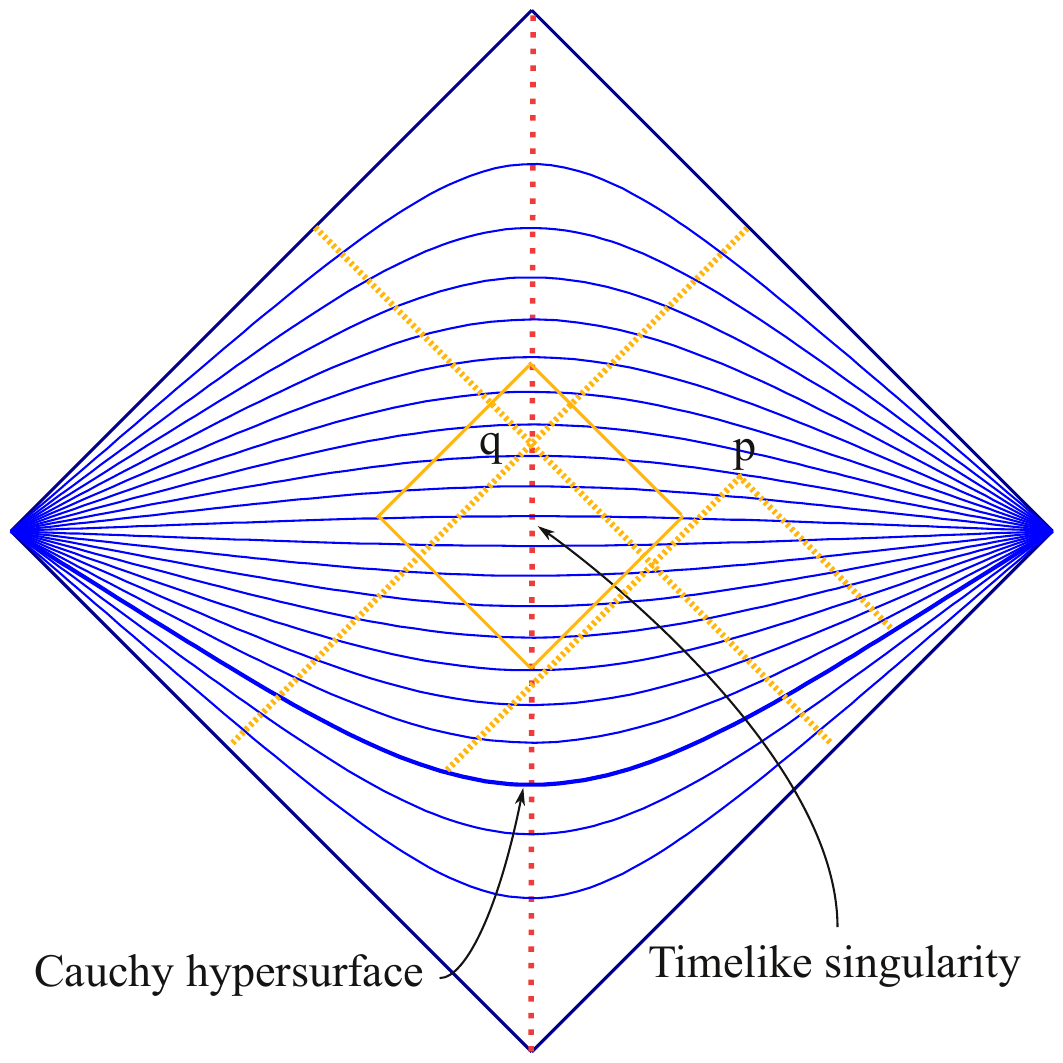}{0.5}{Spacelike foliation with Cauchy hypersurfaces, for the naked Reissner-Nordstr\"om solution extended beyond singularity.}

The solutions whose singularity is not naked have \textit{Cauchy horizons} (hypersurfaces which are boundaries for the Cauchy development of the data on a spacelike hypersurface), and don't admit globally hyperbolic maximal extension. To obtain a globally hyperbolic solution, we have to remove the regions beyond the Cauchy horizons, obtaining in a natural way a subset of the Penrose-Carter diagram, symmetric about the singularity $r=0$, which admits foliation (Figures \ref{up-big} and \ref{up-small-f}).

\image{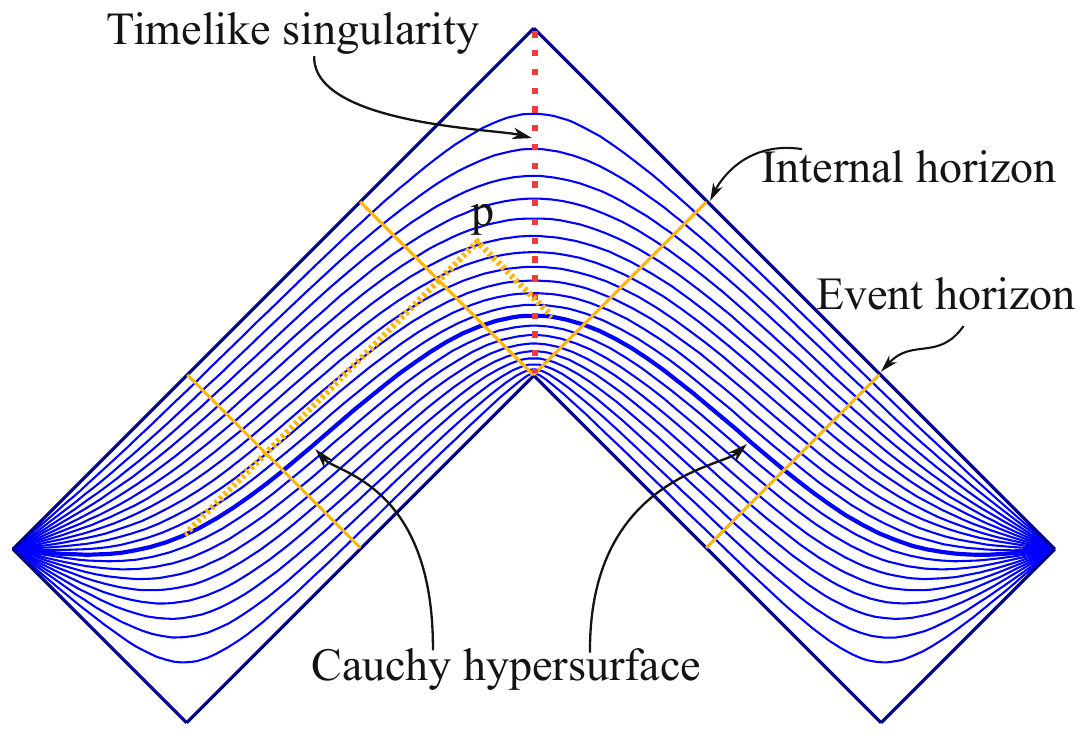}{0.6}{Foliation with Cauchy hypersurfaces, for the non-extremal Reissner-Nordstr\"om solution.}

\image{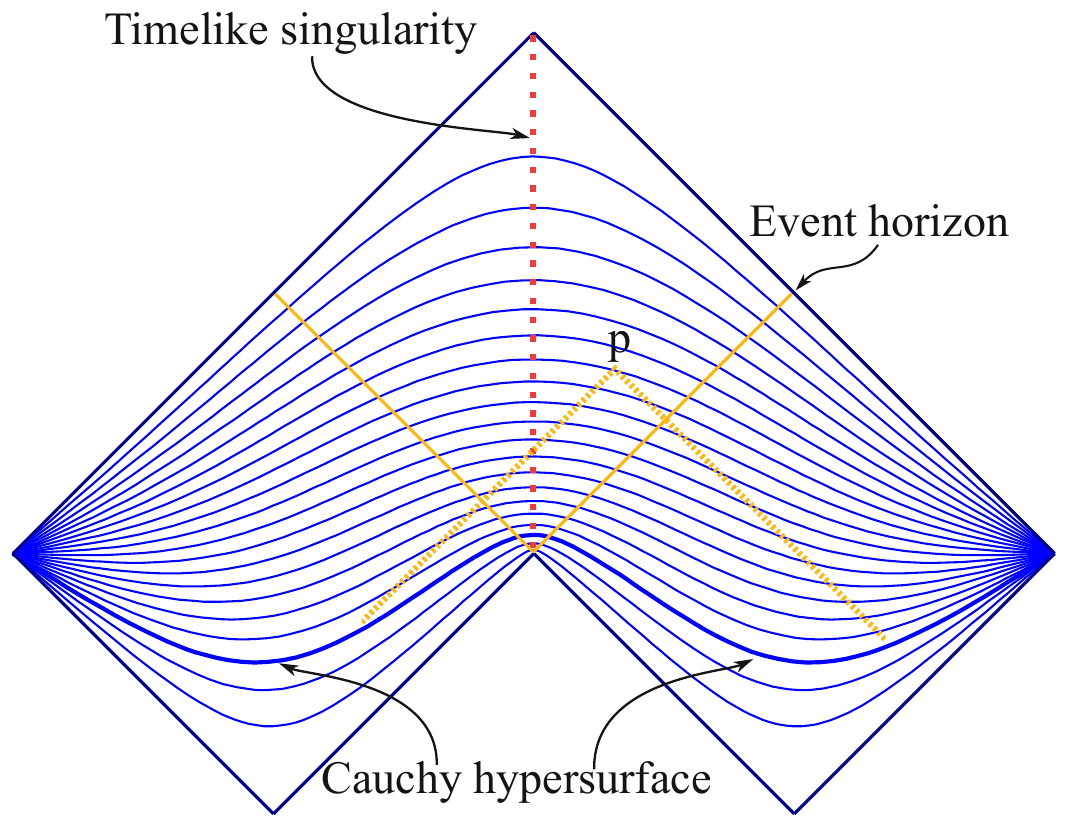}{0.4}{Foliation with Cauchy hypersurfaces, for the extremal Reissner-Nordstr\"om solution.}

We take as prevertices for the Schwarz-Christoffel mapping \eqref{eq_s_c_map} the set
\begin{equation}
\label{eq_prevertices_rn-kerr}
\left(-\infty,-a, 0, a, +\infty, i\right),
\end{equation}
where $0<a$ is real. The corresponding angles are
\begin{equation}
\label{eq_angles_rn-kerr}
\left(\frac {\pi}{2},\frac {\pi}{2},\frac {3\pi}{2},\frac {\pi}{2},\frac {\pi}{2},\frac {\pi}{2}\right).
\end{equation}
By choosing $a$ properly, we obtain the foliations represented in diagrams \ref{up-big} and \ref{up-small-f}, for the non-extremal, respectively the extremal solutions. Given that $\alpha_-=\alpha_+$ and the edges are inclined at most at $\frac{\pi}{4}$, and alternate in such a way that the level curves with $\tn{Im}(z)\in(0,1)$ always make an angle strictly between $-\frac{\pi}{4}$ and $\frac{\pi}{4}$, our foliations are spacelike everywhere.

Each of the figures \ref{diamond-rn-naked}, \ref{up-big}, and \ref{up-small-f}, contained a highlighted spacelike hypersurface. This hypersurface is intersected by all past (future) directed inextensible causal curves through a point $p$ from its future (past), therefore it is Cauchy.

\section{Globally hyperbolic {\kn} spacetime}
\label{s_kn_global}

\subsection{Removing the closed timelike curves}

For the {\kn} black holes, a new problem seems to forbid the existence of globally hyperbolic solutions: the existence of closed timelike curves. We will see that they can be removed, by a proper choice of ${\CS}$, ${\CT}$, and ${\CM}$.

The {\kn} ring singularity is at $r=0$ and $\cos\theta=0$. In Kerr-Schild coordinates, the solution can be analytically extended through the disk $r=0$, to another spacetime region which has $r<0$, hence is not isometric to the region $r>0$ (see Fig. \ref{kerr-schild}). On the other hand, in our coordinates with even ${\CS}$, ${\CT}$, and ${\CM}$, the analytic extension to $\rho<0$ gives a region which is isometric to that with $\rho>0$, and we can make the identification $(\rho,\tau,\mu,\theta)=(-\rho,\tau,\mu,\theta)$. In addition to removing the wormhole beyond $r=0$, also the closed timelike curves vanish, because there will be no region $r<0$. This allows us to find foliations similar to those for the {\rn} solution.

\image{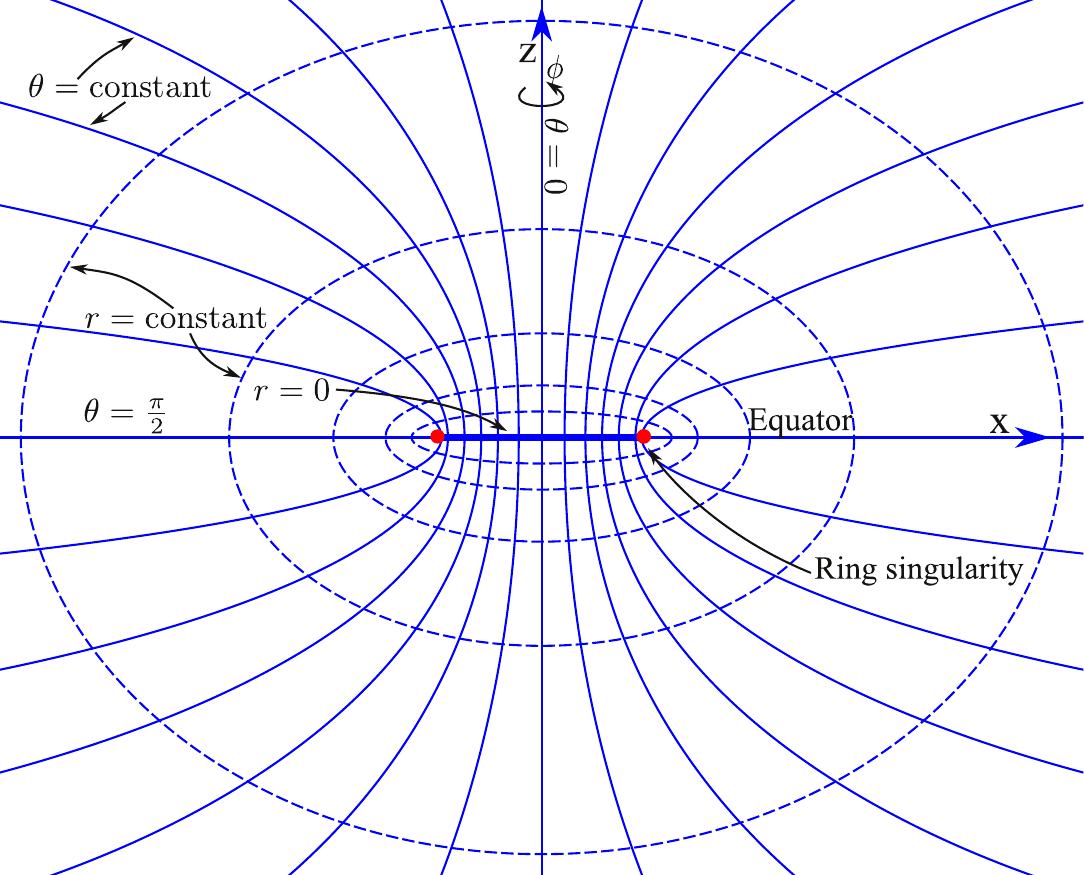}{0.95}{The standard Kerr and {\kn} solutions, in Kerr-Schild coordinates, admit analytic continuation beyond the disk $r=0$, into a spacetime which contains closed timelike curves. In our solution with even ${\CS}$, ${\CT}$, and ${\CM}$, the regions $\rho<0$ and $\rho>0$ can be identified, removing the wormhole and closed timelike curves.}

\subsection{The Penrose-Carter diagrams for our solution}
\label{s_kn_ext_ext_penrose_carter}

The Penrose-Carter diagrams the standard Kerr and {\kn} spacetimes have Penrose-Carter diagrams similar to those for the {\rn} spacetimes, although their symmetry is axial, not spherical, the singularity is ring-shaped, and closed timelike curves are present. Our approach can eliminate the closed timelike curves, and the Penrose-Carter diagrams for the {\rn} spacetimes are even more appropriate for the Kerr and {\kn} spacetimes. Therefore, we will use spacelike foliations similar to those for the {\rn} case, the difference being that the singularity is ring-shaped (see Figure \ref{kn-ext}). In particular, we can use the same Schwarz-Christoffel mappings, and the same globally hyperbolic subsets obtained by removing what's beyond the Cauchy horizons (fig. \ref{kn-ext}).

\image{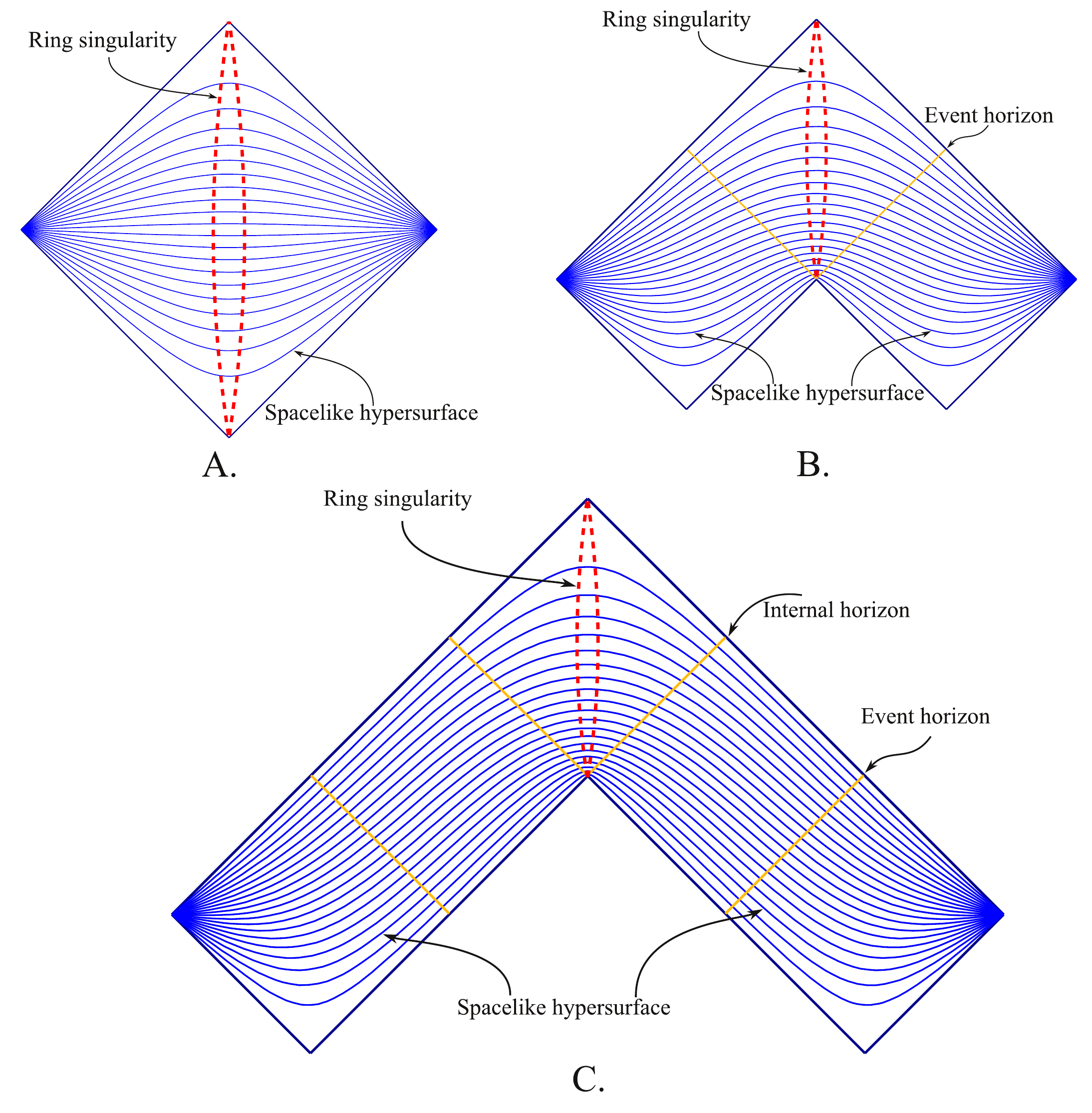}{0.95}{Space-like foliation of the extended {\kn} solutions. \textbf{A.} Naked singularity. \textbf{B.} Extremal solution. \textbf{C.} Non-extremal solution.}

\section{Non-stationary black holes}
{\label{s_non_stationary_bh}

The stationary solutions are idealizations. In reality, a black hole doesn't necessarily exist from the beginning of the universe, it may be younger. Also, it may evaporate completely after a finite time interval.

Non-stationary black hole solutions can be obtained by considering the parameters $a,q,m$ to be variable. This allows us to patch solutions together, and construct more general solution, which change in time.

For example, a spacelike foliation of a non-rotating and electrically neutral black hole which formed after the beginning of the Universe, and which continues to exist forever, is represented in Figure \ref{superman}.

\image{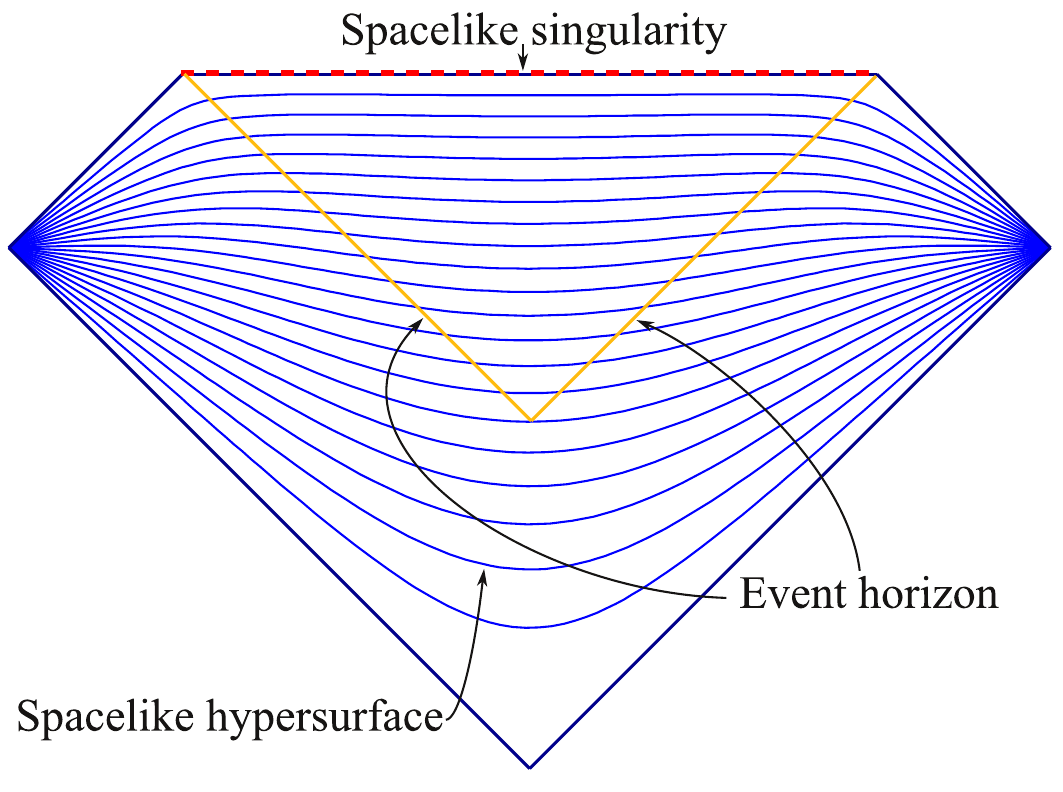}{0.75}{The spacelike foliation of a non-rotating and electrically neutral black hole formed after the beginning of the Universe, which continues to exist forever.}

The prevertices of the Schwarz-Christoffel mapping \eqref{eq_s_c_map} whose image is represented in Figure \ref{superman} are given by the set
\begin{equation}
\label{eq_prevertices_superman}
\left(-\infty,0, +\infty, a+i, -a+i\right),
\end{equation}
where $0<a$ is a positive real number. The angles are respectively
\begin{equation}
\label{eq_angles_superman}
\left(\frac {\pi}{2},\frac {\pi}{2},\frac {\pi}{2},\frac {3\pi}{4},\frac {3\pi}{4}\right).
\end{equation}

We see that, because the typical spacelike hypersurface in the foliation in Figure \ref{3oo3s} is diffeomorphic with the space $\R^3$ of the Minkowski spacetime $\R^3\times\R$, the topology doesn't change because of the occurrence of a neutral non-rotating black hole. In fact, the typical spacelike hypersurface of the foliation of a charged and rotating black hole also has the same topology as $\R^3$, as we have seen. So they can appear and evaporate as well in an $\R^3$ space, without disrupting the topology. This condition is required to have a good time evolution.

If the non-rotating and electrically neutral black hole is primordial (exists from the beginning of the universe), but evaporates completely after a finite time, the spacelike foliation is as represented in Figure \ref{up-small-s}. The prevertices and the angles are identical to those for Figure \ref{up-small-f}.

\image{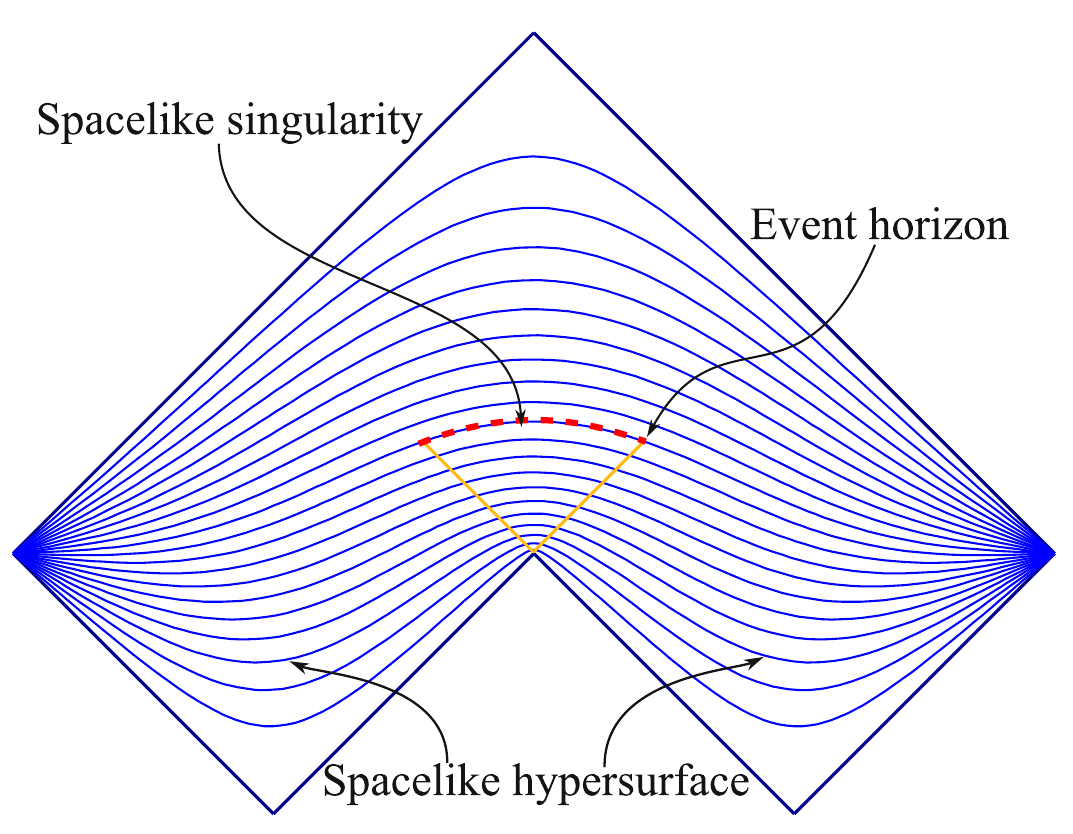}{0.65}{The spacelike foliation for a non-rotating and electrically neutral primordial black hole, which evaporates after a finite time.}

If this non-rotating and electrically neutral black hole is not primordial and it evaporates completely after a finite time, then at large distances the spacetime is very close to the Minkowski spacetime, being assymptotically flat. Consequently, a spacelike foliation looks like that in Figure \ref{diamond-s}.

\image{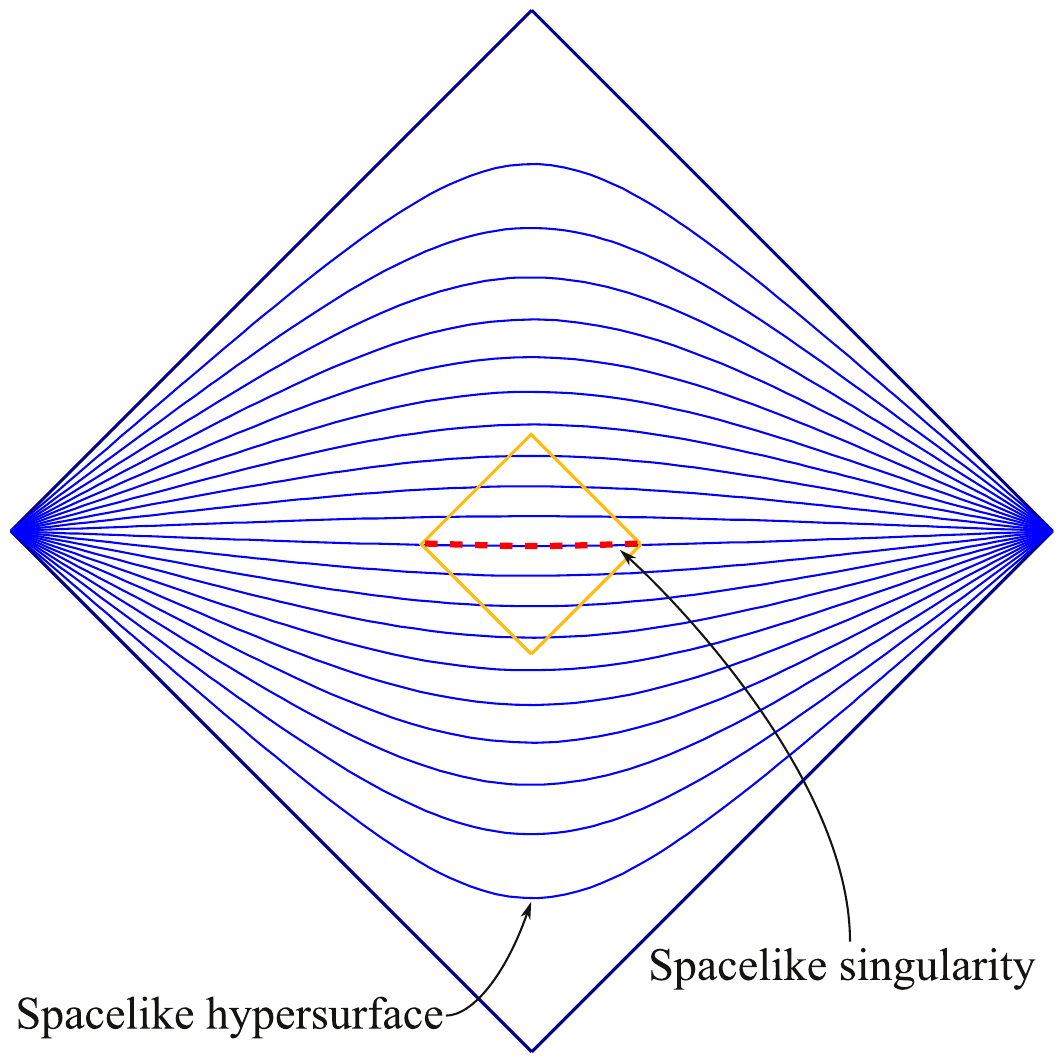}{0.45}{The spacelike foliation for a non-rotating and electrically neutral black hole formed after the beginning of the Universe, and which evaporates after a finite time.}

The prevertices of the diagram represented in Figure \ref{diamond-s} are the same as equations \eqref{eq_prevertices_diamond-s} and \eqref{eq_angles_diamond-s}.

In fig. \ref{evaporating-bh-s}, we compare the standard interpretation of a spacelike singularity, which is assumed to destroy information, and the interpretation resulting from the possibility to extend the solution beyond the singularity, given by our coordinates which make the {\schw} solution analytic at the singularity. We use the case of reference in discussions about the black hole information paradox, which is the evaporating Oppenheimer-Snyder black hole.
\image{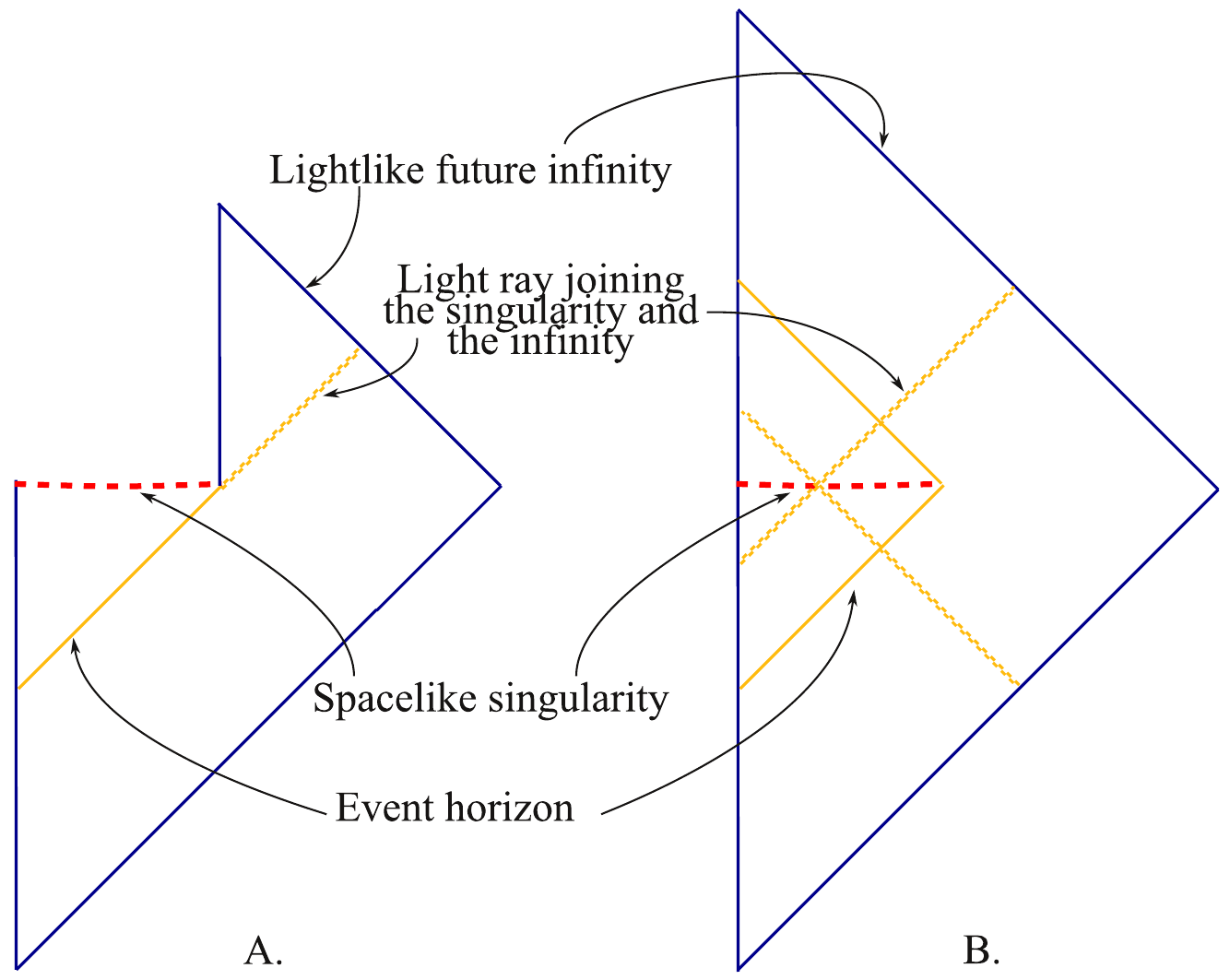}{0.7}{The Penrose-Carter diagram for an evaporating Oppenheimer-Snyder black hole.
\textbf{A.} The standard view is that the singularity destroys the information.
\textbf{B.} Our approach suggests that the singularity doesn't destroy the information.}

Let's see now what happens if the singularity is timelike, as in the case of the charged and/or rotating black holes. If the black hole is primordial and evaporates, the corresponding spacelike foliation is represented in Figure \ref{up-small}.

\image{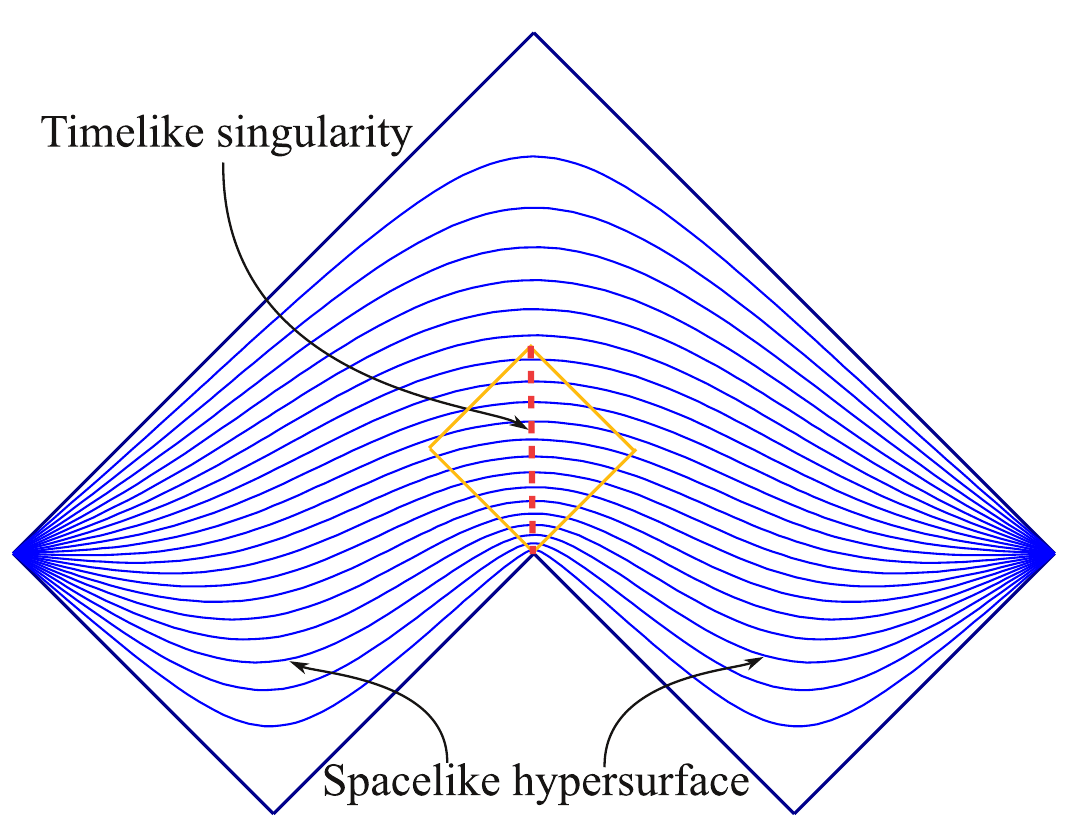}{0.65}{Primordial evaporating black hole with timelike singularity.}

The prevertices and the angles are the same as those for the Figure \ref{up-small-f}.

The spacelike foliation of a spacetime containing a black hole which is not primordial and does not evaporate is represented in Figure \ref{diamond-t-f}. The prevertices and the angles are again those from equations \eqref{eq_prevertices_diamond-s} and \eqref{eq_angles_diamond-s}.

\image{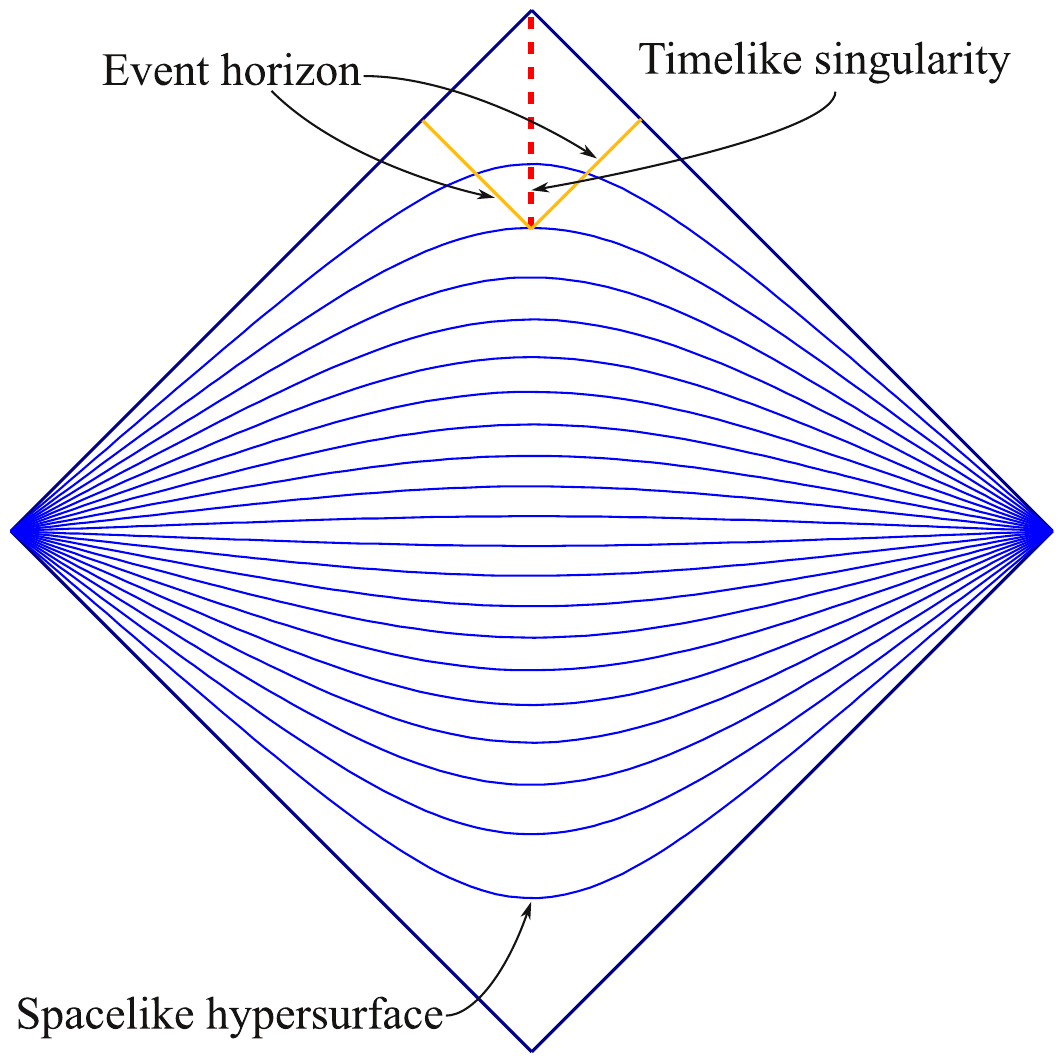}{0.45}{Non-primordial non-evaporating black hole with timelike singularity.}

The same prevertices and angles are used to construct the spacelike foliation for a non-primordial evaporating black hole, represented in Figure \ref{diamond-t}. 

\image{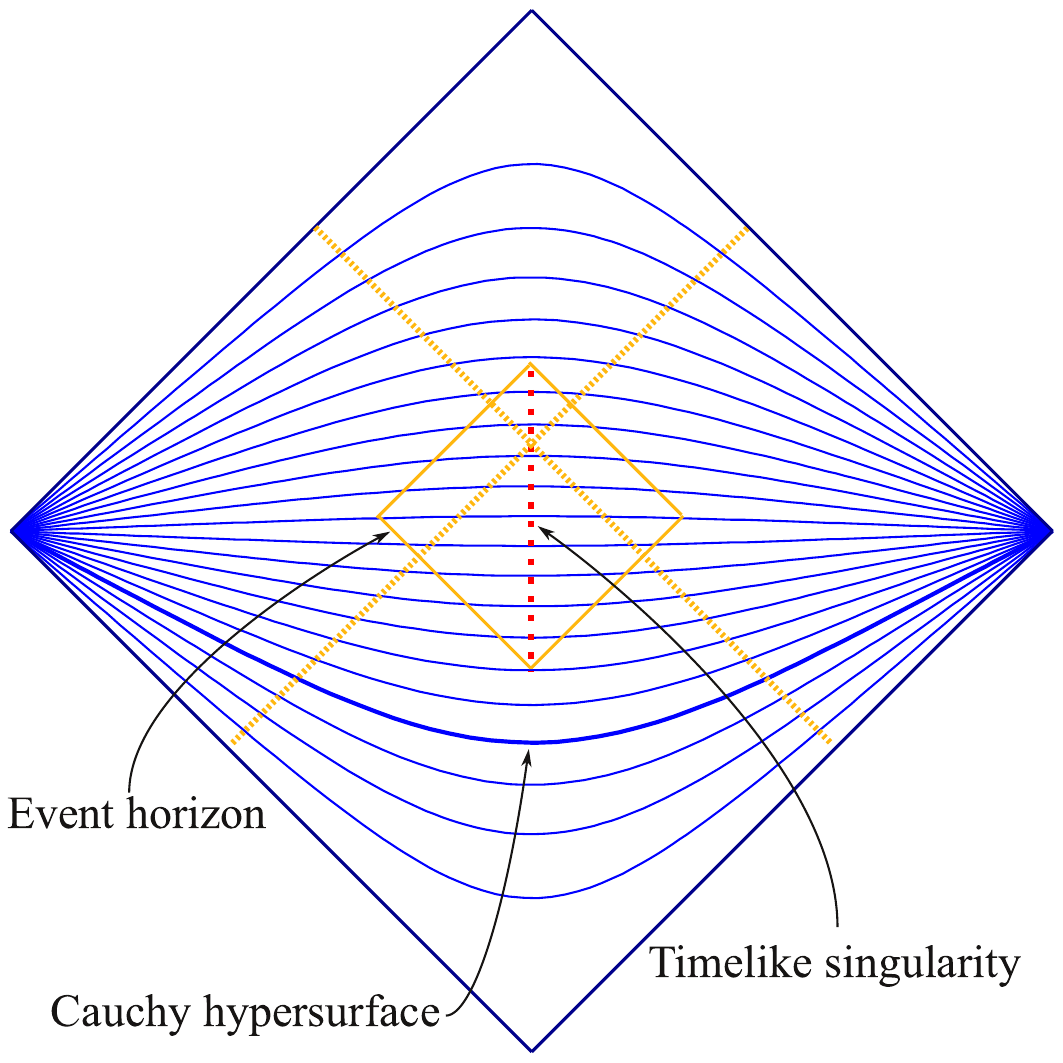}{0.5}{Non-primordial evaporating black hole with timelike singularity. The fact that the points of the singularity become visible to distant observers is not a problem for the global hyperbolicity, because the null geodesics can be extended beyond the singularity.}

If the black hole rotates, the singularity is ring-shaped (as in Fig. \ref{kn-ext}), but the diagrams are similar.

\section{Conclusions}

This chapter shows how one can extend smoothly the black hole solutions at the singularities, then find an appropriate foliation of the spacetime into spacelike hypersurfaces, and an appropriate extension of the spacetime at singularities, so that the topology of the spacelike hypersurfaces of the foliations is preserved. In our typical examples, the time evolution in the presence of singularities is restored. The examples given here showed this explicitly for the neutral and charged, rotating or non-rotating, primordial or not, evaporating or not black holes. These models have been shown here to have Penrose-Carter diagrams conformally equivalent to that of the Minkowski spacetime, inheriting therefore from the latter the global hyperbolicity. Consequently, the Cauchy data is preserved, and the information loss is avoided. This allows the construction of Quantum Field Theories in such curved spacetimes (\cite{HP96}, p. 9), and the unitarity is restored.

\chap{ch_quantum_gravity}{Quantum gravity from metric dimensional reduction at singularities}

This chapter is based on author's original research, from article \cite{Sto12d}, and author's talk given in May 2012 at JINR, Dubna \cite{Sto12f}.

Old and recent theoretical investigations suggest that the quantization of gravity would be possible, if a dimensional reduction takes place at high energy scales. This would also solve other problems of Quantum Field Theory. But what is missing is a mechanism which would cause the dimensional reduction.

The mathematical approach to singularities in General Relativity presented in this Thesis leads to the conclusion that at singularities, the metric tensor becomes degenerate. Geometry, and some of the fields, become independent of the directions in which the metric is degenerate. Effectively, this is as if one or more spacetime dimensions vanish at singularities. Consequently, it appears that singularities leads in a natural way to the spontaneous dimensional reduction needed by Quantum Gravity.

\section{Introduction}

The most successful theories in fundamental theoretical physics are \textit{General Relativity} (GR) and \textit{Quantum Field Theory} (QFT). They offer accurate and complementary descriptions of the physical reality, and their predictions were confirmed with very high precision.

But they are not without problems, especially with infinities. In GR, infinities are present in the form of singularities. In QFT, infinities appear in the \textit{perturbative expansion}. In this Thesis we tried to offer a solution to the infinities in GR. The infinities in QFT are considered much less problematic now, due to renormalization techniques, which are shown to apply to the entire \textit{Standard Model} of particle physics \cite{HV72dimreg,tHooft73dimreg,tHooft98glorious}, and the \textit{renormalization group} \cite{SP53RG,GML54RG,BS56RG,BS80intro,Shirkov96,Shirkov99}.

But it seems that, when one tries to combine GR and QFT, infinities reappear, and renormalization can't remove them. GR without matter fields is perturbatively {\nonrenormalizable} at two loops \cite{HV74qg,GS86uvgr}. It required number of higher derivative counterterms with their coupling constants is infinite. The main cause of the problem is the dimension of the Newton constant, which is $[\mathcal G_N]=2-D=-2$ in mass units.

Many theoretical\footnote{Experimental data is not available for such small distances.} investigations were made in understanding \textit{ultraviolet} (UV) limit in QFT, the \textit{small scale}. Particularly in various approaches to Quantum Gravity (QG), the evidence accumulated so far suggests, or even requires, that in the UV limit there is a \textit{dimensional reduction} to two dimensions. What is under debate is the meaning, the nature, the explicit cause of this spontaneous dimensional reduction.

In this chapter we suggest that our approach to singularities predicts in a very concrete way the dimensional reduction.

The usual position regarding the apparent incompatibility between QFT and GR, manifested as the {\nonrenormalizability} of GR, is that General Relativity is to be blamed, and has to be modified. It is hoped that one of the many known approaches to QG will cure not only the {\nonrenormalizability}, but also the problem of singularities, if we are willing to give up up one or more fundamental principles of GR.

We will take here an opposite position, namely that solving the problem of singularities in GR, has as effect the dimensional reduction in the UV limit, which would make GR renormalizable.

According to 't Hooft \cite{tH09fundST},
\begin{quote}
Quantum Gravity is usually thought of as a theory, under construction, where the postulates of quantum mechanics are to be reconciled with those of general relativity, without allowing for any compromise in either of the two. \end{quote}
In this chapter, we will try to go as far as we can in reconciliating QFT and GR, without making compromises.

In section \sref{s_hints_of_dimensional_reduction}, we will review some of the hints suggesting that  if the spacetime becomes $2$-dimensional at small scales, the quantization of gravity will become possible. In section \sref{s_dimensional_reduction_singularities}, we will explain how singularities studied in this Thesis cause dimensional reduction. Then we will show some connections between dimensional reduction caused by singularities, and other types of dimensional reduction, proposed so far in some of the approaches to QG.

\section{Suggestions of dimensional reduction coming from other approaches}
\label{s_hints_of_dimensional_reduction}

The original method of regularization through dimensional reduction appeared as a method of regularization in  QFT. It is rather formal, apparently with no implications to the actual physical dimensions. It was suggested by the observations that the loop integrations depend continuously on the number of dimensions. The poles can be avoided by replacing the dimension $4$ by $4-\varepsilon$. At the end we just make $\varepsilon\to 0$ \cite{BG72dimreg,HV72dimreg,tHooft73dimreg}.

On the other hand, Quantum Gravity works fine in two dimensions. This justifies the consideration of the possibility that, at small scales, the number of dimensions is somehow reduced, while maintaining four dimensions at large scales. We review in the following some of the indications suggesting the dimensional reduction.

\subsection{Dimensional reduction in Quantum Field Theory}

\subsubsection{Several hints of dimensional reduction from QFT}

Several different situations in QFT suggest that the two-dimensionality plays an important role. Since the discovery of the first exactly solvable QFT model \cite{Thi58}, the two-dimensional QFT proved to be very productive, leading in a direct and non-perturbative way to interesting results, which suggested hypotheses for four-dimensions (see \cite{AAR91} and \cite{frishman2010non} and references therein).


In the Standard Model, one problem is that, to avoid the destabilization of the electroweak symmetry breaking scale, the mass of the Higgs boson needs to be fine-tuned, to an accuracy of $10^{-32}$ \cite{ADFLS10}. By fixing the cutoff at the electroweak scale, the SM works well, and it is believed that this indicates new physics beyond this scale. In \cite{ADFLS10} is explored the possibility of keeping unmodified the structure of the SM, and reduce the dimension. One main idea is that the dimensional reduction to $d=2$ space dimensions makes the Higgs terms linearly divergent, and to $d=1$ makes them logarithmically divergent. This would make unnecessary the fine-tuning of the Higgs mass.

\subsubsection{Fractal universe and measure dimensional reduction}
\label{s_dim_red_fractal}

G. Calcagni developed the \textit{fractal universe program}, with the purpose to keep from the {\HL} gravity the nice feature that, in the UV regime, it leads to a two-dimensional phase, but without giving up Lorentz invariance. He kept the isotropic scaling, and compensates it, by replacing the standard measure used in the action with another measure, (initially a Lebesgue-Stieltjes measure) which reduces the Hausdorff dimension at the UV fixed point to two. The action becomes fractional, and the resulting theory is non-unitary, being dissipative. But the energy is still conserved.

The used action was
\begin{equation}
\label{eq_stmes}
S=\int_\mc M\de\varrho(x)\, \mc L(\phi,\partial_\mu\phi),
\end{equation}
with the measure
\begin{equation}
\label{eq_stme}
\de \varrho(x) = v\,\de x^D
\end{equation}
He used a weight of the form $v:=\prod_{\mu=0}^{D-1} f_{(\mu)}(x)$.

The original argument in \cite{Calc2012FractalGeometryFT} was that this measure leads to perturbative renormalizability.
By similar ideas, Calcagni obtained a modified action for General Relativity, and from this, a modified version of the Einstein equation.

He refined the theory introduced in \cite{Calc2010FractalQFT,Calc2010FractalUniverse,Calc2011FractalGravity}, and put on more rigorous mathematical basis, in the subsequent papers \cite{Calc2011FractalDiscreteContinuum,Calc2011FractalGeometry,Calc2012FractalGeometryFT,ACOS2011FractionalNC,Calc2012MomentumLaplace,Calc2012DiffusionQG,Calc2012DiffusionMultiFractional}.
The mathematically foundation of the fractal universe theory is a Lebesgue-Stieltjes measure or a fractional measure \cite{Calc2011FractalGeometry}, fractional calculus, and fractional action principles \cite{el2005fractional,el2008fractional,udriste2008euler}. 

Later, in \cite{calcagni2013quantum}, it was shown that modified measure is not enough to ensure it, and in fact the original power-counting argument from \cite{Calc2012FractalGeometryFT} fails. Yet, in \cite{calcagni2013quantum} is argued that ``[t]he interest in fractional theories is not jeopardized'', and in \cite{calcagni2013multi} that multi-scale theories are still of interest, not only for QG, but also for cosmology (accelerated expansion, cosmological constant, alternative to inflation, early universe, big bounce, \etc).

Maybe the modified measure approach is insufficient to obtain perturbative renormalizability, but the dimensional flow at singularities is richer and provides some additional features.

We will show in section \sref{s_dimensional_reduction_singularities_measure} that our theory of singularities leads precisely to a measure of the form \eqref{eq_stme}.

\subsubsection{Topological dimensional reduction}
\label{s_dim_red_topological}

An interesting approach, proposed by D. V. Shirkov, was initially aimed to replace the Higgs boson, which was expected to be found at $140\pm 25$ GeV, with the Ginzburg-Landau-Higgs solution of a constant classical Higgs field at $\sim~250$ GeV.  But in the absence of a Higgs quantum field, the weak force would be {\nonrenormalizable}. Shirkov realized that the regularization can be obtained by working in a spacetime with variable topology, so that the number of dimensions varies from $D=4$ in the IR limit, to $D=1+d<4$ in the UV limit, and the coupling constant is assumed to run between these limits.

Shirkov proposed a $g\varphi^4$ QFT model described by a self-interacting Lagrangian $L=T-V$ \cite{shirkov2010coupling}, where
\begin{equation}
	V(m,g;\varphi) = \dsfrac{m^2}{2}\varphi^2 + \dsfrac{4\pi^{d/2}M^{4+d}}{9}g_d\varphi^4,\,g>0
\end{equation}

The toy model he used was a manifold obtained by joining two cylinders $S_{R,L}$ and $S_{r,l}$, of lengths $L,l$ and radii $R>r$, with a transition region $S_{coll}$ of variable radius. One interesting problem taking place on this manifold, was to solve some QFT equations, and then take the limit $r\to 0$.

In momentum space, the volume element $\de^4 k$ was replaced by
\begin{equation}
	\de_M k = \dsfrac{\de^4 k}{1 + \frac{k^2}{M^2}},
\end{equation}
yielding a one-loop Feynman integral, whose IR asymptote is the function
\begin{equation}
\ln\dsfrac{q^2}{m_i^2},
\end{equation}
and the UV asymptote, the function
\begin{equation}
\ln\dsfrac{4M^2}{m_i^2} + \dsfrac{M^2}{q^2}\ln\dsfrac{q^2}{M^2}.
\end{equation}
This reversed the running of the coupling constant in the small distance region, to a minimum value $\bar g_2(\infty)<\bar g_2(M_{dr}^2)$, where $M_{dr}$ is the reduction scale.

Another feature of this approach is the possibility that the coupling constants of the Standard Model converge in the UV limit toward the same minimum value, without requiring by this a GUT scheme like $SU(5)$ \cite{georgi1982lie,BH2010algebra}.

The exploration of the idea of topological dimensional reduction was continued by P. Fiziev and D. V. Shirkov in \cite{FS2011KG,Fiz2010Riem,FS2012Axial,shirkov2012dreamland}, and applied to Klein-Gordon equation.

In section \sref{s_dimensional_reduction_singularities_metric_vs_topologic}, we establish strong connections between our approach, and the topological dimensional reduction approach of Shirkov and Fiziev.

\subsection{Dimensional reduction in Quantum Gravity}

A review of the indications that various approaches to QG suggest a dimensional reduction at small scales, is done by Carlip \cite{Car09SDR,Car10sssst}. 

The \textit{causal dynamical triangulations} approach approximates spacetime by flat four-simplicial manifolds, as in quantum Regge calculus \cite{Reg61}. Additional constraints on the lengths of the edges of the simplices are imposed \cite{AJL00,AJL04,AJL05r,AJL05s,AJL09aqg}. To enforce causality, it is required that the time-slicing at discrete times is fixed, and that the time-like edges are in the same direction. The \textit{spectral dimension} -- the effective dimension of the diffusion process, related to the dispersion relation of the corresponding differential operator, turns out to be four at large distances, but two at short distances \cite{AJL05s}. Spectral dimension is calculated in other approaches to QG, and agrees \cite{LR05fractal,Hor09spectral,Mod2008fractalLQG,Car10sssst,Ben2009fractalQST,MN2010sdQU,NS2011sdUG}. There is a connection between the spectral dimension and the spacetime dimension, but they are not equivalent \cite{SVW2011spectral,SVW2011dispersion,FS2012Axial,Calc2012DiffusionMultiFractional}.

It has been noted \cite{Car95} that, because in dimension $<3$ the Weyl curvature vanishes, the vacuum Einstein equation has only solutions of constant curvature (flat, if the cosmological constant is $0$). Hence, local degrees of freedom are absent, that is, gravitational waves in the classical case, and gravitons in QG, are absent. As we have seen in section \sref{s_wch}, {\quasireg} singularities have vanishing Weyl curvature (see \sref{s_dimensional_reduction_singularities_weyl}).

\subsubsection{Asymptotic safety for the renormalization group}

Even though GR appears to be {\nonrenormalizable}, the renormalization group analysis may provide useful indications. The idea of \textit{asymptotic safety}, proposed by S. Weinberg  \cite{Wein79AS}, is that, although there is an infinite number of coupling constants, they become ``unified'', under the renormalization group flow, when approaching an ultraviolet fixed point. That is, they converge to a finite-dimensional \textit{ultraviolet critical surface}. To make the $3+1$-dimensional gravity asymptotically safe, one can supplement the Einstein-Hilbert Lagrangian density $\mathcal L=-\dsfrac{1}{16\pi \mathcal G}R\sqrt{-\det g}$ with higher-order curvature terms. For example, by replacing it with $\mathcal L=f(R)\sqrt{-\det g}$. Weinberg observed that ``there is an asymptotically safe theory of pure gravity in $2 + \varepsilon$ dimensions, with a one-dimensional critical surface'', and ``[a]symptotic safety is also preserved when we add matter fields'', if we add certain compensatory fields \cite{Wein79AS}.

The evidence for asymptotic safety \cite{RS01rgqg,Lit04fixed,Nied07asqg,HW05qg,Reu07frg,CAR09qg}, involves especially two dimensions \cite{KKN96ren2dim,Lit06fixed}. The spectral dimension near the fixed point also is $d_S=2$ \cite{LR05fractal}. Other indications are that the existence of a non-Gaussian fixed point for the dimensionless coupling constant $g_N={\mathcal G}_N\mu^{d-2}$ requires two-dimensionality \cite{Nied07asqg}.

\subsubsection{{\HL} gravity}
\label{s_dim_red_horava}

{\HL} gravity, inspired by the quantum critical phenomena in condensed matter systems, was proposed in 2009 by {\hor} \cite{Hor09qglp}. The idea is to assume that space and time behave differently at scaling:
\begin{equation}
\label{eq_horava_anisotropy}
\begin{array}{l}
\left\{
\begin{array}{ll}
	\mathbf{x}&\mapsto b\mathbf{x}, \\
	t&\mapsto b^z t. \\
\end{array}
\right.
\\
\end{array}
\end{equation}
To describe an UV fixed point, it is required that $z=D-1=3$, and sometimes that $z=4$.
The anisotropy \eqref{eq_horava_anisotropy} is required to be a symmetry of the solution, but not necessarily of the action itself. In the UV limit, the theory describes interacting non-relativistic gravitons. In $1+3$ dimensions, it is power-counting renormalizable.

Lorentz symmetry is broken in the UV limit, but it is conjectured that it is restored at large distances, and $z\to 1$. The theory leads to field equations which are of high order in space, to cancel the divergences of the loop integrals, and of second-order in time, to avoid ghosts. The spectral dimension is $4$ at low energies, and $2$ at high energies \cite{Hor09spectral}.

The volume element has the dimension
\begin{equation}
	[\de t\de^d\mathbf{x}] = -d-z,
\end{equation}
which is incremented by each time derivative by 
\begin{equation}
	[\partial t] = z,
\end{equation}
hence, $\kappa$ is dimensionless when $z=d$, since its scaling dimension is
\begin{equation}
	[\kappa] = \dsfrac{z-d}{2}.
\end{equation}

Some objections have been raised, most of them objections coming from the difficulty to prove that GR is recovered in the IR limit \cite{CNPS09HL,WSV10HL,Sot11HL,Vis11HL,BPS09HL,KP10HL,HKG10HL,PS10HL,BPS10HL,WW11HL}.

In section \sref{s_dimensional_reduction_singularities_anisotropy} we will show how our singularities suggest a mechanism for the anisotropy \eqref{eq_horava_anisotropy}.

\section{Dimensional reduction at singularities}
\label{s_dimensional_reduction_singularities}

While dimensional reduction seems to be needed for Quantum Gravity, little is known about what causes it. But the singularities studied in this Thesis may provide an explanation, because they act like sources for the dimensional reduction.

\subsection{The dimension of the metric tensor}
\label{s_dimensional_reduction_singularities_metric}

At points $p\in M$ where the metric tensor $g_{ab}$ is degenerate, the distance vanishes along the degenerate directions, which belong to the space  $\ker(\flat):=T^\perp_p(M)\leq T_pM$ (fig. \ref{degenerate-metric}). These directions can be factored out by taking the quotient
\begin{equation}
\label{eq_factorize}
	\coannih{T_p}M:=\dsfrac{T_pM}{T^\perp_p(M)},
\end{equation}
which has the dimension equal to the rank of the metric at $p$:
\begin{equation}
	\dim{\coannih{T_p}M}=\rank g_p.
\end{equation}

Because the metric $g$ is degenerate, it can be reduced to a metric tensor on the space $\coannih{T_p}M$, and its inverse is a metric tensor on $\annih{T_p}M$, both spaces having dimension $\rank g_p$.

\subsection{Metric dimension {\vs} topological dimension}
\label{s_dimensional_reduction_singularities_metric_vs_topologic}

Let $(M,g)$ be a singular {\semiriem} manifold. Even if the rank of the metric is, at the points where the metric becomes degenerate, lower than $\dim M$, the topological dimension of the manifold remains $\dim M$.
Mathematically, there are three layers in Differential Geometry: the \textit{topological structure}, the \textit{differential structure}, and the \textit{geometric structure}. The  dimension of the vector space used in the charts of the atlas is the topological dimension of $M$. The \textit{metric dimension}, or the \textit{geometric dimension} is the rank of the metric, and can be at most equal to the topological dimension.

Kupeli showed that, if $(M,g)$ is a {\rstationary} {\semiriem} manifold of constant signature $(k,l,m)$, it is locally isomorphic to a direct product manifold $P\times_0 N$ between a (\nondeg) {\semiriem} manifold $P$ of signature $(l,m)$, and a $k$-dimensional manifold $N$, whose metric can be considered equal to $0$ \cite{Kup87b}. Therefore, from the viewpoint of geometry, at any point $p\in M$ the degenerate directions can be ignored, and we can identify the $D$-dimensional manifold $M$ at $p$, with the $\rank g_p$-dimensional manifold $P$. The information contained in the metric $g$ of the manifold $M$ can be obtained by pull-back, from that of a metric on a manifold $(P,h)$ with smaller topological dimension.

In addition to the metric, in GR there are present other fields. To admit smooth metric contractions between covariant indices, and to admit smooth covariant derivatives (Chapter \ref{ch_ssr}), the fields have to behave as if they don't ``see'' the degenerate directions.

This connects our results with the topological dimensional reduction of D.V. Shirkov and P. Fiziev (see section \sref{s_dim_red_topological}). Our approach confirm their observation \cite{FS2012Axial}:
\begin{quote}
dimensional reduction of the physical space in general relativity (GR) can be regarded as an unrealized
and as yet untapped consequence of Einstein's equations (EEqs) themselves which takes place around singular points
of their solutions.
\end{quote}

\subsection{Metric dimensional reduction and the Weyl tensor}
\label{s_dimensional_reduction_singularities_weyl}

According to Theorem \ref{thm_wch}, in dimension $4$, the Weyl curvature vanishes at {\quasireg} singularities, because at the points where the metric is degenerate, $\dim{\annih{T_p}M}\leq 3$. Because the Ricci decomposition is smooth, the Weyl curvature tensor $C_{abcd}$ is smooth, and remains small around the singularity, where it vanishes. 

Examples of {\quasireg} singularities include the {\schw} black hole (section \sref{s_black_hole_schw}), which can represent a classical neutral spinless particle. It is not yet clear  whether the singularities of stationary charged and rotating black holes are {\quasireg}, but we know they are analytic. Such black holes undergo geometric dimensional reduction at singularities. For the charged ones, the electromagnetic potential and its field are analytic, and finite at $r=0$ (see sections \sref{s_black_hole_rn} and \sref{s_black_hole_kn}).

At singularities, because the Weyl curvature tensor vanishes, when quantizing, gravitons are absent, making gravity renormalizable \cite{Car95}.

\subsection{Lorentz invariance and metric dimensional reduction}
\label{s_dimensional_reduction_singularities_lorentz_invariance}

At the points where the metric is degenerate, Lorentz invariance is affected. The metric is, at those points, not invariant to a group of transformations isomorphic to Lorentz, group. The role of the Lorentz group is taken by a larger group -- the \textit{Barbilian group} \cite{Barb39}. If we factorize (at a point as in \eqref{eq_factorize}, or on a region of constant signature, if the metric is {\rstationary}) we reduce the Barbilian group to a subgroup, which is an orthogonal subgroup of the Lorentz group.

Hence, in our approach, Lorentz invariance is maintained at all points where the metric is {\nondeg}, and the theory is equivalent with Einstein's. At the points where the metric is degenerate, no theory can maintain Lorentz invariance. The best we can do at singularities is to maintain the invariance with respect to the subgroup of the Lorentz group, obtained after the metric dimensional reduction.

Other approaches to Quantum Gravity, like Loop Quantum Gravity and {\HL} gravity, gave up Lorentz invariance even at the points where the metric is {\nondeg}, and hope that it emerges at large scales.

\subsection{Particles lose two dimensions}
\label{s_dimensional_reduction_singularities_charged}

From the viewpoint of GR, charged particles are charged black holes. As viewed in Chapter \ref{ch_black_hole}, the metric of a charged black hole of type {\rn} or {\kn} can be made analytic. The electromagnetic potential and its field are also analytic, and finite at $r=0$. This can be easily extended to Yang-Mills fields. 

Equation \eqref{eq_rn_ext_ext} shows that, in our coordinates,  at $\rho=0$ (which is equivalent to $r=0$), the metric loses two dimensions, $\de\rho$ and $\de\tau$. The spherical metric seems to be canceled by the factor $\rho^{2\CS}$ but this is only apparent, being due to the warped product -- similar to the spherical coordinates for $\R^3$, when we know that multiplying with $r^2$ doesn't make the metric degenerate at $r=0$.


Metric dimensional reduction is present also in the {\schw} and {\kn} models. So, we see that dimensional reduction takes place for particles which are charged or neutral, spinning or not.
The dimensional reduction, and the finiteness of the gauge potential and fields at the singularity. Future developments of our research program should check how these features affect the perturbative expansions and the renormalization group analysis, including for the gravitational field, taking into account the new results.

\subsection{Particles and spacetime anisotropy}
\label{s_dimensional_reduction_singularities_anisotropy}

For the metric \eqref{eq_rn_ext_ext} to admit a foliation in spacelike hypersurfaces, the condition $\CT\geq 3\CS$ is required in \eqref{eq_coordinate_ext_ext} \cite{Sto11f}. This allows the space and time to be distinguished at $\rho=0$, with the price of an anisotropy between space and time. The anisotropy is manifest when going back to the old {\rn} coordinates $(t,r)$ --  a rescaling which is isotropic in the coordinates $(\tau,\rho)$, is anisotropic in the coordinates $(t,r)$. If our coordinates $(\tau,\rho)$ are the correct ones, diffeomorphism invariance should be considered valid in these coordinates, and not in the singular coordinates $(t,r)$.

P. {\hor} showed that an anisotropy like \eqref{eq_horava_anisotropy} leads to the correct dimension for the Newton constant (see \sref{s_dim_red_horava} and \cite{Hor09qglp}). 
Could the anisotropy we obtained at the {\rn} singularity, be related to the one proposed by {\hor}? They are different in nature: our anisotropy follows from our coordinates, and GR, while {\hor} modified GR.
We don't need to impose the recovery of standard GR at low energies, since our approach is already equivalent to GR everywhere outside the singularities.

\subsection{The measure in the action integral}
\label{s_dimensional_reduction_singularities_measure}

As approaching the degenerate singularities, the volume form
\begin{equation}
\label{eq_volume}
\vol := \sqrt{-\det g}\de x^0\wedge\ldots\wedge \de x^{D-1} = \sqrt{-\det g}\de x^D
\end{equation}
tends to $0$. In the case of black hole singularities, this is due to the fact that in our coordinates the metric is analytic.
The volume form appears in the action principle
\begin{equation}
\label{eq_action}
S=\int_\mc M\vol(x)\, \mc L
\end{equation}
If, in some coordinates $(x^\mu)$, the metric is diagonal, then
\begin{equation}
\label{eq_volume_diag}
\vol(x) = \prod_{\mu=0}^{D-1}\sqrt{\abs{g_{\mu\mu}(x)}}\de x^D.
\end{equation}
We see that we obtain a measure of the type proposed by G. Calcagni \cite{Calc2010FractalQFT,Calc2010FractalUniverse}. To see this, we take in the Lebesgue measure from section \sref{s_dim_red_fractal}, equation \eqref{eq_stmes}, the weights
\begin{equation}
\label{eq_stme_sgr_weights}
f_{(\mu)}(x) = \sqrt{\abs{g_{\mu\mu}(x)}}.
\end{equation}
Hence, much of the analysis developed in \cite{Calc2010FractalQFT,Calc2010FractalUniverse} for QFT follows from the fact that the metric becomes degenerate.
The measure becomes 
\begin{equation}
\label{eq_stme_sgr}
\de \varrho(x) = \sqrt{-\det g}\de x^D,
\end{equation}
the standard measure from General Relativity, which we now allow to vanish.

Calcagni justifies the change of the measure by considering the spacetime to be fractal, {\ie} having Hausdorff dimension which changes with the scale. The Lebesgue-Stieltjes measure he postulates is independent of GR, and he applies it both to flat spacetime QFT, and to QG \cite{Calc2010FractalQFT,Calc2010FractalUniverse,Calc2011FractalGravity}.
He applies his action \eqref{eq_stmes} to Special Relativistic QFT. On the other hand, the identification we proposed in \eqref{eq_stme_sgr_weights} ises in place of the weights $f_{\mu}$, the square root of the components of a diagonal metric, which we allowed to become zero. This suggests that, at small scales, QFT is improved by GR, by dimensional reduction.

Calcagni applies to GR the same recipe he used for Special Relativistic QFT, but multiplies the weights with $\sqrt{-\det g}$ \cite{Calc2010FractalQFT}. His approach leads therefore to a modified Einstein equation, and modified GR.

On the other hand, our approach keeps the standard GR everywhere where the metric is not degenerate.
The measure is similar to Calcagni's measure which tames the infinities of QFT, but it is just given by the volume form of a metric which becomes degenerate.

\subsection{How can dimension vary with scale?}
\label{s_dimensional_reduction_singularities_scale}

We have seen that the singularities with smooth but degenerate metric exhibit dimensional reduction, and showed how it connects with dimensional reduction from other approaches. It seems that our approach confirmed the intuition of some of these versions of dimensional reduction. What is more important, degenerate metric follows naturally from GR, makes both big-bang and black hole singularities behave well, and it was not invented ad-hoc to provide the needed dimensional reduction.

Our geometric dimensional reduction manifests as we approach a singularity. But what is needed in QFT and QG, is that the dimensional reduction depends on the scale.
We will give only a qualitative justification, leaving the quantitative proof for future research.

We will assume that it is enough to sum over the trivial topology, and there's no need to include non-trivial topologies. Because the topology of the black holes solutions which can be used to describe particles can be taken the same as for the Minkowski spacetime \cite{Sto11e,Sto11f,Sto11g,Sto12e}, probably it will be enough to work with the same topology, and not sum over all possible ones\footnote{Also in {\HL} gravity and CDT is not summed over other topologies.}.

If particles are to be represented by black hole singularities with degenerate metric, then, if the number of particles are present in a given volume grows, the integral of the volume form in a local coordinate system of a region containing them decreases (fig. \ref{dim-red-feynman}). That is, the average volume element decreases, as required for example in Calcagni's approach.
The Feynman diagrams of higher order involve a larger number of particles in the same volume. This leads to smaller volume element, hence to lower contribution in the integrals.

\image{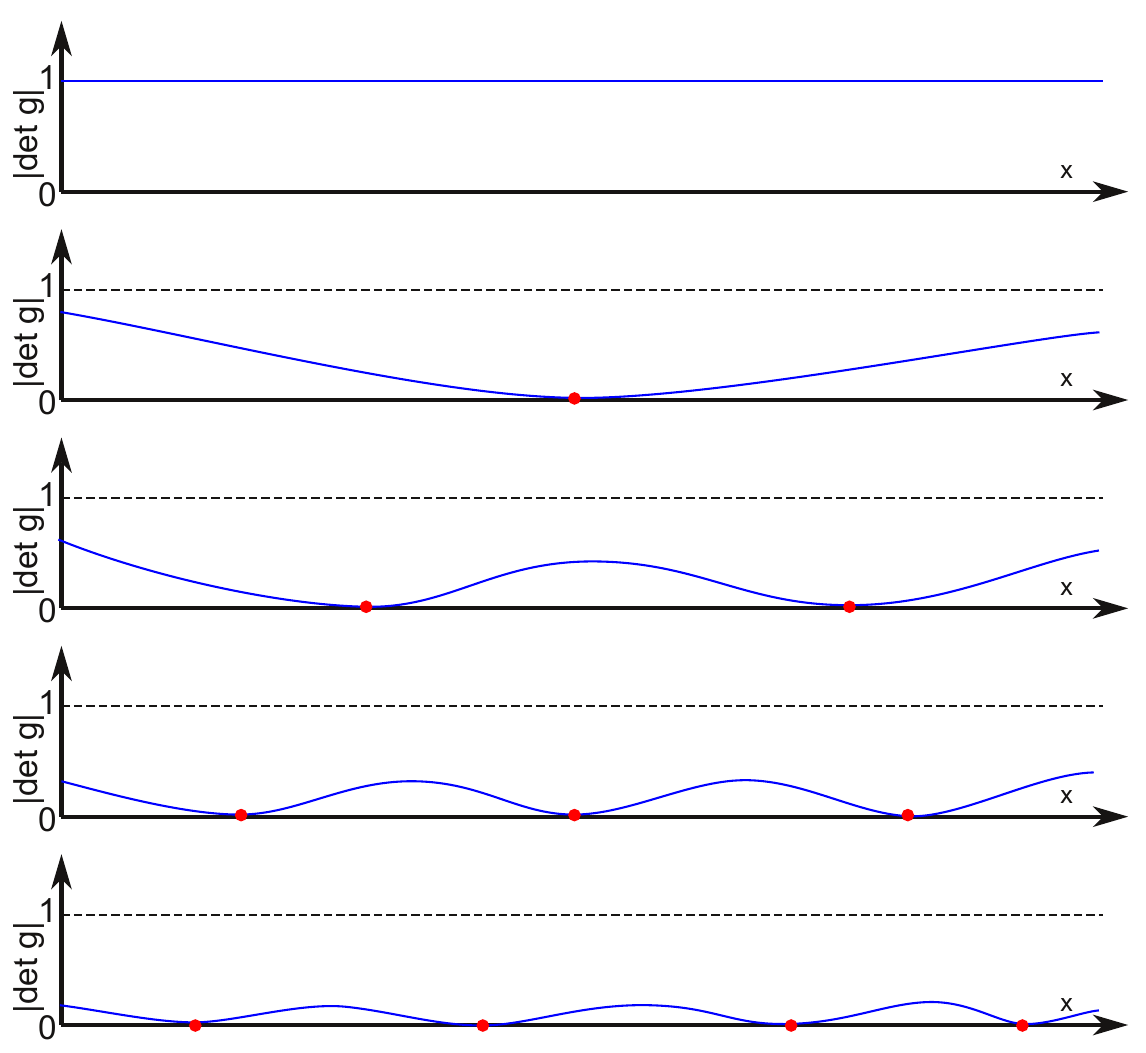}{0.7}{Artistic representation of the principle that the metric's average volume element decreases as the number of singularities ({\ie} particles) in the region increases. Singularities are represented as red dots, and $|\det g|$ is represented by blue curves.}

While this argument refers to local coordinates, it is invariant, in the sense that the constraint $\det g=0$ at singularities is the same in all coordinates obtained by a nonsingular transformation. Hence, the diffeomorphism invariance (general covariance) is maintained in our argument.

Similar considerations can start with the fact that $C_{abcd}=0$ for {\quasireg} singularities, and therefore, in average, as the number of particles in the Feynman diagrams increase, the Weyl curvature becomes smaller, and the local degrees of freedom like gravitons vanish with the energy.

Hence, our main conjecture, which we name \textit{the hypothesis of average dimensional reduction}, is that the needed dimensional reduction which appears at smaller scales, is due to the dimensional reduction at the singularities representing the particles.

\section{Conclusions}

We reviewed several results indicating that some problems concerning field quantization, including of gravity, would be resolved if one sort or another of dimensional reduction takes place at small scales. We showed that the theory of singularities developed in this Thesis leads to a (geo)metric dimensional reduction, manifested in many ways similarly to dimensional reduction from other approaches, ours having the bonus that General Relativity is not modified. Therefore, our approach to the problem of singularities opens new possibilities to be explored in Quantum Gravity.

\imagerot{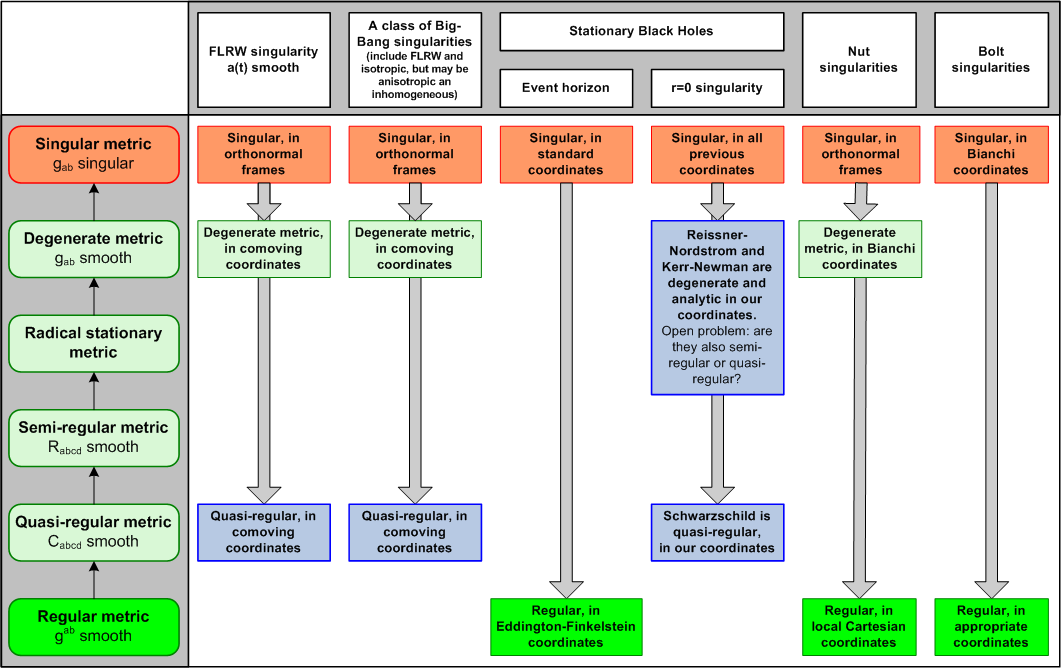}{1.6}{Several examples of singularities, how they fit in our scheme, and how they are resolved.}{270}



\addtocontents{toc}{\vspace{2em}} 

\appendix 



\addtocontents{toc}{\vspace{2em}} 

\backmatter


\label{Bibliography}

\lhead{\emph{Bibliography}} 

\bibliographystyle{abbrvnat} 


\pagebreak
\addtotoc{My Papers}
\section*{My Papers}
\label{MyBibliography}
\lhead{\emph{My Papers}} 

\begingroup
\renewcommand\bibsection{}

\endgroup

\end{document}